\documentclass{article}
\usepackage{amsfonts}
\usepackage{amssymb}
\usepackage{amsmath}
\usepackage{amsthm}
\usepackage{euscript}
\usepackage{hyperref}
\hypersetup{pdftex,colorlinks=true}
\usepackage{graphicx}

\newcommand{\blue}{}

\def\[#1\]{\begin{equation*}#1\end{equation*}}
\makeatletter
\def\beq{%
   \relax\ifmmode
      \@badmath
   \else
      \ifvmode
         \nointerlineskip
         \makebox[.6\linewidth]%
      \fi
      $$
   \fi
}
\def\eeq{%
   \relax\ifmmode
      \ifinner
         \@badmath
      \else
         $$
      \fi
   \else
      \@badmath
   \fi
   \ignorespaces
}

\def\enddisplaymath{\eeq\global\@ignoretrue}
\makeatother

\newtheorem{theorem}{Theorem}[section]
\newtheorem{proposition}[theorem]{Proposition}
\newtheorem{lemma}[theorem]{Lemma}

\theoremstyle{remark}
\newtheorem*{Remark}{Remark}

   \title{$q$-Distributions on boxed plane partitions}
   \author{Alexei Borodin, Vadim Gorin, Eric M.~Rains}

\begin{document}
\maketitle

\begin{abstract}
We introduce elliptic weights of boxed {plane} partitions and prove that they give rise to a
generalization of MacMahon's product formula for the number of plane partitions in a box. We then
focus on the most general positive degenerations of these weights that are related to orthogonal
polynomials; they form three two-dimensional families. For distributions from these families we
prove two types of results.

First, we construct explicit Markov chains that preserve these
distributions. In particular, this leads to a relatively simple
exact sampling algorithm.

Second, we consider a limit when all dimensions of the box grow and
plane partitions become large, and prove that the local correlations
converge to those of ergodic translation invariant Gibbs measures.
For fixed proportions of the box, the slopes of the limiting Gibbs
measures (that can also be viewed as slopes of tangent planes to the
hypothetical limit shape) are encoded by a single quadratic
polynomial.

\end{abstract}

\tableofcontents

\section{Introduction} The uniform distribution on boxed plane
partitions (equivalently, lozenge tilings of a hexagon) is one of
the most studied models of random surfaces. There are four principal
types of results regarding this model that have been proved.

(1) Law of large numbers: Under the global scaling (the bounding
box/he\-xa\-gon is fixed and the mesh is going to zero), the measure
concentrates on surfaces that are close to a certain deterministic
\emph{limit shape}. The limit shape can be obtained as the unique
solution to a suitable variational problem. The solution is encoded
by a second degree polynomial in two variables, see \cite{CKP},
\cite{CLP}, \cite{DMB}, \cite{Des}, \cite{KO}.

(2) Locally near any point of the limit shape, the measure on
tilings converges to a (uniquely defined, see \cite{Sh})
translation-invariant ergodic Gibbs measure on lozenge tiling of the
plane of a given slope, and the slope coincides with the slope of
the tangent plane to the limit shape at the chosen point, see
\cite{Gor} and also \cite{Ke1}, \cite{Ke2}, \cite{KO}, \cite{KOS}.

(3) The correlation kernel of the random point process of lozenges
of one of the types is explicitly expressed in terms of classical
Hahn orthogonal polynomials, see \cite{Gor},
\cite{J_nonintersecting}, \cite{J_Hahn}, \cite{JN}.

(4) A few algorithms, both asymptotic and exact, have been proposed
to generate the random tilings in question, see \cite{BG},
\cite{Kr}, \cite{P1}, \cite{P2}, \cite{W1}, \cite{W2}.

These are complemented by the classical MacMahon product formula
for the total number of plane partitions in a given box, see, e.g.,
Section 7.21 in \cite{St}.

In this paper we study measures on boxed plane partitions that
generalize the uniform distribution. The weight of a tiling is
defined as the product of certain simple factors over all
lozenges of a fixed type, see Section \ref{Section_pr_models} for
definitions. One special case is the weight $q^{volume}$, where
\emph{volume} is the volume of the corresponding plane partition,
and $q$ is an arbitrary positive number.

In the most general case we consider, the weight of a lozenge is
\emph{elliptic}. Our initial motivation came from the fact that
these weights lead to a nice generalization of the MacMahon formula
mentioned above, see Theorem \ref{th:macmahon} in the Appendix.

For the asymptotic analysis, we look at the top degeneration of the
elliptic weight that is related to orthogonal polynomials.

Our asymptotic results amount to proving analogs of (2), (3),
and (4) above. As for (1), we derive the corresponding variational
problem (which differs from the one for the uniform case by the
presence of an external potential), and show that the hypothetical
limit shape (obtained from (2)) solves the corresponding Euler-Lagrange
equation. However, we do not prove the concentration phenomenon
rigorously.

One interesting feature of the limit shapes that arise is that the curve
that bounds the frozen regions may have one or two nodes in vertices of the
hexagon, see Section \ref{Section_comp_simulations} for illustrations.

In terms of orthogonal polynomials, our models go all the way up to
the top of the $q$-Askey scheme --- the most general models we analyze
asymptotically are related to the q-Racah classical orthogonal
polynomials. We also show that the elliptic weights lead to
the biorthogonal functions constructed in \cite{SpiridonovVP/ZhedanovAS:2000b}.
We hope to return to the asymptotic analysis of this case in a later
publication.

Our proof of (2) follows the same steps as the argument for the
uniform case in \cite{Gor}. It is based on the general method of
computing limits of correlation kernels suggested
in \cite{BO} and \cite{O}. The crucial property we need is the
second order difference equation satisfied by the q-Racah orthogonal
polynomials.

Our perfect sampling algorithm is obtained from a more general
construction of relatively simple Markov chains that change the size
of the box (one side increases by 1 and another side decreases by
1), and that map the measures from the class we consider to similar
ones. The construction follows the approach of \cite{BF}; the key
facts that make that approach possible reduce to certain recurrence
relations for the q-Racah polynomials.

A computer simulation of the above-mentioned Markov chains can be
found at
\texttt{http://www.math.caltech.edu/papers/Borodin-Gorin-Rains.exe}.

\smallskip

\noindent\textbf{Acknowledgements}. AB was partially supported by NSF grant DMS-0707163. \blue{VG
was partially supported by the Moebius Contest Foundation for Young Scientists.} EMR was partially
supported by NSF grant DMS-0833464.

\section{Model and results}

\label{section_Model_and_results}

\subsection{Combinatorial interpretations}

For any integers $a,b,c\ge 1$ consider a hexagon with sides
$a,b,c,a,b,c$ drawn on the regular triangular lattice. Denote by
$\Omega_{a\times b\times c}$ the set of all tilings of this hexagon
by rhombi obtained by gluing two of the neighboring elementary
triangles together (such rhombi are called {\it lozenges\/}).
 An element of $\Omega_{3\times 3\times 3}$ is
shown in Figure 1.

\begin{center}
 \scalebox{0.5}{\includegraphics{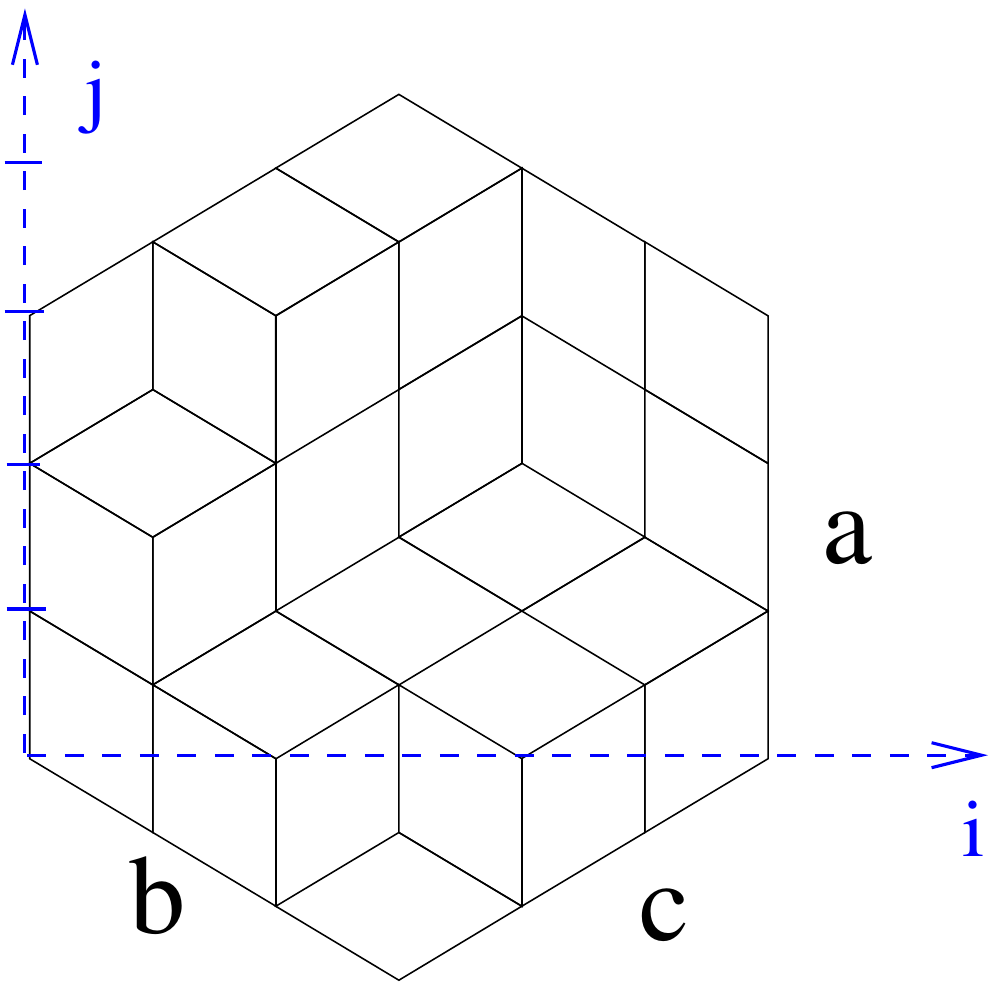}} \hspace{1cm}
  \scalebox{0.75}{\includegraphics{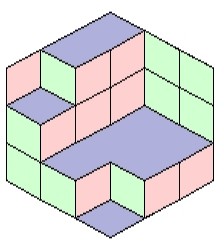}}

 Figure 1. Tiling of a $3\times 3\times 3$ hexagon.
\end{center}

Lozenge tilings of a hexagon can be identified with 3-D Young
diagrams (equivalently, boxed plane partitions) or with stepped
surfaces. The bijection is best described pictorially. We show a 3-D
shape corresponding to a tiling in Figure 1.

It is convenient for us to slightly modify both hexagon and lozenges
by means of a simple affine transform of the plane.

\begin{center}
 {\scalebox{0.5}{\includegraphics{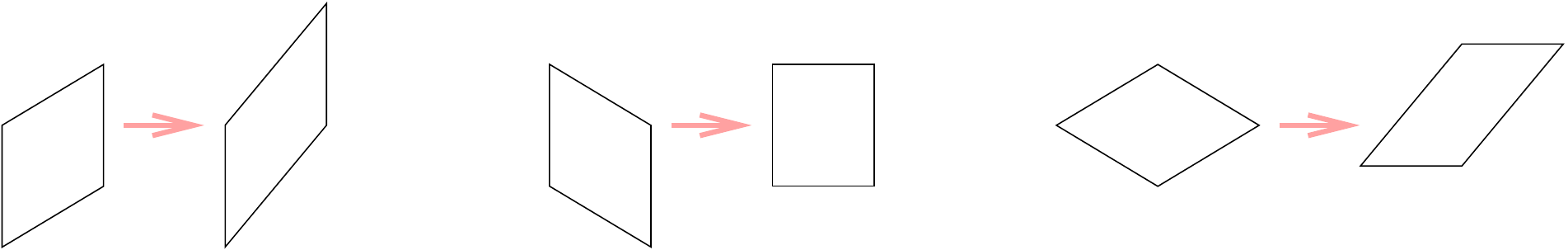}}}

 Figure 2. Affine modification of lozenges
\end{center}

We thus obtain a tiling of a slightly different hexagon.

\begin{center}
 {\scalebox{0.5}{\includegraphics{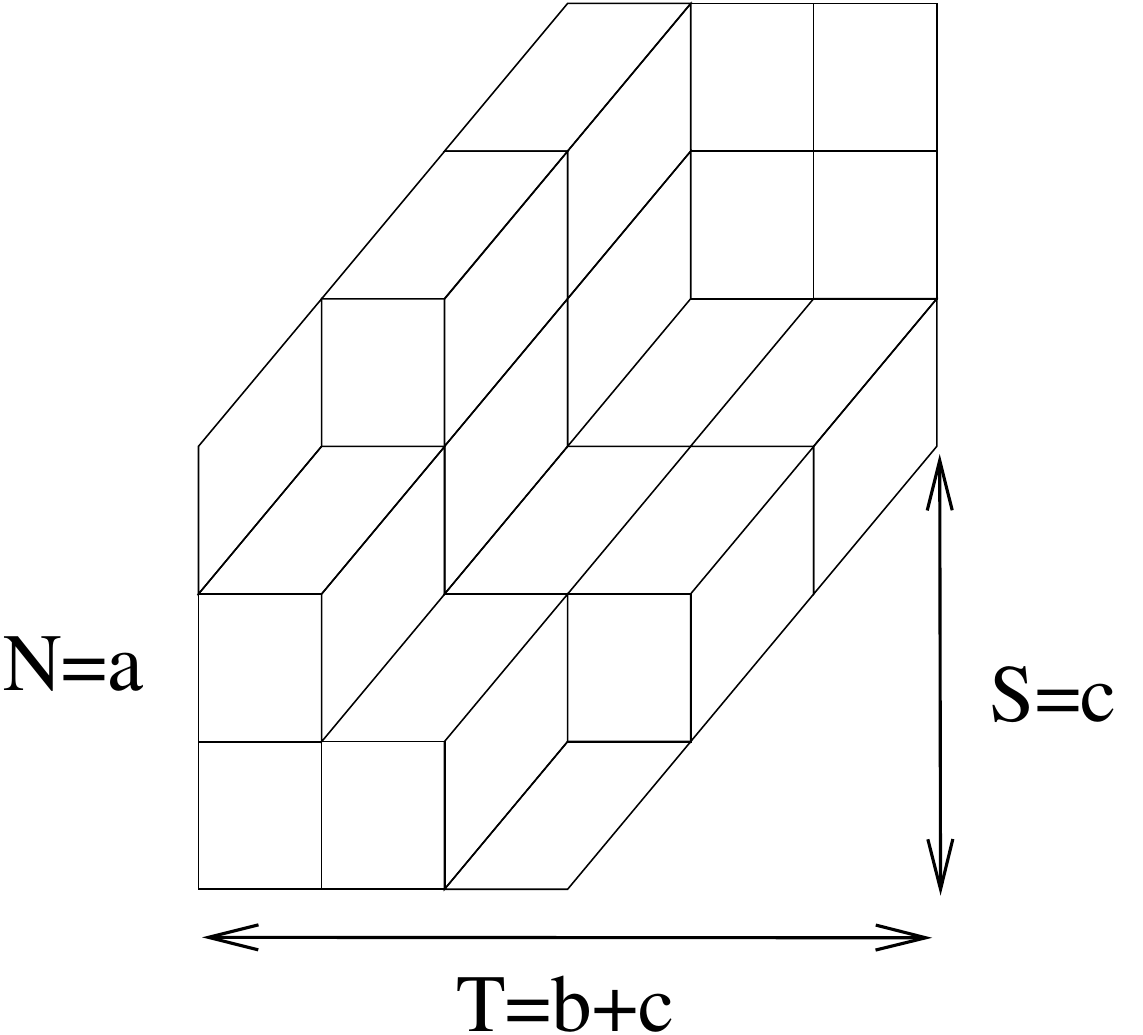}}} {\scalebox{0.5}{\includegraphics{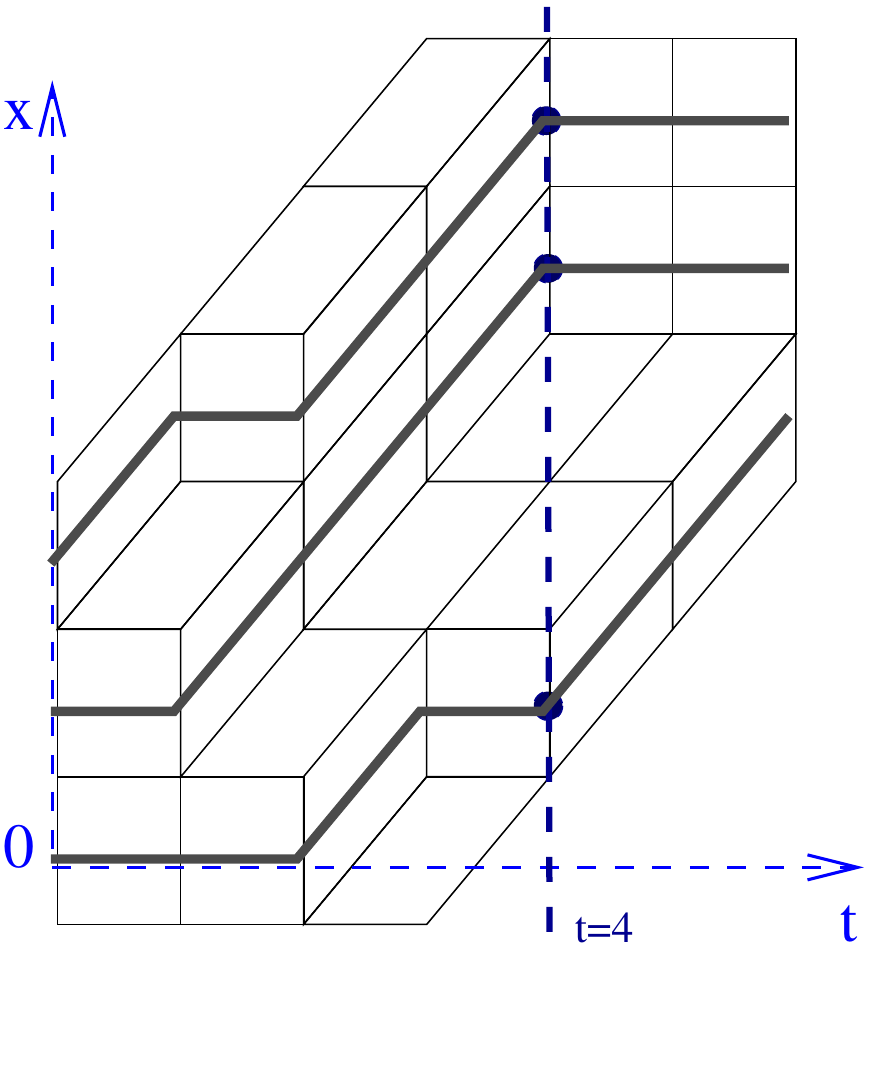}}}

 Figure 3. Modified tiling of a $3\times 3\times 3$ hexagon and the corresponding family of non-intersecting paths.
\end{center}

In what follows we use different parameters instead of $a$, $b$,
$c$. Set $N=a$, $T=b+c$, $S=c$. We will also denote the set
$\Omega_{a\times b\times c}$ by $\Omega(N,T,S)$.

Each tiling corresponds to a family of nonintersecting paths as shown in Figure 3.

Consider a section of our family of paths by a vertical line
$t=t_0$. Clearly, this gives an $N$-tuple of points in $\mathbb Z$.
Thus, our tiling can be viewed as an $N$-point configuration
varying in time $t=0,1,\dots, T$. Note that when $t=0$ the
configuration consists of points $\{0,1,\dots,N-1\}$, while for
$t=T$ the configuration consists of points $\{S,\dots,S+N-1\}$.

\subsection{Probability models}
\label{Section_pr_models}

Let us introduce the probability measures on $\Omega(N,T,S)$ that are
studied in this paper. For any ${\mathcal T} \in \Omega(N,T,S)$, we
define its weight $w({\mathcal T})$ and consider the probability
distribution on $\Omega(N,T,S)$ given by the formula
$$
 {\rm Prob}\{{\mathcal T}\}=\frac{w({\mathcal T})}{\sum\limits_{{\mathcal T}'\in \Omega(N,T,S)}w({\mathcal T'})}
$$

The weights we consider are such that the probability of a tiling is
proportional to the product of certain weights corresponding to the
horizontal lozenges in it, i.e.,
$$
 w({\mathcal T})=\prod\limits_{{\scalebox{0.2}{\includegraphics{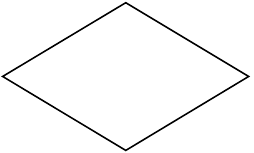}}}\in {\mathcal
 T}}{w({\scalebox{0.2}{\includegraphics{hor_lozenge.pdf}}})}
$$

Note that the number of horizontal lozenges in a tiling of a given
hexagon is fixed (i.e., it does not depend on the tiling). Hence,
multiplying $w({\scalebox{0.2}{\includegraphics{hor_lozenge.pdf}}})$
by a nonzero constant does not change the probability distribution.

In the most general case considered in Section \ref{sc:Appendix},
$w({\scalebox{0.2}{\includegraphics{hor_lozenge.pdf}}})$ is an
elliptic weight given by

\begin{equation}
\label{elliptic_weight}
 w({\scalebox{0.2}{\includegraphics{hor_lozenge.pdf}}})=
\frac{(u_1u_2)^{1/2}q^{j-1/2}\theta_p(q^{2j-1}u_1u_2)}{\theta_p(q^{j-3i/2-1}u_1,q^{j-3i/2}u_1,q^{j+3i/2-1}u_2,q^{j+3i/2}u_2)},
\end{equation}
where the coordinates of the topmost point of
${\scalebox{0.2}{\includegraphics{hor_lozenge.pdf}}}$ are $(i,j)$
(the $i$ and $j$ axes are pictured in Figure 1), $u_1$, $u_2$, $p$,
$q$ are (generally speaking, complex) parameters,
$$
\theta_p(x)=\prod_{i=0}^{\infty}(1-p^ix)(1-p^{i+1}/x)
$$
and $\theta_p(a,b,c\dots)=\theta_p(a)\theta_p(b)\theta_p(c)\dots$

Mostly we will not be concerned with this most general case, nor with the most general
trigonometric case obtained by taking $p\to 0$ (note $\theta_0(x)=1-x$), as in these cases the
kernels involve biorthogonal functions.  The most general orthogonal polynomial case is the limit
\blue{
\begin{equation}
\label{eq_limit_transition}
 p\to 0,\quad u_1=O(\sqrt p),\quad u_2=O(\sqrt p),\quad u_1u_2=p \kappa^2
 q^{-S},
\end{equation}}
in which case the weight function is \blue{
\begin{equation}
\label{qRacah_weight}
 w({\scalebox{0.2}{\includegraphics{hor_lozenge.pdf}}})= \kappa q^{j-(S+1)/2}-\frac1{\kappa
 q^{j-(S+1)/2}};
\end{equation}}
this is also the most general case in which the weight is
independent of $i$.

 Clearly, the factor \blue{$q^{-(S+1)/2}$} can be removed if
we replace $\kappa $ by \blue{$\kappa '=\kappa \cdot q^{-(S+1)/2}$.}
However, this choice is more convenient for our further
considerations.

We need to make sure that the probabilities of tilings are positive.
This leads to certain restrictions on the parameters. There are three
possible cases:

\begin{enumerate}
\item \emph{imaginary $q$-Racah case:} $q$ is a positive real number, $\kappa $ is an arbitrary pure imaginary complex
number;
\item \emph{real $q$-Racah case:} $q$ is a positive real number, $\kappa $ is a real number with
additional restrictions depending on the size of a hexagon; $\kappa$
cannot lie inside the interval \blue{$[q^{-N+1/2},q^{(T-1)/2}]$ or
$[q^{(T-1)/2},q^{-N+1/2}]$,} depending on whether $q>1$ or $q<1$;

\item \emph{trigonometric $q$-Racah case:} $q$ and $\kappa$ are complex numbers
of modulus 1, $q=e^{i\alpha}$, $\kappa =e^{i\beta}$, plus additional
restrictions on $\kappa $ depending on the size of a hexagon: \blue{
both $-\alpha (T-1)/2+\beta$ and $\alpha (N-1/2) +\beta$ } must lie
in the same interval of the form $[\pi k, \pi (k+1)]$, $k\in\mathbb
Z$. In this case \blue{
$$\kappa q^{j-(S+1)/2}-\frac1{\kappa q^{j-(S+1)/2}}=2\sqrt{-1}\sin\left(\alpha(j-(S+1)/2)+\beta\right),$$}
and the factor $2\sqrt{-1}$ here can be omitted.
\end{enumerate}

The names of the cases are related to those of the classical
orthogonal polynomials that appear in the analysis.

Denote the resulting measure on $\Omega(N,T,S)$ by
$\mu(N,T,S,q,\kappa )$.

There are further limit transitions.

If we send $\kappa \to 0$ then we get the \emph{$q$-Hahn case}
$$
 w({\scalebox{0.2}{\includegraphics{hor_lozenge.pdf}}})=q^{-j}.
$$
Thus, the probability of the plane partition of volume (=number of
$1\times1\times1$ boxes) $V$ is proportional to $q^{-V}$. On the
other hand, if we send $\kappa \to\infty$ we will get the weights
$q^V$. Therefore, the case of general $\kappa $ can be viewed as an
interpolation between the measures $q^{volume}$ and $q^{-volume}$.

Another possibility is to set $\kappa =q^{K}$ and then send $q\to 1$. The weight of a horizontal
lozenge tends to \blue{
$$
 w({\scalebox{0.2}{\includegraphics{hor_lozenge.pdf}}})={K+j-(S+1)/2}.
$$
} We call this case the \emph{Racah case}. One has to impose
restrictions to ensure positivity: $K$ cannot lie inside the
interval $[-N+1/2,(T-1)/2]$.

Finally, if we either send $\kappa \to 0$ and set $q=1$, or send
$q\to 1$ and then send $ K \to\infty$, then we obtain the \emph{Hahn
case}
$$
w({\scalebox{0.2}{\includegraphics{hor_lozenge.pdf}}})=1
$$
In this case our probability distribution on $\Omega(N,T,S)$ is
uniform.

\smallskip

Below we mostly work with the imaginary $q$-Racah case, but all the
results can be carried over to all the other cases mentioned above
by the appropriate substitutions of parameters and degenerations.

\subsection{Representation of the weight}

Let us view our tiling as a pile of $1\times 1\times 1$ cubes  in
the box located between the planes $x_1=0$, $x_1=a$, $x_2=0$,
$x_2=b$, $x_3=0$, $x_3=c$. Then Figure 1 represents a projection of
the border of the 3-D-diagram to the plane $x_1+x_2+x_3=0$ parallel
to the vector $(1,1,1)$.

For any $v\in\mathbb{R}^3$, denote by $h(v)$ the Euclidian distance
from $v$ to the plane $x_1+x_2+x_3=0$ divided by $\sqrt{3}$, and
denote by $\hat h(v)$ the distance from $v$ to the union of
coordinate planes $x_1x_2x_3=0$ computed along the
$(1,1,1)$-direction and divided by $\sqrt{3}$.

Recall that the weight of a tiling is given by the formula:
$$
 w({\mathcal T})=const \cdot \prod\limits_{{\scalebox{0.2}{\includegraphics{hor_lozenge.pdf}}}\in {\mathcal
 T}}w(j),
$$
where $j$ is as in Figure 1.

Grouping all the $1\times1\times1$ cubes of the plane partitions
into columns with fixed coordinates $(x_2,x_3)$, we can rewrite the
weight in the form
$$
 w({\mathcal T})=const \cdot \prod\limits_{\scalebox{0.19}{\includegraphics{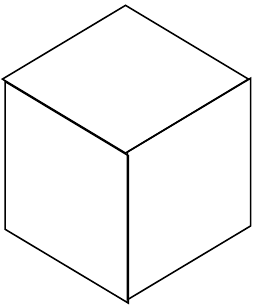}}}\frac{w(j)}{w(j-1)},
$$
where the product is taken over all cubes of the plane partition,
and $j$ stands for the $j$-coordinate (see Figure 1) of the top
vertex of the cube.

Collecting factors with the same $j$, we obtain
$$
 w({\mathcal T})=const \cdot \prod\limits_{v}\left(\frac{w(j)}{w(j-1)}\right)^{h(v)-\hat h(v)},
$$
where the product is taken over all points $v$ on the border of the
plane partition, whose three coordinates $(x_1,x_2,x_3)$ are
integers. Equivalently, one can think of the product being taken
over all vertices of the triangular lattice inside the hexagon. Note
that we may replace $h(v)-\hat h(v)$ by $h(v)$ since the remaining
product depends only on $(N,S,T)$. This gives
$$
 w({\mathcal T})=const \cdot
 \prod\limits_{v}\left(\frac{w(j)}{w(j-1)}\right)^{h(v)}.
$$

For the $q$-Racah case we obtain \blue{
$$
 w({\mathcal T})=const \cdot \prod\limits_{v}\left(\frac{{\kappa q^{j-(S+1)/2}-\dfrac1{\kappa q^{j-(S+1)/2}}}
}{{\kappa q^{j-(S+3)/2}-\dfrac1{\kappa q^{j-(S+3)/2}}}
}\right)^{h(v)},
$$}
while for the $q$-Hahn case
$$
 w({\mathcal T})=const \cdot \prod\limits_{v}\left(\frac{q^{-j}}
{q^{-j+1}}\right)^{h(v)}=const \cdot
\prod\limits_{v}q^{-h(v)}=const\cdot q^{-volume}.
$$

\label{Section_weight_representation}

\subsection{Results and variational interpretation}

Our results are of two kinds.

First, for each of the probability distributions on tilings described above
we construct explicit discrete time Markov chains that relate random
tilings of hexagons of various sizes.  The elementary steps of these chains
change the size of the hexagon from $a\times b\times c$ to $a\times (b\mp1)
\times (c\pm1)$.

Randomness in each step consists in generating finitely many
independent one-dimensional random variables. It takes $O(a(b+c))$
arithmetic operations to produce a tiling of the $a\times (b-1)
\times (c+1)$ hexagon using a tiling of the $a\times b \times c$
hexagon.

Together with the trivial observation that there is exactly one tiling
of an $a\times (b+c)\times 0$ hexagon, these chains provide, in
particular, an efficient perfect sampling algorithm for random
tilings distributed according to the real $q$-Racah, imaginary
$q$-Racah, trigonometric $q$-Racah, $q$-Hahn, Racah and Hahn
distributions.

A description of the algorithm can be found in Section
\ref{Section_perfect_sampling_algorithm}, and in Section
\ref{Section_comp_simulations} we provide some pictures generated
using this algorithm.

\smallskip

Second, we evaluate the asymptotics of the local behavior of our
measures as all sides of the hexagon tend to infinity comparably.

It is known, see \cite{Ke1}, \cite{Ke2}, \cite{OR}, \cite{Sh},
\cite{KOS}, \cite{BS}, that for any three positive numbers
$(p_1,p_2,p_3)$ with $p_1+p_2+p_3=1$, there exists a unique
translation-invariant Gibbs measure on lozenge tilings of the whole
plane such that in a large box, the numbers $\{p_j\}_{j=1}^3$
provide asymptotic ratios of the number of lozenges of three types.
It is convenient to encode the triple $(p_1,p_2,p_3)$ by a complex
number $z$ in the upper-half plane so that the triangle $(0,1,z)$
has angles $(\pi p_1,\pi p_2,\pi p_3)$. The correspondence between
angles and lozenge types is indicated in the figure below.

\begin{center}
 \scalebox{0.6}{\includegraphics{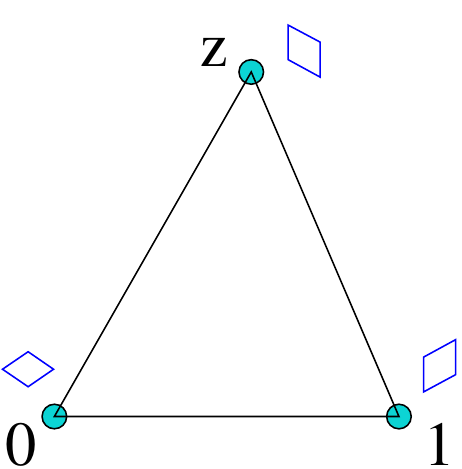}}
\end{center}

If one of the $p_j$'s tends to 1 (for example, $z$ tends to a point
in $\mathbb{R}$ away from $\{0,1\}$), then the corresponding measure
degenerates to the ``frozen'' (nonrandom) tiling of the plane by
lozenges of the corresponding type.

\begin{theorem}
\label{Th_Local_measures}  Introduce a small parameter $\varepsilon \ll 1$, and set
 $$
   S= \mathsf S\varepsilon^{-1} + o(\varepsilon^{-1}),\quad
   T= \mathsf T\varepsilon^{-1} + o(\varepsilon^{-1}),\quad
   N= \mathsf N\varepsilon^{-1} + o(\varepsilon^{-1}),\quad
   q= {\mathsf q}^{\varepsilon + o(\varepsilon)}.
 $$
 Then near a given point $(t,x)=(\mathsf t\varepsilon^{-1},\mathsf x\varepsilon^{-1})$,
 the random tiling  converges, as $\varepsilon\to 0$, to a certain
 ergodic translation-invariant Gibbs measure.
The parameter $z$ of this measure (encoding the slope $(p_1,p_2,p_3)$ via
angles as described) is the unique solution in the upper half-plane to the
(quadratic) equation
 \begin{equation}
  \label{eq_Q_qRacah}
  Q(u,v)=0,
\end{equation}
 where
\begin{equation}
 u=\frac{z{\mathsf q}^{\mathsf t}-\kappa ^2 {\mathsf q}^{-\mathsf S+2\mathsf x}}{1-z\kappa ^2{\mathsf q}^{-\mathsf S+2\mathsf x-\mathsf
 t}}\,,\qquad
 v=\frac{(1-z){\mathsf q}^{\mathsf x}}{1-z\kappa ^2{\mathsf q}^{-\mathsf S+2\mathsf x-\mathsf
 t}}\,,
\end{equation}
 and $Q$ is a degree 2 polynomial in $(u,v)$:
 \begin{multline}
\label{eq_Q_def_qRacah}
 Q(u,v)=u^2 \\
+ \biggl({\mathsf q}^{\mathsf T-\mathsf S-\mathsf N}+\kappa ^2(1+{\mathsf q}^{-\mathsf S+\mathsf
N+\mathsf T}
 +{\mathsf q}^{-2\mathsf S+\mathsf T}+{\mathsf q}^{-\mathsf S-\mathsf N}
 -{\mathsf q}^{-\mathsf S}-{\mathsf q}^{-\mathsf S+\mathsf T})+\kappa ^4{\mathsf q}^{-\mathsf S+\mathsf N}\biggr)v^2\\
+ \biggl({\mathsf q}^{\mathsf T-\mathsf S}+{\mathsf q}^{-\mathsf N}+\kappa ^2({\mathsf q}^{\mathsf
N}+{\mathsf q}^{- \mathsf S})\biggr)uv
 -
 \biggl({\mathsf q}^{\mathsf T}+{\mathsf q}^{\mathsf T-\mathsf S-\mathsf N}+\kappa ^2(1+q^{\mathsf N-\mathsf S+\mathsf T})\biggr)v \\
- (1+{\mathsf q}^{\mathsf T})u +{\mathsf q}^{\mathsf T}.
\end{multline}
If the solutions to this equation in $z$ are real, then the limit measure
is frozen.
\end{theorem}

\medskip
 Let us now explain how one could guess these formulas.
 In Section \ref{Section_Bulk_limits} we present a rigorous proof of Theorem \ref{Th_Local_measures} which
 uses an argument of a different kind. The remainder of this section
 is purely empirical; we hope to address the same issues rigorously in
 a later publication.
\smallskip

Although Theorem \ref{Th_Local_measures} describes the microscopic
behavior of our model, the parameters $(p_1,p_2,p_3)$ of the limit
measure are closely connected with macroscopic properties.

If we view tilings as stepped surfaces in a box and scale them in
such a way that the bounding box remains fixed, then it is plausible
that in the limit we will observe a deterministic limit shape. The
normal vector to this limit shape at any point has to coincide with
the vector $(p_1,p_2,p_3)$ of the local limit measure at this point.

The concentration of the measure near the limit shape is known in
the $q$-Hahn case, see \cite{CKP}, \cite{KO}. The limit shape is the
unique solution of a certain variational problem. It is not hard to
pose such a variational problem in the $q$-Racah case as well.

Recall that in Section \ref{Section_weight_representation} we found the following representation
for the weight of a tiling: \blue{
$$
 w({\mathcal T})=const \cdot \prod\limits_{v}\left(\frac{\kappa q^{j-(S+1)/2}-\dfrac1{\kappa q^{j-(S+1)/2}}}
 {\kappa q^{j-(S+3)/2}-\dfrac1{\kappa q^{j-(S+3)/2}}}\right)^{h(v)}.
$$}

Taking the logarithm of $w({\mathcal T})$ and removing the constant we get \blue{
$$
 \sum\limits_{v}h(v) \left[ \ln\left({\kappa q^{j-(S+1)/2}-\frac1{\kappa q^{j-(S+1)/2}}}\right)
   -\ln\left({\kappa q^{j-(S+3)/2}-\frac1{\kappa q^{j-(S+3)/2}}}\right)
   \right].
$$}

This is a Riemannian sum for an integral, and as $\varepsilon\to 0$
this yields, up to second order terms,
\begin{equation*}
  \frac1{\varepsilon^2}\int\limits_{\text{hexagon}} \mathsf h(\mathsf x,\mathsf t) \dfrac{\partial\ln({\kappa \mathsf q^{-\mathsf S/2+\mathsf j}-\frac1{\kappa \mathsf q^{-\mathsf S/2+\mathsf
 j}}})}{\partial \mathsf j},
\end{equation*}
where $\mathsf x$, $\mathsf t$ are normalized coordinates inside the
hexagon, $\mathsf h$ is the normalized height function, and $\mathsf
j=\mathsf x-\mathsf t/2$.

Following \cite{CKP} and \cite{KOS} we know that the number of
stepped surfaces in an $\varepsilon$-neighborhood of a given height
function is asymptotically
$$
\exp\left(\frac{1}{\varepsilon^2}\int \sigma(\nabla \mathsf
h)\right),
$$
 where $\sigma$ is the surface tension that can be expressed through the Lobachevski function.
 (Note that the signs in \cite{CKP} and \cite{KOS} are different, we follow \cite{CKP}.)

Consequently, the probability of the stepped surfaces in the
$\varepsilon$-neighborhood of a given height function is
asymptotically proportional to
$$
\exp\left[\frac{1}{\varepsilon^2}\left(\int \sigma(\nabla \mathsf
h)+\int\limits \mathsf h(\mathsf x,\mathsf t)
\dfrac{\partial\ln({\kappa \mathsf q^{-\mathsf S/2+\mathsf
j}-\frac1{\kappa \mathsf q^{-\mathsf S/2+\mathsf
 j}}})}{\partial \mathsf j} \right)\right],
$$
and to find the limit shape one has to maximize this expression. We
conclude that the limit shape can be found as a solution of the
variational problem:
$$
 \int \sigma(\nabla \mathsf h)+\int\limits \mathsf h(\mathsf x,\mathsf t) \dfrac{\partial\ln({\kappa \mathsf q^{-\mathsf S/2+\mathsf j}-\frac1{\kappa \mathsf q^{-\mathsf S/2+\mathsf
 j}}})}{\partial \mathsf j} \to \max
$$

Let us write down the Euler-Lagrange equation for this variational
problem. It is convenient to use our parameter $z$ instead of
partial derivatives of the limit height function $\mathsf h$. We
know (see \cite{KO}) that the Euler-Lagrange equation for the first
term is the complex Burgers equation
$$
 \frac{z_t}{z}-\frac{z_x}{1-z}=0.
$$

 Adding it to the trivial Euler-Lagrange equation for the second term we
 obtain
$$
 \frac{z_t}{z}-\frac{z_x}{1-z}=\dfrac{\partial\ln({\kappa \mathsf q^{-\mathsf S/2+\mathsf j}-\frac1{\kappa \mathsf q^{-\mathsf S/2+\mathsf
 j}}})}{\partial \mathsf j},
$$
or
$$
\frac{z_t}{z}-\frac{z_x}{1-z}= \ln(\mathsf q) \frac{\kappa ^2\mathsf
q^{-\mathsf S+2x-t}+1} {{\kappa ^2\mathsf q^{-\mathsf S+2x-t}-1}}.
$$

This is a quasilinear equation. The equations for characteristics
have the form
$$
 \frac{dt}{ds}=\frac 1z,
\quad \frac{dx}{ds}=\frac1{z-1}, \quad \frac{dz}{ds}=\ln(\mathsf q)
\frac{\kappa ^2\mathsf q^{-\mathsf S+2x-t}+1} {{\kappa ^2\mathsf
q^{-\mathsf S+2x-t}-1}}.
$$

We find two first integrals
$$
 u=\frac{z{\mathsf q}^{\mathsf t}-\kappa ^2 {\mathsf q}^{-\mathsf S+2\mathsf x}}{1-z\kappa ^2{\mathsf q}^{-\mathsf S+2\mathsf x-\mathsf
 t}},\qquad
 v=\frac{(1-z){\mathsf q}^{\mathsf x}}{1-z\kappa ^2{\mathsf q}^{-\mathsf S+2\mathsf x-\mathsf t}}.
$$

Any solution of the partial derivative equation has the form (cf.
the proof of Corollary 1 in \cite{KO})
$$
Q(u,v)=0,
$$
where $Q$ is a suitable analytic function. This equation defines $z$
as a function of $\mathsf t$ and $\mathsf x$.

In the $q$-Hahn case, it is known that $Q(u,v)$ is a second degree
polynomial. It is natural to assume that the same is true for the
$q$-Racah case, and then the requirement that the $(\mathsf
t,\mathsf x)$-curve where $z$ degenerates to $\mathbb{R}$ is tangent
to the six sides of the hexagon leads to the formula of Theorem
\ref{Th_Local_measures}. It remains a challenge to check if $Q$ is
still algebraic for more general polygonal domains, as was shown
in \cite{KO} for the $q$-Hahn case.

\section{Weight sums}

\label{Section_weight_sums}

A horizontal lozenge in the tiling interpretation corresponds to a
hole (=absence of a particle) in the nonintersecting paths or
$N$-point configuration interpretation.  The coordinates $(i,j)$ of the
horizontal lozenge in the formulas \eqref{elliptic_weight} and
\eqref{qRacah_weight} correspond to the coordinates $(t,x)$ of a hole
in the following way:
$$
\begin{cases}
i=t,\\
\blue{ j=x-t/2+1.}
\end{cases}
$$

Now consider any vertical section $t=t_0$ of the family of
nonintersecting paths (see Figure 3) and fix the corresponding
$N$-point configuration $x_1<x_2<\dots<x_N$.

Denote by $C(t;x_1,\dots,x_N)$ the product of weights corresponding
to the holes on this vertical line.

Denote by $L(t;x_1,\dots,x_N)$ the sum of the products of weights
corresponding to holes situated to the left of the vertical line.
The sum is taken over all families of paths connecting the points
$\{(0,0),(0,1),\dots,(0,N-1)\}$ to the points
$\{(t_0,x_1),(t_0,x_2),\dots,(t_0,x_N)\}$.

Denote by $R(t;x_1,\dots,x_N)$ the sum of the products of weights
corresponding to holes situated to the right of the vertical line.
The sum is taken over all families of paths connecting the points
$\{(t_0,x_1),(t_0,x_2),\dots,(t_0,x_N)\}$ to the points
$\{(T,S),(T,S+1),\dots,(T,S+N-1)\}$.

 The following three propositions correspond to the case of the
$q$-Racah weight \eqref{qRacah_weight}.

Set
$$
 {\mu_{t,S}(x)}=q^{-x}+\kappa ^2q^{x-S-t+1}
$$

Here and below we use the $q$-Pochhammer symbol:
$$
 (a;q)_{n}=(1-a)(1-aq)\dots(1-aq^{n-1}).
$$
\begin{proposition}
\label{prop_L_t} We have
\begin{multline*}
L_t(x_1,\dots,x_N)=const_t \cdot \prod_{1\le i<j\le
N}(\mu_{t,S}(x_i)-\mu_{t,S}(x_j)) \\ \times
 \prod_{i=1}^N \frac{q^{x_i(t+N-1)} (1-\kappa ^2q^{2x_i-S-t+1})}{(q^{-1};q^{-1})_{t+N-1-x_i}\cdot(q;q)_{x_i}\cdot
 (\kappa^2 q^{x_i-t-S+1};q)_{t+N}}\,.
\end{multline*}

\end{proposition}
\begin{proof}
 Here and in the proof of the next lemma we should consider four cases depending on the value of $t$ (see formulas
 \eqref{QRacah_case1}-\eqref{QRacah_case4}). The proofs are similar in all
 four cases and we consider only the one
 that corresponds to the pictures in Lemma \ref{Lemma_Weight_left} and Lemma
 \ref{Lemma_weight_right} (i.e., $S<t<T-S$)

 Let us use Lemma \ref{Lemma_Weight_left} that expresses the required weight sum in terms of the
 point configuration complementary to $\{x_i\}$, i.e., in terms of the
 positions of the horizontal lozenges
 that we call holes. Denote the positions of holes by $\{y_i\}$. Performing the limit transition
 \eqref{eq_limit_transition} we get.

 \begin{multline*}
  L_t(x_1,\dots,x_N)=const \cdot \prod_{1\le i<j\le
S}(\mu_{t,S}(y_i)-\mu_{t,S}(y_j))\\ \times \prod_{1\le i\le S}
\frac{q^{(S-t) (S+N-1-y_i)}(q^{t-S+1};q)_{S+N-1-y_i}
(q^{-2N-S+t+1}/\kappa^2;q)_{S+N-1-y_i}}{(q;q)_{S+N-1-y_i}(q^{-2N+1}/\kappa^2;q)_{S+N-1-y_i}}\,.
 \end{multline*}

 To finish the proof we rewrite the last formula in terms of particles $\{x_i\}$ instead of holes
 $\{y_i\}$. For the second factor this procedure is simple, while for the first one we use the
 following observation:
 \begin{multline*}
  \prod_{1\le i<j\le
N}(\mu_{t,S}(y_i)-\mu_{t,S}(y_j))= \prod_{1\le i<j\le
N}(\mu_{t,S}(x_i)-\mu_{t,S}(x_j))\\ \times \prod_{0\le u<v\le
S+N-1}(\mu_{t,S}(u)-\mu_{t,S}(v))
\\ \times \prod_{1\le i\le N}\frac{1}{\prod\limits_{0\le u <x_i}
(\mu_{t,S}(x_i)-\mu_{t,S}(u))\prod\limits_{x_i<u\le S+N-1}
(\mu_{t,S}(u)-\mu_{t,S}(x_i))}\,.
 \end{multline*}
 The product over $u<v$ depends only on $t$, while the last two
products over $u$ are easily expressible in terms of $q$-Pochhammer
symbols.

\end{proof}

\begin{proposition}
\label{prop_R_t} We have
\begin{multline*}
R_t(x_1,\dots,x_N)=const_t \cdot \prod_{1\le i<j\le N}(\mu_{t,S}(x_i)-\mu_{t,S}(x_j))\\
\times
 \prod_{i=1}^N\frac{(1-\kappa ^2q^{2x_i-S-t+1})q^{x_i(T-t+N-1)}}{(q^{-1};q^{-1})_{S+N-1-x_i}\cdot(q;q)_{x_i+T-t-S}
 \cdot (\kappa^2q^{x_i-T+1};q)_{N+T-t}}\,.
\end{multline*}

\end{proposition}
\begin{proof}
 Performing the limit transition \eqref{eq_limit_transition} in Lemma
 \ref{Lemma_weight_right} we get
 \begin{multline*}
  R_t(x_1,\dots,x_N)=const \cdot \prod_{1\le i<j\le
S}(\mu_{t,S}(y_i)-\mu_{t,S}(y_j))\\ \times \prod_{1\le i\le S} \frac{q^{(S-T+t)(S+N-1-y_i)}
(q^{-N-S+1};q)_{S+N-1-y_i}
(q^{T-S-N+1}/\kappa^2;q)_{S+N-1-y_i}}{(q^{-N+1-T+t};q)_{S+N-1-y_i}(q^{t-N+1}/\kappa^2;q)_{S+N-1-y_i}}
 \end{multline*}
 Again expressing all the factors in terms of $\{x_i\}$ we obtain the desired result.

\end{proof}

\begin{proposition}
\label{prop_C_t} We have
$$
 C_t(x_1,\dots,x_N)=const_t \cdot\prod_{i=1}^N \frac{q^{x_i}}{1-\kappa
 ^2q^{2x_i-S-t+1}}\,.
$$
\end{proposition}
\begin{proof}
 Clearly,
 $$
  C_t(x_1,\dots,x_N)=\prod_{i=1}^{S} \left(\kappa q^{y_i-S/2-t/2+1/2}-\frac1{\kappa q^{y_i-S/2-t/2+1/2}}\right),
 $$
 where $\{y_i\}$ is the point configuration complementary to $\{x_i\}$. Expressing the
 product in terms of $\{x_i\}$ we get the result.
\end{proof}

\section{Distributions and transition probabilities}
\label{Section_distributions}

We consider our probability measure on the set of tilings of a given
hexagon as a Markov chain in the space of $N$-tuples of integers.
The Markov property can be easily seen in the following form: The
past and the future are independent given the present. In this way
the Markov property reduces to the fact that a lozenge tiling of the
hexagon is a union of a tiling to the left of a vertical line
$t=const$ and a tiling to the right of this vertical line. Denote
this Markov chain by $X(t)$, $t=0,1,\dots T$.

Set
$$
 \mathfrak X_{N,T}^{S,t}=\{x\in\mathbb Z: \max(0,t+S-T)  \le x\le \min(t+N-1,S+N-1)\}
$$
and
$$
 \mathcal X_{N,T}^{S,t}=\{ (x_1,x_2,\dots,x_N)\in (\mathfrak
 X_{N,T}^{S,t})^N: x_1<x_2<\dots<x_N\}.
$$
$\mathfrak X_{N,T}^{S,t}$ is the section of our hexagon by the
vertical line with coordinate $t$, and $\mathcal X_{N,T}^{S,t}$ is
the set of all $N$-tuples  in this section.

Clearly, $X(t)$ takes values in $\mathcal X_{N,T}^{S,t}$.

The following theorem gives the one-dimensional distributions of our
chain at the time $t$, which are probability distributions on
$N$-tuples of integers. Denote by $\rho_{S,t}$ the one-dimensional
distribution of the process $X(t)$ (below we explain why we keep $S$
in the notation but omit all the other parameters).

\begin{theorem}
\label{Theorem_One_dim_distribution}
$$
 {\rm Prob}\{X(t)=(x_1,x_2,\dots,x_N)\} = const\cdot \prod_{i<j}(\mu_{t,S}(x_i)-\mu_{t,S}(x_j))^2 \prod_{i=1}^N w_{t,S}(x_i),
$$
where
$$
\mu_{t,S}(x)=q^{-x}+\kappa ^2q^{x-S-t+1} $$ and
\begin{multline*}
w_{t,S}(x)=\frac{(-1)^{t+S}q^{x(2N+T-1)}(1- \kappa ^2q^{2x-t-S+1})}{(q;q)_x (q;q)_{T-S-t+x}
(q^{-1};q^{-1})_{t+N-x-1}(q^{-1};q^{-1})_{S+N-x-1}}\\
\times\frac1{ ( \kappa ^2q^{x-T+1};q)_{T+N-t}( \kappa
^2q^{x-t-S+1};q)_{N+t} }\,.
\end{multline*}
\end{theorem}

\begin{proof}
Clearly,
$$
 {\rm Prob}\{X(t)=(x_1,x_2,\dots,x_N)\}\propto L(t;x_1,\dots,x_N)C(t;x_1,\dots,x_N)R(t;x_1,\dots,x_N)
$$
Propositions \ref{prop_L_t}, \ref{prop_R_t}, \ref{prop_C_t} imply
the result.
\end{proof}

Observe that $w_{t,S}(x)$ is (up to the factor not depending on $x$)
the weight function of the $q$-Racah orthogonal polynomials, see
e.~g. \cite[Section 3.2]{KS}. These polynomials are given by the
formula
\begin{equation}
\label{qRacah}
 R_n(\mu(x); \alpha, \beta, \gamma, \delta \mid
 q)=\,_4\phi_3\left(\begin{matrix}
 q^{-n},\alpha\beta q^{n+1},q^{-x},\gamma\delta q^{x+1}\\ \alpha q,\beta\delta q,\gamma q
 \end{matrix}\,\Bigl|\, q;q
 \right),
\end{equation}
where
$$
 \mu(x)=q^{-x}+\gamma\delta q^{x+1},
$$
and $\alpha q=q^{-M}$ or $\beta\delta q=q^{-M}$ or $\gamma q=q^{-M}$
for a nonnegative integer $M$. They are orthogonal on
$\{0,1,\dots,M\}$ with respect to the weight function
$$
 w(x)=\frac{(\alpha q, \beta\delta q, \gamma q, \gamma\delta q;q)_x}{(q,\alpha^{-1}\gamma\delta q,\beta^{-1}\gamma q,\delta
 q;q)_x}\frac{1-\gamma\delta q^{2x+1}}{(\alpha\beta q)^x (1-\gamma\delta q)}.
$$

The correspondence between the parameters of polynomials and
parameters of our model is established in the following way.

\begin{enumerate}
\item $t<S$,\quad $T-t-S>0$,\quad $0\le x \le t+N-1$,
 \begin{equation}
  \label{QRacah_case1}
   \begin{array}{lll}
      q_{(qRacah)} &= & q\\
      \alpha_{(qRacah)} &= & q^{-S-N}\\
      \beta_{(qRacah)} &= & q^{S-T-N}\\
      \gamma_{(qRacah)} & = & q^{-t-N}\\
      \delta_{(qRacah)}&=&   \kappa ^2 q^{-S+N}
   \end{array}
 \end{equation}

\item $S-1<t<T-S+1$,\quad $0\le x\le S+N-1$,
 \begin{equation}
  \label{QRacah_case2}
   \begin{array}{lll}
      q_{(qRacah)} &= & q\\
      \alpha_{(qRacah)} &= & q^{-t-N}\\
      \beta_{(qRacah)} &= & q^{t-T-N}\\
      \gamma_{(qRacah)} & = & q^{-S-N}\\
      \delta_{(qRacah)}&=&   \kappa ^2 q^{-t+N}
   \end{array}
 \end{equation}

\item $T-S-1<t<S$,\quad $0\le x-(t+S-T) \le T-S+N-1$,
 \begin{equation}
  \label{QRacah_case3}
   \begin{array}{lll}
      q_{(qRacah)} &= & q\\
      \alpha_{(qRacah)} &= & q^{-T+t-N}\\
      \beta_{(qRacah)} &= & q^{-t-N}\\
      \gamma_{(qRacah)} & = & q^{-T-N+S}\\
      \delta_{(qRacah)}&=&   \kappa ^2 q^{-T+t+N}\\
      x_{(qRacah)}&=&T-t-S+x
   \end{array}
 \end{equation}

 \item $t>T-S-1$,\quad $t>S-1$,\quad $0\le x-(t+S-T)\le T-t+N-1$,
 \begin{equation}
  \label{QRacah_case4}
   \begin{array}{lll}
      q_{(qRacah)} &= & q\\
      \alpha_{(qRacah)} &= & q^{-T-N+S}\\
      \beta_{(qRacah)} &= & q^{-S-N}\\
      \gamma_{(qRacah)} & = & q^{-T+t-N}\\
      \delta_{(qRacah)}&=&   \kappa ^2 q^{-T+N+S}\\
      x_{(qRacah)}&=&T-t-S+x
   \end{array}
 \end{equation}

\end{enumerate}

Let us also describe what happens if one performs the limit transitions
described in Section \ref{Section_pr_models}.

If we send $\kappa \to 0$ (the weight of a plane partition becomes
proportional to $q^{-volume})$, then we obtain the weight function
$$
w(x)=\frac{q^{x(2N+T-1)}}{(q;q)_x (q;q)_{T-S-t+x}
(q^{-1};q^{-1})_{t+N-x-1}(q^{-1};q^{-1})_{S+N-x-1}}
$$
This is exactly (up to the factor not depending on $x$) the weight
function of the $q$-Hahn polynomials.

If we set $ \kappa =q^{K}$ and send $q\to 1$, then the weight
function becomes
\begin{multline*}
 w(x)=\frac 1{x! (T-S-t+x)! (t+N-x-1)!
 (S+N-x-1)!}\\ \times \frac{K+2x-t-S+1}{(K+x-T+1)_{T+N-t} (K+x-t-S+1)_{N+t}}.
\end{multline*}
This is the weight function of the Racah orthogonal polynomials.

Finally, if we send $ \kappa \to 0$ and set $q=1$ (this case
corresponds to the uniform measure on the set of lozenge tilings of
the hexagon) then we get the weight function
$$
w(x)=\frac{1}{x! (T-S-t+x)! (t+N-x-1)! (S+N-x-1)!}.
$$
This is the weight function of the Hahn polynomials. This case was
previously studied in \cite{J_nonintersecting}, \cite{J_Hahn},
\cite{Gor}, see also references therein.

\bigskip

We also need the transition probabilities of the Markov chain $X(t)$.

\begin{proposition}
\label{Proposition_tr_prob}
\begin{multline*}
 {\rm
 Prob}\{X(t+1)=Y|X(t)=X\}\\=const\cdot\prod_{i<j}\frac{\mu_{t+1,S}(y_i)-\mu_{t+1,S}(y_j)}
 {\mu_{t,S}(x_i)-\mu_{t,S}(x_j)}\prod_{y_i=x_i+1}
 w_1(x_i)\prod_{y_i=x_i}
 w_0(x_i),
\end{multline*}
where
$$
 w_0(x)=-(1-q^{x+T-t-S})\frac{1- \kappa ^2q^{x+N-t}} {1- \kappa ^2
 q^{2x-t-S+1}}
$$
and
$$
 w_1(x)=q^{T+N-1-t}(1-q^{x-S-N+1}) \frac{1- \kappa ^2 q^{x-T+1}} {1- \kappa ^2
 q^{2x-t-S+1}}.
$$
\end{proposition}
\begin{proof} We use
\begin{multline*}
 {\rm
 Prob}\{X(t+1)=Y|X(t)=X\}=\frac{L_t(X)C_t(X)C_{t+1}(Y)R_{t+1}(Y)}{L_t(X)C_t(X)R_t(X)}
 \\=\frac{R_{t+1}(Y)C_{t+1}(Y)}{R_{t}(X)}
\end{multline*}
and Propositions \ref{prop_C_t}, \ref{prop_R_t}.
\end{proof}

Next let us compute the \emph{cotransition} probabilities ($t\to
t-1$).

\begin{proposition}
\label{Proposition_co-tr}
\begin{multline*}
 {\rm
 Prob}\{X(t-1)=Y|X(t)=X\}=\\ const\cdot\prod_{i<j}\frac{\mu_{t-1,S}(y_i)-\mu_{t-1,S}(y_j)}
 {\mu_{t,S}(x_i)-\mu_{t,S}(x_j)}
 \tilde w_1(x_i)\prod_{y_i=x_i}
 \tilde w_0(x_i),
\end{multline*}
where
$$
 \tilde w_0(x)=-(1-q^{x-t-N+1})\frac {1- \kappa ^2q^{x-S-t+1}}{1- \kappa ^2
 q^{2x-t-S+1}}
$$
and
$$
 \tilde w_1(x)=q^{-(t+N-1)}(1-q^x)\frac {1- \kappa ^2q^{x+N-S}}{1- \kappa ^2
 q^{2x-t-S+1}}
$$
\end{proposition}
\begin{proof} We use
\begin{multline*}
 {\rm
 Prob}\{X(t-1)=Y|X(t)=X\}=\frac{L_{t-1}(Y)C_{t-1}(Y)C_{t}(X)R_{t}(X)}{L_t(X)C_t(X)R_t(X)}\\
 =\frac{L_{t-1}(Y)C_{t-1}(Y)}{L_{t}(X)}
\end{multline*}
and Propositions \ref{prop_L_t}, \ref{prop_C_t}.
\end{proof}

\section{Families of stochastic matrices}

This section and the next one are similar to \cite{BG}, where the
Hahn case was treated, and we have tried to keep the notations and
statements of theorems unchanged where possible.

\subsection{Definition of matrices}

We want to introduce four families of stochastic matrices
$P^{S,t}_{t+}$, $P^{S,t}_{t-}$, $P^{S,t}_{S+}$, $P^{S,t}_{S-}$.

$P^{S,t}_{t+}(X,Y)$ is an $|\mathcal X^{S,t}|\times|\mathcal
X^{S,t+1}|$ matrix, $X=(x_1<\dots<x_N)\in\mathcal X^{S,t}$,
$Y=(y_1<\dots<y_N)\in\mathcal X^{S,t+1}$;

if  $y_i-x_i\in\{0,1\}$ for every $i$, then
\begin{multline*}
P^{S,t}_{t+}(X,Y) =
const\cdot\prod_{i<j}\frac{\mu_{t+1,S}(y_i)-\mu_{t+1,S}(y_j)}
 {\mu_{t,S}(x_i)-\mu_{t,S}(x_j)}\prod_{y_i=x_i+1}
 w_1(x_i)\prod_{y_i=x_i}
 w_0(x_i),
\end{multline*}
 where
$$
 w_0(x)=-(1-q^{x+T-t-S})\frac{1- \kappa ^2q^{x+N-t}} {1- \kappa ^2
 q^{2x-t-S+1}}\,,
$$$$
 w_1(x)=q^{T+N-1-t}(1-q^{x-S-N+1}) \frac{1- \kappa ^2 q^{x-T+1}} {1- \kappa ^2
 q^{2x-t-S+1}},
$$
 and $P^{S,t}_{t+}(X,Y)=0$
otherwise.

$P^{S,t}_{S+}(X,Y)$ is an $|\mathcal X^{S,t}|\times|\mathcal
X^{S+1,t}|$ matrix, $X=(x_1<\dots<x_N)\in\mathcal X^{S,t}$,
$Y=(y_1<\dots<y_n)\in\mathcal X^{S+1,t}$;

If  $y_i-x_i\in\{0,1\}$ for every $i$, then
\begin{multline*}
P^{S,t}_{S+}(X,Y)=
const\cdot\prod_{i<j}\frac{\mu_{t,S+1}(y_i)-\mu_{t,S+1}(y_j)}
 {\mu_{t,S}(x_i)-\mu_{t,S}(x_j)}\prod_{y_i=x_i+1}
 w_1(x_i)\prod_{y_i=x_i}
 w_0(x_i),
\end{multline*}
where
$$
 w_0(x)=-(1-q^{x+T-t-S})\frac{1- \kappa ^2q^{x+N-S}} {1- \kappa ^2
 q^{2x-t-S+1}}\,,
$$
$$
 w_1(x)=q^{T+N-1-S}(1-q^{x-t-N+1}) \frac{1- \kappa ^2 q^{x-T+1}} {1-\kappa ^2
 q^{2x-t-S+1}},
$$
 and $P^{S,t}_{S+}(X,Y)=0$
otherwise.

$P^{S,t}_{t-}(X,Y)$ is an $|\mathcal X^{S,t}|\times|\mathcal
X^{S,t-1}|$ matrix, $X=(x_1<\dots<x_N)\in\mathcal X^{S,t}$,
$Y=(y_1<\dots<y_n)\in\mathcal X^{S,t-1}$;

If  $y_i-x_i\in\{-1,0\}$ for every $i$, then
\begin{equation*}
P^{S,t}_{t-}(X,Y)=const\cdot
\prod_{i<j}\frac{\mu_{t-1,S}(y_i)-\mu_{t-1,S}(y_j)}
 {\mu_{t,S}(x_i)-\mu_{t,S}(x_j)}\prod_{y_i=x_i-1}
 \tilde w_1(x_i)\prod_{y_i=x_i}
 \tilde w_0(x_i),
\end{equation*}
where
$$
 \tilde w_0(x)=-(1-q^{x-t-N+1})\frac {1- \kappa ^2q^{x-S-t+1}}{1- \kappa ^2
 q^{2x-t-S+1}}\,,
$$
$$
 \tilde w_1(x)=q^{-(t+N-1)}(1-q^x)\frac {1- \kappa ^2q^{x+N-S}}{1- \kappa ^2
 q^{2x-t-S+1}},
$$
 and $P^{S,t}_{t-}(X,Y)=0$
otherwise.

$P^{S,t}_{S-}(X,Y)$ is an $|\mathcal X^{S,t}|\times|\mathcal
X^{S-1,t}|$ matrix, $X=(x_1<\dots<x_N)\in\mathcal X^{S,t}$,
$Y=(y_1<\dots<y_n)\in\mathcal X^{S-1,t}$;

If  $y_i-x_i\in\{-1,0\}$ for every $i$, then
\begin{equation*}
P^{S,t}_{S-}(X,Y)=const\cdot\prod_{i<j}\frac{\mu_{t,S-1}(y_i)-\mu_{t,S-1}(y_j)}
 {\mu_{t,S}(x_i)-\mu_{t,S}(x_j)}\prod_{y_i=x_i-1}
 \tilde w_1(x_i)\prod_{y_i=x_i}
 \tilde w_0(x_i),
\end{equation*}
where
$$
 \tilde w_0(x)=-(1-q^{x-S-N+1})\frac {1- \kappa ^2q^{x-S-t+1}}{1- \kappa ^2
 q^{2x-t-S+1}}\,,
$$
$$
 \tilde w_1(x)=q^{-(S+N-1)}(1-q^x)\frac {1- \kappa ^2q^{x+N-t}}{1- \kappa ^2
 q^{2x-t-S+1}},
$$
 and $P^{S,t}_{S-}(X,Y)=0$
otherwise.

\medskip

Looking at the sets that parameterize rows and columns of these
matrices one can say that $P^{S,t}_{t+}$ increases $t$,
$P^{S,t}_{t-}$ decreases $t$, while $P^{S,t}_{S+}$ increases $S$ and
$P^{S,t}_{S-}$ decreases $S$. This explains our notation.

\begin{theorem}
\label{Th_4stochmatr} With appropriate choices of normalizing
constants, all four types of matrices defined above are stochastic.
They preserve the family of measures $\rho_{S,t}$. In other words
\begin{equation}
\label{Simple_stM}
 \sum\limits_{Y\in\mathcal X^{S,t\pm 1}} P^{S,t}_{t\pm}(X,Y)=1,\quad
 \sum\limits_{Y\in\mathcal X^{S,S\pm 1}} P^{S,t}_{t\pm}(X,Y)=1,
\end{equation}
\begin{gather*}
 \rho_{S,t\pm 1}(Y)=\sum\limits_{X\in\mathcal X^{S,t}}
 P^{S,t}_{t\pm}(X,Y)\cdot\rho_{S,t}(X),\\
 \rho_{S\pm 1,t}(Y)=\sum\limits_{X\in\mathcal X^{S,t}}
 P^{S,t}_{S\pm}(X,Y)\cdot\rho_{S,t}(X).
\end{gather*}

\end{theorem}

\begin{proof}
Propositions \ref{Proposition_tr_prob} and \ref{Proposition_co-tr}
imply the claim for $P^{S,t}_{t+}(X,Y)$ and $P^{S,t}_{t-}(X,Y)$.

Now observe that the space $\mathcal X^{S,t}$ is unaffected when we
interchange parameters $t$ and $S$, i.e.,
$$
 \mathcal X^{S,t}=\mathcal X^{t,S}.
$$

Moreover, the measures $\rho_{S,t}$ are also invariant under
$S\leftrightarrow t$, i.e.,
$$
\rho_{S,t}=\rho_{t,S}.
$$
(This is a consequence of our special choice of the parameter
$\kappa $ which included additional factor $q^{-S/2}$.)

Finally, note that $P^{S,t}_{t+}(X,Y)$ becomes $P^{S,t}_{S+}(X,Y)$
under $S\leftrightarrow t$ and $P^{S,t}_{t-}(X,Y)$ becomes
$P^{S,t}_{S-}(X,Y)$.

Therefore, applying $S\leftrightarrow t$ to the relations for
$P_{t\pm}$ we obtain the needed relations for $P_{S\pm}$.
\end{proof}

\subsection{Determinantal representation}
In this section we write our stochastic matrices in a determinantal
form. This representation is very convenient for various
computations.

First, we introduce $4$ new two-diagonal matrices.

For $x\in\mathfrak X^{S,t}$,  $y\in\mathfrak X^{S,t+1}$,
$$
U^{S,t}_{t+}(x,y)=\begin{cases}
              -q^{T+N-1-t}(1-q^{x-S-N+1}) \frac{1- \kappa ^2 q^{x-T+1}} {1- \kappa ^2
 q^{2x-t-S+1}},&\text{if }y=x+1,\\
              (1-q^{x+T-t-S})\frac{1- \kappa ^2q^{x+N-t}} {1- \kappa ^2
 q^{2x-t-S+1}},&\text{if }y=x,\\
              0,&\text{otherwise;}
             \end{cases}
$$
for $x\in\mathfrak X^{S,t}$, $y\in\mathfrak
             X^{S+1,t}$,
$$
U^{S,t}_{S+}(x,y)=\begin{cases}
              -q^{T+N-1-S}(1-q^{x-t-N+1}) \frac{1- \kappa ^2 q^{x-T+1}} {1-\kappa ^2
 q^{2x-t-S+1}},&\text{if }y=x+1,\\
              (1-q^{x+T-t-S})\frac{1-\kappa ^2q^{x+N-S}} {1-\kappa ^2
 q^{2x-t-S+1}},&\text{if }y=x,\\
              0,&\text{otherwise;}
             \end{cases}
$$
for $x\in\mathfrak X^{S,t}$, $y\in\mathfrak
             X^{S,t-1}$,
$$
U^{S,t}_{t-}(x,y)=\begin{cases}
              q^{-(t+N-1)}(1-q^x)\frac {1-\kappa ^2q^{x+N-S}}{1-\kappa ^2
 q^{2x-t-S+1}},&\text{if }y=x-1,\\
              -(1-q^{x-t-N+1})\frac {1-\kappa ^2q^{x-S-t+1}}{1-\kappa ^2
 q^{2x-t-S+1}},&\text{if }y=x,\\
              0,&\text{otherwise;}
             \end{cases}
$$
and for  $x\in\mathfrak X^{S,t}$, $y\in\mathfrak X^{S-1,t}$,
$$
U^{S,t}_{S-}(x,y)=\begin{cases}
              q^{-(S+N-1)}(1-q^x)\frac {1-\kappa ^2q^{x+N-t}}{1-\kappa ^2
 q^{2x-t-S+1}},&\text{if }y=x-1,\\
              -(1-q^{x-S-N+1})\frac {1-\kappa ^2q^{x-S-t+1}}{1-\kappa ^2
 q^{2x-t-S+1}},&\text{if }y=x,\\
              0,&\text{otherwise.}
             \end{cases}
$$

It is possible to express the stochastic matrices $P^{S,t}_{t\pm}$,
$P^{S,t}_{S\pm}$ as certain minors of the matrices defined
above.

\begin{proposition} We have
\label{Proposition_det_form_transitional_probabilities}
$$P^{S,t}_{t+}(X,Y)=const\cdot\prod_{i<j}\frac{\mu_{t+1,S}(y_i)-\mu_{t+1,S}(y_j)}
 {\mu_{t,S}(x_i)-\mu_{t,S}(x_j)}
   \det[U^{S,t}_{t+}(x_i,y_j)]_{i,j=1,\dots,N}
$$

$$P^{S,t}_{S+}(X,Y)=const\cdot\prod_{i<j}\frac{\mu_{t,S+1}(y_i)-\mu_{t,S+1}(y_j)}
 {\mu_{t,S}(x_i)-\mu_{t,S}(x_j)} \det[U^{S,t}_{S+}(x_i,y_j)]_{i,j=1,\dots,N}
 $$

$$P^{S,t}_{t-}(X,Y)=const\cdot\prod_{i<j}\frac{\mu_{t-1,S}(y_i)-\mu_{t-1,S}(y_j)}
 {\mu_{t,S}(x_i)-\mu_{t,S}(x_j)}
   \det[U^{S,t}_{t-}(x_i,y_j)]_{i,j=1,\dots,N}
 $$

$$P^{S,t}_{S-}(X,Y)= const\cdot\prod_{i<j}\frac{\mu_{t,S-1}(y_i)-\mu_{t,S-1}(y_j)}
 {\mu_{t,S}(x_i)-\mu_{t,S}(x_j)}
   \det[U^{S,t}_{S-}(x_i,y_j)]_{i,j=1,\dots,N}
 $$
\end{proposition}
\begin{proof}
Straightforward computation using the definitions of the stochastic
matrices $P^{S,t}_{t\pm}$, $P^{S,t}_{S\pm}$ and the matrices
$U^{S,t}_{t\pm}$, $U^{S,t}_{S\pm}$.

Any submatrix of a two-diagonal matrix, which has a nonzero
determinant, is block-diagonal, where each block is either an upper
or a lower triangular matrix. Thus, any nonzero minor is a product
of suitable matrix elements.
\end{proof}

\subsection{Commutativity}

\begin{theorem}
 The families of stochastic matrices $P^{S,t}_{t\pm}$ and
 $P^{S,t}_{S\pm}$ commute, that is
 $$
   P^{S,t}_{t+}\cdot P^{S,t+1}_{S-}=P^{S,t}_{S-}\cdot
   P^{S-1,t}_{t+},
 $$
 $$
   P^{S,t}_{t-}\cdot P^{S,t-1}_{S-}=P^{S,t}_{S-}\cdot
   P^{S-1,t}_{t-},
 $$
$$
   P^{S,t}_{t+}\cdot P^{S,t+1}_{S+}=P^{S,t}_{S+}\cdot
   P^{S+1,t}_{t+},
 $$
 $$
   P^{S,t}_{t-}\cdot P^{S,t-1}_{S+}=P^{S,t}_{S+}\cdot
   P^{S+1,t}_{t-},
 $$
 for any meaningful values of $S$ and $t$.
\end{theorem}

\begin{proof}
Proofs of all four cases are very similar and we consider only the first one.
\begin{multline*}
 (P^{S,t}_{t+}\cdot P^{S,t+1}_{S-})(X,Y)=\sum\limits_{Z\in\mathcal
 X^{S,t+1}}P^{S,t}_{t+}(X,Z)\cdot P^{S,t+1}_{S-}(Z,Y)\\
 =const\cdot \prod_{i<j}\frac{\mu_{t+1,S-1}(y_i)-\mu_{t+1,S-1}(y_j)}
 {\mu_{t,S}(x_i)-\mu_{t,S}(x_j)} \\ \times
  \sum\limits_{Z\in\mathcal X^{S,t+1}}  \det[U^{S,t}_{t+}(x_i,z_j)]_{i,j=1,\dots,N}
  \det[U^{S,t+1}_{S-}(z_i,y_j)]_{i,j=1,\dots,N}.
\end{multline*}
Applying the Cauchy-Binet identity we obtain
\begin{multline*}
  \sum\limits_{Z\in\mathcal X^{S,t+1}}  \det[U^{S,t}_{t+}(x_i,z_j)]_{i,j=1,\dots,N}
  \det[U^{S,t+1}_{S-}(z_i,y_j)]_{i,j=1,\dots,N}\\ =
  \det[ (U^{S,t}_{t+}\cdot
  U^{S,t+1}_{S-})(x_i,y_j)]_{i,j=1,\dots,N}.
\end{multline*}
Thus,
\begin{multline*}
(P^{S,t}_{t+}\cdot P^{S,t+1}_{S-})(X,Y)\\ =const\cdot
\prod_{i<j}\frac{\mu_{t+1,S-1}(y_i)-\mu_{t+1,S-1}(y_j)}
 {\mu_{t,S}(x_i)-\mu_{t,S}(x_j)}\det[ (U^{S,t}_{t+}\cdot  U^{S,t+1}_{S-})(x_i,y_j)]_{i,j=1,\dots,N}.
\end{multline*}
Similarly,
\begin{multline*}
(P^{S,t}_{S-}\cdot P^{S-1,t}_{t+})(X,Y)\\=const\cdot
\prod_{i<j}\frac{\mu_{t+1,S-1}(y_i)-\mu_{t+1,S-1}(y_j)}
 {\mu_{t,S}(x_i)-\mu_{t,S}(x_j)} \det[ (U^{S,t}_{S-}\cdot
U^{S-1,t}_{t+})(x_i,y_j)]_{i,j=1,\dots,N}.
\end{multline*}
Our claim reduces to verifying the equality
$$
U^{S,t}_{t+}\cdot  U^{S,t+1}_{S-}=U^{S,t}_{S-}\cdot U^{S-1,t}_{t+}.
$$
Note that this will also imply the coincidence of normalization
constants, since all matrices under consideration are stochastic.

A straightforward computation yields
$$
 U^{S,t}_{t+S-}=U^{S,t}_{t+}\cdot  U^{S,t+1}_{S-}=U^{S,t}_{S-}\cdot
 U^{S-1,t}_{t+},
$$
where
$$U_{t+S-}^{S,t}(x,y)=\begin{cases}
  u_1 ,& \text{ if } y=x+1,\\
  u_0,&\text{ if } y=x,\\
   u_{-1},& \text{ if } y=x-1,\\
  0,&\text{ otherwise,}
\end{cases}
$$
and
$$
 u_1=q^{T+N-1-t}(1-q^{x-S-N+1})
 (1-q^{x-S-N+2})\frac {(1-\kappa ^2q^{x-S-t+1})({1- \kappa ^2 q^{x-T+1}})}{(1-\kappa ^2
 q^{2x-t-S+2}))1- \kappa ^2
 q^{2x-t-S+1})},
$$
\begin{multline*}
 u_0= -(1-q^{x-S-N+1})\frac{1- \kappa ^2 q^{x+N-t}}{1- \kappa ^2
 q^{2x-t-S+1}}\\
 \times \left(q^{T-S-t}(1-q^{x+1}) \frac{1- \kappa ^2 q^{x-T+1}} {1-\kappa ^2
 q^{2x-t-S+2}}
 + (1-q^{x+T-t-S}) \frac {1-\kappa ^2q^{x-S-t}}{1-\kappa ^2
 q^{2x-t-S}} \right)  ,
\end{multline*}
$$
 u_{-1}=q^{-(S+N-1)} (1-q^{x+T-t-S})(1-q^x) \frac{(1- \kappa ^2q^{x+N-t})(1-\kappa ^2q^{x+N-t-1})} {(1- \kappa ^2
 q^{2x-t-S+1})(1-\kappa ^2
 q^{2x-t-S})}\,.
$$

\end{proof}

\section{Perfect sampling algorithm}

\label{Section_perfect_sampling_algorithm}

\subsection{Definition of transition matrices}

In this section we aim to define two new stochastic matrices
$$P_{S\mapsto S+1}^S(X,Y),\quad X\in\Omega(N,T,S),\quad
Y\in\Omega(N,T,S+1)$$ and
$$P_{S\mapsto S-1}^S(X,Y),\quad X\in\Omega(N,T,S),\quad
Y\in\Omega(N,T,S-1)$$ that preserve the measures $\mu(N,T,S,q,\kappa
)$. Both $P_{S\mapsto S+1}^S$ and $P_{S\mapsto S-1}^S$ depend on
parameters $N$, $T$, $q$, $\kappa $ but we omit these parameters
from the notation.

Suppose we are given a sequence $X=(X(0),X(1),\dots,X(T))\in
\Omega(N,T,S)$ (recall that $X(t)\in \mathcal X^{S,t}$). Below we
construct a random sequence $Y=(Y(0),\dots,Y(T))\in\Omega(N,T,S+1)$
and therefore define the transition probability (or, equivalently,
stochastic matrix) $P_{S\mapsto S+1}^S(X,Y)$.

First note that $Y(0)\in \mathcal X^{S+1,0}$ and $|\mathcal
X^{S+1,0}|=1$. Thus, $Y(0)$ is uniquely defined. We will perform a
\emph{sequential update}. Suppose $Y(0),Y(1),\dots,Y(t)$ have been
already defined. Define the conditional distribution of $Y(t+1)$ given
$X$, $Y(0),Y(1),\dots,Y(t)$ by
\begin{multline}
\label{Two_kinds}
 {\rm Prob}\{Y(t+1)=Z\}=\frac{P^{S+1,t}_{t+}(Y(t),Z)\cdot
 P^{S+1,t+1}_{S-}(Z,X(t+1))}{
 (P^{S+1,t}_{t+}P^{S+1,t+1}_{S-})(Y(t),X(t+1))}\\
 =\frac{P^{S,t+1}_{S+}(X(t+1),Z)\cdot P^{S+1,t+1}_{t-}(Z,Y(t))
 }{
 (P^{S,t+1}_{S+}P^{S+1,t+1}_{t-})(X(t+1),Y(t))}.
\end{multline}
(The second equality follows from $
\rho_{S+1,t+1}(X)P^{S+1,t+1}_{t-}(X,Y)=\rho_{S+1,t}(Y)P^{S+1,t}_{t+}(Y,X)$.)

This definition follows the idea of \cite[Section 2.3]{DF}, see
also \cite{BF}.

Observe that $(P^{S+1,t}_{t+}P^{S+1,t+1}_{S-})(Y(t),X(t+1))>0$ (here
and below see \cite{BG} for more details).

One could say that we choose $Y(t+1)$ using conditional distribution
of the middle point in the successive application of
$P^{S+1,t}_{t+}$ and $P^{S+1,t+1}_{S-}$ (or $P^{S,t+1}_{S+}$ and
$P^{S+1,t+1}_{t-}$ ), provided that we start at $Y(t)$ and finish at
$X(t+1)$ (or start at $X(t+1)$ and finish at $Y(t)$).

After performing $T$ updates we obtain the sequence $Y$.

Equivalently, define $P_{S\mapsto S+1}^S$  by
$$
 P_{S\mapsto S+1}^S(X,Y)=\begin{cases}
 \prod\limits_{t=0}^{T-1}\dfrac{P^{S+1,t}_{t+}(Y(t),Y(t+1))\cdot
 P^{S+1,t+1}_{S-}(Y(t+1),X(t+1))}{
 (P^{S+1,t}_{t+}P^{S+1,t+1}_{S-})(Y(t),X(t+1))},\\
  \quad \quad \text{if } \prod\limits_{t=0}^{T-1}
  (P^{S+1,t}_{t+}P^{S+1,t+1}_{S-})(Y(t),X(t+1))>0,\\
  0, \text{ otherwise.}
  \end{cases}
$$

\begin{theorem}
\label{Th_Mes_pr} The matrix $P^S_{S\mapsto S+1}$ on
$\Omega(N,T,S)\times\Omega(N,T,S+1)$ is stochastic.
 The transition probabilities $P_{S\mapsto S+1}^S(X,Y)$ preserve the measures
 $\mu(N,T,S,q,\kappa )$:
 $$
   \mu(N,T,S+1,q,\kappa )(Y)=\sum_{X\in\Omega(N,T,S)}P_{S\mapsto
   S+1}^S(X,Y)\mu(N,T,S,q,\kappa )(X).
 $$
\end{theorem}
\begin{proof}
See \cite{BG}.
\end{proof}

Similarly to $P_{S\mapsto S+1}$, one defines a transition matrix
$$
P_{S\mapsto S-1}^S(X,Y),\quad X\in\Omega(N,T,S),\quad
Y\in\Omega(N,T,S-1),
$$
by
$$
 P_{S\mapsto S-1}^S(X,Y)=\begin{cases}
 \prod\limits_{t=0}^{T-1}\dfrac{P^{S-1,t}_{t+}(Y(t),Y(t+1))\cdot
 P^{S-1,t+1}_{S+}(Y(t+1),X(t+1))}{
 (P^{S-1,t}_{t+}P^{S-1,t+1}_{S+})(Y(t),X(t+1))},\\
  \quad \quad \text{if } \prod\limits_{t=0}^{T-1}
  (P^{S-1,t}_{t+}P^{S-1,t+1}_{S+})(Y(t),X(t+1))>0,\\
  0, \text{ otherwise.}
  \end{cases}
$$
Similarly to \eqref{Two_kinds} there is another way to write
$P^S_{S\mapsto S-1}$ because of the equality
\begin{multline*}
 \dfrac{P^{S-1,t}_{t+}(Y(t),Y(t+1))\cdot
 P^{S-1,t+1}_{S+}(Y(t+1),X(t+1))}{
 (P^{S-1,t}_{t+}P^{S-1,t+1}_{S+})(Y(t),X(t+1))}\\=
 \dfrac{
 P^{S,t+1}_{S-}(X(t+1),Y(t+1))\cdot P^{S-1,t+1}_{t-}(Y(t+1),Y(t))}{
 (P^{S,t+1}_{S-}P^{S-1,t+1}_{t-})(X(t+1),Y(t))}
\end{multline*}

Similarly to Theorem \ref{Th_Mes_pr} one proves the following claim.

\begin{theorem}
\label{Th_Mes_pr2}
 The matrix $P^S_{S\mapsto S-1}$ on
$\Omega(N,T,S)\times\Omega(N,T,S-1)$ is stochastic.
 The transition probabilities $P_{S\mapsto S-1}^S(X,Y)$ preserve the measures
 $\mu(N,T,S,q,\kappa )$:
 $$
   \mu(N,T,S-1,q,\kappa )(Y)=\sum_{X\in\Omega(N,T,S)}P_{S\mapsto
   S-1}^S(X,Y)\mu(N,T,S,q,\kappa )(X).
 $$
\end{theorem}

\begin{Remark}
The above construction performs a sequential update from $t=0$ to $t=T$. One can
equally well update from $t=T$ to $t=0$ by suitably modifying the definitions. The resulting Markov
chains also preserve the measures $\mu(N,T,S,q,\kappa )$, and they are different from the Markov
chains defined above.
\end{Remark}

\subsection{Algorithm for the $S\mapsto S+1$ step.}

 Now we suggest an algorithmic description of the Markov
chain from the previous section.

Denote
$$
  p(x,t,q,\kappa ,S,T)=\frac{1-q^{x+T-t-S-1}}{q^{T-t-S-1} (1-q^{x+1})} \frac {1- \kappa ^2q^{x-S-t-1}}{1- \kappa ^2 q^{x-T+1}}   \frac{1- \kappa ^2
 q^{2x-t-S+1}} {1-\kappa ^2
 q^{2x-t-S-1}}$$

and
 \begin{multline*}
 P(x,t,q,\kappa ,S,T;k)=\prod_{i=1}^k p(x+i-1,t,q,\kappa ,S,T)\\
 =\frac{(q^{x+T-t-S-1};q)_k}{q^{k(T-t-S-1)} (q^{x+1};q)_k} \frac {(\kappa ^2q^{x-S-t-1};q)_k}
 {(\kappa ^2 q^{x-T+1};q)_k}   \frac{(\kappa ^2
 q^{2x-t-S+1};q^2)_k} {(\kappa ^2
 q^{2x-t-S-1};q^2)_k}
 \,.
 \end{multline*}

Denote by $D(x,t,S;n)$ (it also depends on $q,\kappa,T$, but we omit these parameters) the
probability distribution on $\{0,1,\dots,n\}$ given by
\begin{equation}
\label{Distribution}
 {\rm
Prob}(\{k\})=D(x,t,S;n)\{k\}=\frac{P(x,t,q,\kappa ,S,T;k)}
           {\sum_{j=0}^n{P(x,t,q,\kappa ,S,T;j)}
           }\,.
\end{equation}

 Suppose we are
given $X=(X(0),X(1),\dots,X(T))\in\Omega(N,T,S)$. We want to
construct $Y=(Y(0),Y(1),\dots,Y(T))\in\Omega(N,T,S+1)$.

In the first place we note that $Y(0)$ is uniquely defined,
$$
 Y(0)=(0,1,\dots,N-1).
$$
Then we perform $T$ sequential updates, i.e., for $t=0,1,\dots T-1$
we construct $Y(t+1)$ using $Y(t)$ and $X(t+1)$. Let us describe
each step.

Let $Y(t)=(y_1<y_2<\dots<y_N)$ and $X(t+1)=(x_1<x_2<\dots<x_N)$. We
are going to construct $Y(t+1)=(z_1<z_2<\dots<z_N)$.

Recall that
$$z_i\in\mathfrak X^{S+1,t+1}=\{x\in\mathbb Z\mid \max(0,t+S-T+2)\le x \le \min(t+N,S+N)\}.$$

Observe that $Y(t)$ and $X(t+1)$ satisfy $(P^{S+1,t}_{t+}P^{S+1,t+1}_{S-})(Y(t),X(t+1))>0$. This
implies that $x_i-y_i$ is equal to either $-1$, $0$ or $1$ for every $i$.

$\bullet$ First, consider all indices $i$ such that $x_i-y_i=1$. For
every such $i$ we set $z_i=x_i$.

$\bullet$ Second, consider all indices $i$ such that $x_i-y_i=-1$
and set $z_i=y_i$.

$\bullet$ Finally, consider all remaining indices, i.e., all $i$ such
that $x_i=y_i$. Divide the corresponding $x_i$'s into blocks of
neighboring integers of distance at least one from each other. Call
such a block a $(k,l)$-block, where $k$ is the smallest number in
the block and $l$ is its size. Thus, we have
$$x_i=y_i=k,\quad x_{i+1}=y_{i+1}=k+1,\quad \dots,\quad x_{i+l-1}=y_{i+l-1}=k+l-1$$
and $$ y_{i-1}<k-1,\quad y_{i+l}>k+l.$$

For each $(k,l)$-block we perform the following procedure: consider a random variable $\xi$
distributed according to ${D(k,t,S;l)}$ ($\xi$'s corresponding to different $(k,l)$-blocks are
independent). Set $z_i=x_i$ for the first $\xi$ integers of the block (their coordinates are
$k,k+1,\dots, k+\xi-1$) and set $z_i=x_i+1$ for the rest of the block.

\medskip

At Figure 4 we provide an example of constructing $Y(t+1)$ using
$X(t+1)$ and $Y(t)$: there is only one $(k,l)$-block and it splits
into two groups, here $\xi=2$.

\begin{center}
 {\scalebox{0.5}{\includegraphics{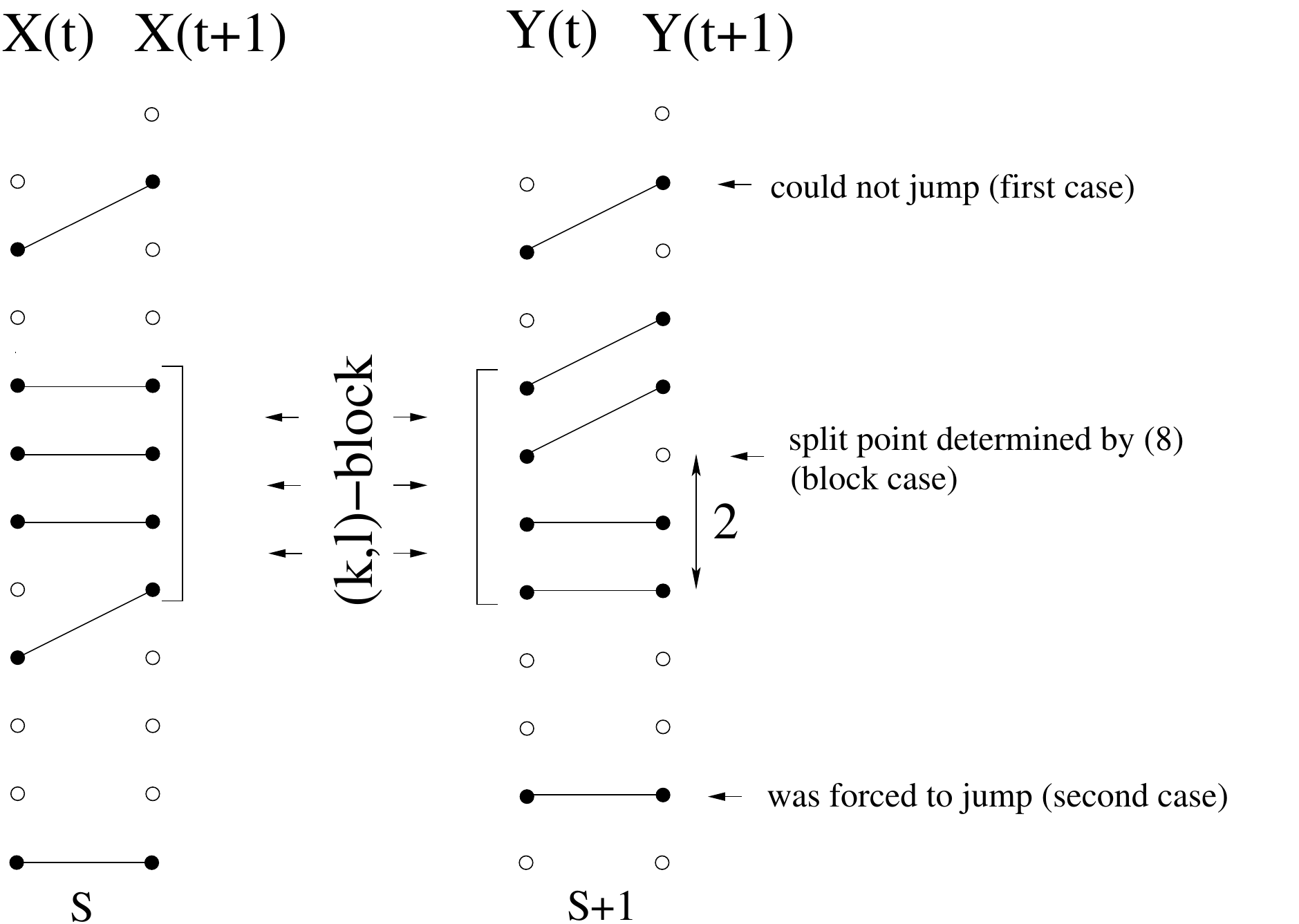}}}

 Figure 4. Example of $(k,l)$-block split, $l=4$, $\xi=2$.
\end{center}

\begin{theorem}
 \label{Alg1}
 The algorithm described above is precisely the $S\mapsto S+1$ Markov
 step given by $P^S_{S\mapsto S+1}$.
\end{theorem}

\begin{proof}
Straightforward computations. See \cite{BG} for some details.

\end{proof}

{\bf Remarks.} Setting $\kappa =0$ in the formulas for the
distribution ${D(x,t,S;n)}$ we obtain the perfect sampling algorithm
for boxed plane partitions distributed as $q^{-volume}$.

Sending $q\to 1$ in the formulas for the distribution ${D(x,t,S;n)}$ as described in Section
\ref{Section_pr_models}, we get a perfect sampling algorithm for the Racah case (recall that in
this case the weight of a horizontal lozenge is proportional to a linear function of its vertical
coordinate).

\subsection{Algorithm for $S\mapsto S-1$ step}

Using similar methods we can also obtain $S\mapsto S-1$ Markov step which gives alternative way to
sample a random tiling: We start from the case $T=S$ and then perform some amount of $S\mapsto S-1$
steps.

The $S\mapsto S-1$ step algorithm is very similar to the $S\mapsto S+1$ one.

Denote
$$
  \hat p(x,t,q,\kappa
,S,T,N)=\frac{q^{t+1-S}(1-q^{x-t-N-1})}{(1-q^{x-S-N+1})} \frac {1-
\kappa ^2q^{x+N-t-1}}{1-\kappa ^2
 q^{x+N-S+1}}   \frac{1-\kappa ^2
 q^{2x-t-S+1}}{1- \kappa ^2 q^{2x-t-S-1}}
$$

and
 \begin{multline*}
 \hat P(x,t,q,\kappa ,S,T,N;k)=\prod_{i=1}^k \hat p(x+i-1,t,q,\kappa ,S,T,N)\\
 =\frac{q^{k(t+1-S)}(q^{x-t-N-1};q)_k}{(q^{x-S-N+1};q)_k} \frac {(\kappa ^2q^{x+N-t-1};q)_k}
 {(\kappa ^2
 q^{x+N-S+1};q)_k}   \frac{(\kappa ^2
 q^{2x-t-S+1};q^2)_k}{(\kappa ^2 q^{2x-t-S-1};q^2)_k}
 \,.
 \end{multline*}

Denote by $\hat D(x,t,S;n)$ the probability distribution on $\{0,1,\dots,n\}$ given by
\begin{equation}
\label{Distribution2}
 {\rm
Prob}(\{k\})=\hat D(x,t,S;n)\{k\}=\frac{\hat P(x,t,q,\kappa ,S,T,N;k)}
           {\sum_{j=0}^n{\hat P(x,t,q,\kappa ,S,T,N;j)}
           }\,.
\end{equation}

Suppose we are given $X=(X(0),X(1),\dots,X(T))\in\Omega(N,T,S)$. We want to construct
$Y=(Y(0),Y(1),\dots,Y(T))\in\Omega(N,T,S-1)$.

As above, note that $Y(0)$ is uniquely defined,
$$
 Y(0)=(0,1,\dots,N-1).
$$
Then we  again perform $T$ sequential updates, i.e., for $t=0,1,\dots T-1$ we construct $Y(t+1)$
using $Y(t)$ and $X(t+1)$. Let us describe each step.

Let $Y(t)=(y_1<y_2<\dots<y_N)$ and $X(t+1)=(x_1<x_2<\dots<x_N)$. We are going to construct
$Y(t+1)=(z_1<z_2<\dots<z_N)$.

Recall that
$$z_i\in\mathfrak X^{S-1,t+1}=\{x\in\mathbb Z\mid \max(0,t+S-T)\le x \le \min(t+N,S+N-2)\}.$$

$Y(t)$ and $X(t+1)$ satisfy $(P^{S-1,t}_{t+}P^{S-1,t+1}_{S+})(Y(t),X(t+1))>0$. This implies that
$x_i-y_i$ is equal to either $0$, $1$ or $2$ for every $i$.

$\bullet$ First, consider all indices $i$ such that $x_i-y_i=0$. For every such $i$ we set
$z_i=x_i$.

$\bullet$ Second, consider all indices $i$ such that $x_i-y_i=2$ and set $z_i=y_i+1$.

$\bullet$ Finally, consider all remaining indices, i.e., all $i$ such that $x_i=y_i+1$. Divide the
corresponding $x_i$'s into blocks of neighboring integers of distance at least one from each other.
Call such a block a $(k,l)'$-block, where $k$ is the smallest number in the block and $l$ is its
size. Thus, we have
$$x_i=y_i+1=k,\quad x_{i+1}=y_{i+1}+1=k+1,\quad \dots,\quad x_{i+l-1}=y_{i+l-1}=k+l-1.$$

For each $(k,l)'$-block we perform the following procedure: consider random variable $\xi$
distributed according to $\hat D(k,t,S;l)$ ($\xi$'s corresponding to different $(k,l)'$-blocks are
independent). Set $z_i=y_i$ for the first $\xi$ integers of the block (their coordinates are
$k-1,k,\dots, k+\xi-2$) and set $z_i=y_i+1$ for the rest of the block.

\begin{theorem}
 The algorithm described above is precisely $S\mapsto S-1$ Markov
 step defined by $P^S_{S\mapsto S-1}$.
\end{theorem}
The proof is similar to Theorem \ref{Alg1}.

\subsection{Markov evolution of the top path}
\label{Section_Mark_chain_top_path}

The $S\mapsto S+1$ Markov step described in the previous section has
the following property: Its projection to the set of topmost
horizontal lozenges (or the topmost holes in terms of
nonintersecting paths and point configurations) is also a Markov
chain. This Markov chain is an exclusion type process. Let us
describe it.

\medskip
The general setting is as follows. The state space of our discrete
time Markov chain consists of semi-infinite particle configurations
$\{e_1<e_2<e_3<\dots\}$ in $\mathbb Z$. At each time moment every
particle either stays or jumps to the left (any distance)
avoiding collisions and jumps over neighbors. Jumps are performed
sequentially. First, the leftmost particle ($e_1$) jumps, then the
second one and so on. The distribution $D$ of the length of the jump
of a particle depends on the number of the particle, moment of time,
current position of the particle ($e_i$) and the distance between
the current position of the particle and the position of the
previous particle ($e_{i-1}$) in the next moment of time. At time
$0$ we have the step initial condition, i.e., $e_i=i+const$.

\medskip
Now let us turn back to our situation. All particles are enumerated
by the parameter $t$ and our time parameter is $S$ that changes from
$0$ to $T$. Consider a sequence $\{u_t^S\}_{t=1,\dots}$, where
$u_t^S$ is the vertical coordinate (in our notation - $x$)
corresponding to the topmost hole inside the hexagon for $t\le S$
and $u_t^S=N+t-1$ for $t>S$. (We can also view $u_t^S$ as the
vertical coordinate corresponding to the $t$th hole, if we count all
holes, not just the ones inside a hexagon, starting from the line
$x=0$.)

The evolution of $\{u_t^S\}$ is precisely our Markov process. When
$S=0$ the configuration consists of points $N,N+1,N+2,\dots$. The
distribution of the length of jump of the particle with coordinate
$u_t^S$ at the time moment $S$ is given by the distribution
${D(u_{t-1}^{S+1}+1,t,S;u_{t}^S-u_{t-1}^{S+1}-1)}$ (see
\eqref{Distribution} for the definition).

Note that when $S=T$ the configuration consists of points $0,1,2,\dots$.

\medskip

We can also obtain in a similar way a Markov chain for the
bottommost holes.

Finally, we may construct two more similar processes using the
$S\mapsto S-1$ Markov chain (for this chains the direction of
particle jumps changes and distributions $D$ are replaced by
distributions $\hat D$).

\section{Correlation kernel}
\label{Section_Corr_kernel}

The aim of this section is to obtain the formulas for the
correlation functions of random point configurations in $\mathbb
Z^2$ obtained from the random tilings we are interested in.
\smallskip

\subsection{Expression via orthogonal polynomials}

Recall that a tiling of a hexagon corresponds to some family of nonintersecting paths that can be
viewed as a point configuration in $\mathbb Z^2$. Let us denote this configuration by $\mathbb M$.

\begin{center}
 {\scalebox{0.6}{\includegraphics{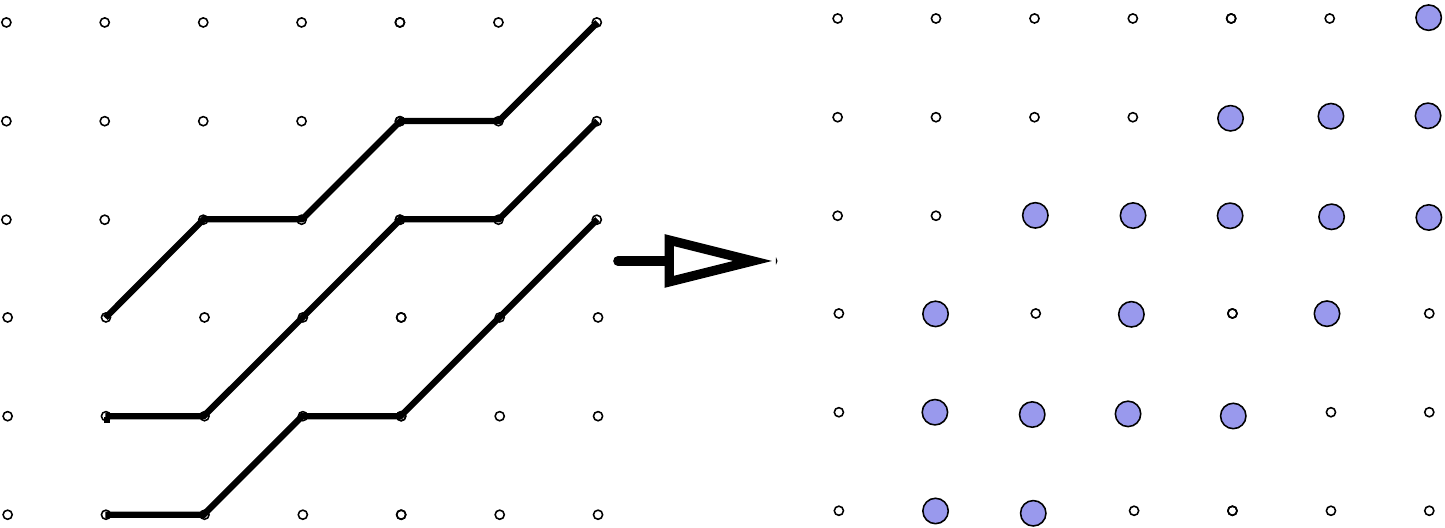}}}
\end{center}

As above, we denote the horizontal coordinate by $t$ and the
vertical coordinate by $x$.

We want to compute the correlation functions of this random point configuration.

Recall that the $n$th correlation function is defined by
$$
 \rho_n(t_1,x_1;\dots;t_n,x_n)=  {\rm Prob}\{(t_1,x_1)\in\mathbb M,\dots,(t_n,x_n)\in \mathbb M \}
$$
for any collection $\{(t_i,x_i)\}_{i=1,\dots,n}$ of distinct points in $\mathbb Z^2$.

To compute the correlation functions $\rho_n$ we are going to use a variant of the Eynard-Mehta
theorem (see \cite{EM} and \cite[Section 7.4]{BO}). Let us state it first.

\begin{proposition}
\label{proposition_EM} Assume that for every time moment $t$ we are given
 an orthonormal system $\{f^t_n\}_{n\ge 0}$ in
$l_2(\{0,1, \dots,L\})$  and a set of numbers
 $c_0^t,c_1^t, \dots$. Denote
$$
v_{t,t+1}(x,y)=\sum_{n\ge 0}c_n^t f^t_n(x)f^{t+1}_n(y).
$$

 Assume also that we are given a discrete time  Markov process
${\cal P}_t$ taking values in $N$-tuples of elements of the set $\{0,1, \dots,L\}$, with
one-dimensional distributions
$$\Bigl(\det\left[f_{i-1}^t(x_j)\right]_{i,j=1,\dots ,N}\Bigr)^2$$
 and transition probabilities

 $$\frac{\det\left[v_{t,t+1}(x_i,y_j)\right]_{i,j=1,\dots ,N}
 \det\left[f_{i-1}^{t+1}(y_j)\right]_{i,j=1,\dots ,N}} {\det\left[f_{i-1}^t(x_j)\right]_{i,j=1,\dots ,N}
  \prod\limits_{n=0}^{N-1}c_n^t}. $$

  Then
$$
 {\rm Prob}\{x_1\in {\cal P}_{k_1},\dots, x_n\in {\cal P}_{k_n}\}\\
 =\det\left[K(k_i,x_i;k_j,x_j)\right]_{i,j=1,\dots ,n},
$$ where
\begin{gather*}
K(k,x;l,y)=\sum_{i=0}^{N-1}\frac{1}{c_i^{l,k}}f_i^k(x)f_i^l(y)
 ,\, k\ge l; \\ K(k,x;l,y)=-\sum_{i\ge N}c_i^{k,l} f_i^k(x)f_i^l(y),
 \,k<l; \\ c_i^{k,k}=1,\, c_i^{k,l}=c_i^k\cdot c_i^{k+1}\cdot\dots\cdot
 c_i^{l-1}.
 \end{gather*}
 \end{proposition}

\begin{theorem}
 The Markov process $X(t)$ meets the assumptions of Proposition
 \ref{proposition_EM}.
\end{theorem}
The Markov process ${\cal P}_t$ is precisely our Markov process
$X(t)$. The orthonormal functions $f^t_n(x)$ are the normalized
$q$-Racah polynomials multiplied by the square root of their weight
function (see Section \ref{Section_distributions} for the definition
of $q$-Racah polynomials, their weight function, and the
correspondence between parameters of these polynomials and our
parameters $t,q,\kappa ,N,T,S$):
\begin{equation}
\label{EM_Orthonormal_system}
 f^t_n(x)=\sqrt{w_{t,S}(x)}\frac{R_n^t(x)}{\sqrt{(R_n^t,R_n^t)}}\,,
\end{equation}
where $(R_n^t,R_n^t)$ is the squared norm of the $q$-Racah
polynomials with respect to the weight function $w_{t,S}(x)$. This
norm can be obtained from the norm of the $q$-Racah polynomials
provided in $\cite{KS}$ ($w_{t,S}(x)$ differs from the weight
function of \cite{KS} by a factor not depending on $x$). The
explicit formula is a little bit different in the four cases of
correspondence between parameters of polynomials and
$t,q,\kappa,N,T,S$. For instance, in the case given by formula
\eqref{QRacah_case1}:
\begin{multline*}
(R_n^t,R_n^t)=\frac{(-1)^{t+S}(q^{-2N-T+2},k^{-2}q^{S-N};q)_{t+N-1}}{({\kappa^{-2}}{q^{-2N+1}},q^{S-T-N+1},\kappa^2
q^{-t-S+2},q^{-t-N+1};q)_{t+N-1} }\\
\times \frac{(1-q^{-T-2N})(q,q^{-T-N+t+1},{\kappa^{-2}}{q^{-2N+1}},q^{S-T-N};q)_n}
{(1-q^{-2N-T+2n+1})(q^{-S-N},q^{-2N-T+1},\kappa^2 q^{-T},q^{-t-N+1};q)_n }
\\
\times \frac{\kappa^{2n} q^{-n(S+t+1)}} {(q;q)_{T-S-t} (q^{-S-N+1};q)_{S+N-1} (\kappa^2
q^{-T+1};q)_{T+N-t}}.
\end{multline*}
However, this long formula is not important for us, since factors involving it always cancel out.
In particular, in the case given by \eqref{QRacah_case1} the quotient
$(R_n^{t+1},R_n^{t+1})/(R_n^t,R_n^t)$, that is crucial for us, is simply
$$
 \frac{(R_n^{t+1},R_n^{t+1})}{(R_n^t,R_n^t)}=-\frac{(1-q^{T+N-t-n-1})(1-q^{-t-N+n})}{(1-q^{-t-N})^2}\,.
$$

The constants $c_n^t$ are given by
\begin{equation}
\label{EM_numbers}
 c_n^t=\sqrt{(1-q^{-N-t+n})(1-q^{T+N-t-n-1})}.
\end{equation}
\begin{proof}

Theorem \ref{Theorem_One_dim_distribution} yields
$$
 {\rm Prob}\{X(t)=(x_1,x_2,\dots,x_N)\} = \frac1Z \prod_{i<j}(\mu_{t,S}(x_i)-\mu_{t,S}(x_j))^2 \prod_{i=1}^N
 w_{t,S}(x_i).
$$
On the other hand
$$
\det\left[f_{i-1}^t(x_j)\right]_{i,j=1,\dots
,N}=const\cdot\prod_{i=1}^N
 \sqrt{w_{t,S}(x_i)} \det\left[R_{i-1}(x_j)\right]_{i,j=1,\dots ,N}.
$$
The last determinant is a Vandermonde determinant in variables $\mu_{t,S}(x_j)$, hence
$$
\Bigl(\det\left[f_{i-1}^t(x_j)\right]_{i,j=1,\dots ,N}\Bigr)^2=\frac1{Z'}
\prod_{i<j}(\mu_{t,S}(x_i)-\mu_{t,S}(x_j))^2 \prod_{i=1}^N
 w_{t,S}(x_i).
$$
Coincidence of the constants ($Z=Z'$) follows from the fact that the
left-hand side in the last equality defines a probability
distribution.

Thus, the one-dimensional distributions of our process have the required form.

Next, we need the following standard facts.

\begin{lemma} The following relation for the basic hypergeometric function holds:
\label{Lemma_qRacah_relation}
\begin{multline}
\label{phi_relation_1}
(c-w)(1-d)_4\phi_3\left(\genfrac{}{}{0pt}{}{a,b,c,qd}{u,v,qw}\biggr|q;q\right)+(w-d)(1-c)_4\phi_3\left(\genfrac{}{}{0pt}{}{a,b,qc,d}{u,v,qw}\biggr|q;q\right)
\\=(c-d)(1-w)_4\phi_3\left(\genfrac{}{}{0pt}{}{a,b,c,d}{u,v,w}\biggr|q;q\right).
\end{multline}
In terms of $q$-Racah polynomials, this relation can be rewritten  as
\begin{multline}
 \label{qRacah_relation_case1}
 (q^{-x}-q\gamma)(1-\gamma\delta q^{x+1}) R_n\biggl(\mu(x); \alpha, \beta, q\gamma, \delta \mid q\biggr)
 \\+ (q\gamma-\gamma\delta q^{x+1})(1-q^{-x})R_n\biggl(\mu(x-1); \alpha, \beta, q\gamma, \delta \mid q\biggr)
 \\= (q^{-x}-\gamma\delta q^{x+1})(1-q\gamma)R_n\biggl(\mu(x); \alpha, \beta, \gamma, \delta \mid q\biggr)
\end{multline}
or as
\begin{multline}
 \label{qRacah_relation_case2}
 (q^{-x}-q\alpha)(1-\gamma\delta q^{x+1}) R_n\biggl(\mu(x); q\alpha, q^{-1}\beta, \gamma, q\delta \mid q\biggr)
 \\+ (q\alpha-\gamma\delta q^{x+1})(1-q^{-x})R_n\biggl(\mu(x-1); q\alpha, q^{-1}\beta, \gamma, q\delta \mid q\biggr)
 \\= (q^{-x}-\gamma\delta q^{x+1})(1-q\alpha)R_n\biggl(\mu(x); \alpha, \beta, \gamma, \delta \mid
 q\biggr),
\end{multline}

where $R_n$ is given by (\ref{qRacah}).

\smallskip

For the balanced terminating $_4\phi_3\left(\genfrac{}{}{0pt}{}{a,b,c,d}{u,v,w}\biggr|q;q\right)$
(i.e., one of $a$, $b$, $c$ or $d$ equals $q^{-n}$ and $a\cdot b\cdot c\cdot d=u\cdot v\cdot w$)
we also have the following relation:

\begin{multline}
\label{phi_relation_2}
(c-u)(1-vc^{-1})(wq^{-1}-1)_4\phi_3\left(\genfrac{}{}{0pt}{}{a,b,q^{-1}c,d}{u,v,q^{-1}w}\biggr|q;q\right)\\
+(u-d)(1-vd^{-1})(wq^{-1}-1)_4\phi_3\left(\genfrac{}{}{0pt}{}{a,b,c,q^{-1}d}{u,v,q^{-1}w}\biggr|q;q\right)
\\=(c-d)(wq^{-1}-b)(1-aqw^{-1})_4\phi_3\left(\genfrac{}{}{0pt}{}{a,b,c,d}{u,v,w}\biggr|q;q\right).
\end{multline}

In terms of $q$-Racah polynomials, the last relation can \blue{be} rewritten as

\begin{multline}
 \label{qRacah_relation_case3}
 (q^{-x}-q\gamma)(1-\beta\delta q^{x+1})(\alpha -1) R_n\biggl(\mu(x+1); q^{-1}\alpha, q\beta, \gamma, q^{-1}\delta \mid q\biggr)
 \\+ (q\gamma-\gamma\delta q^{x+1})(1-q^{-x}\beta\gamma^{-1})(\alpha -1)R_n\biggl(\mu(x); q^{-1}\alpha, q\beta, \gamma, q^{-1}\delta \mid q\biggr)
 \\= (q^{-x}-\gamma\delta q^{x+1})(\alpha-\alpha\beta q^{n+1})(1-\alpha^{-1}q^{-n})R_n\biggl(\mu(x); \alpha, \beta, \gamma, \delta \mid q\biggr)
\end{multline}
or as
\begin{multline}
 \label{qRacah_relation_case4}
 (q^{-x}-q\alpha)(1-\beta\delta q^{x+1})(\gamma -1) R_n\biggl(\mu(x+1); \alpha, \beta, q^{-1}\gamma, \delta \mid q\biggr)
 \\+ (q\alpha-\gamma\delta q^{x+1})(1-q^{-x}\beta\gamma^{-1})(\gamma -1)R_n\biggl(\mu(x); \alpha, \beta, q^{-1}\gamma, \delta \mid q\biggr)
 \\= (q^{-x}-\gamma\delta q^{x+1})(\gamma-\alpha\beta q^{n+1})(1-\gamma^{-1}q^{-n})R_n\biggl(\mu(x); \alpha, \beta, \gamma, \delta \mid
 q\biggr),
\end{multline}

\end{lemma}
\begin{proof}
To prove the first relation for the basic hypergeometric function we expand $_4\phi_3$ into series
in $q$ and perform straightforward computations in every term.

To obtain \eqref{qRacah_relation_case1} and \eqref{qRacah_relation_case2} we simply rewrite the
relation \eqref{phi_relation_1} in terms of $q$-Racah polynomials using their definition
\eqref{qRacah}.

Next, we observe that $q$-Racah polynomials form an orthogonal
basis in the corresponding $l_2$ space. Consequently, we can write
dual relations for \eqref{qRacah_relation_case1} and
\eqref{qRacah_relation_case2} and these are precisely
\eqref{qRacah_relation_case3} and \eqref{qRacah_relation_case4}. It
is easily seen that two last relations are equivalent to just one
relation for basic hypergeometric function \eqref{phi_relation_2}.
\end{proof}

Using the last lemma we obtain the following one:
\begin{lemma}
\label{Lemma_qRacah_relation_funct}
\begin{multline*}
v_{t,t+1}(x,y)=\sum_{n\ge 0}c_n^t
f^t_n(x)f^{t+1}_n(y)=\sqrt\frac{w_{t,S}(x)}{w_{t+1,S}(y)}\left(\delta_{x+1}^y
 w_1(x)+\delta_{x}^{y} w_0(x)\right),
\end{multline*}
where
$$
 w_0(x)=-(1-q^{x+T-t-S})\frac{1- \kappa ^2q^{x+N-t}} {1- \kappa ^2
 q^{2x-t-S+1}},
$$

$$
 w_1(x)=q^{T+N-1-t}(1-q^{x-S-N+1}) \frac{1- \kappa ^2 q^{x-T+1}} {1- \kappa ^2
 q^{2x-t-S+1}},
$$
while $w_{t,S}(x)$ stands for the weight function corresponding to the parameters $t,q,\kappa
,N,T,S$ (see Section \ref{Section_distributions} and Theorem \ref{Theorem_One_dim_distribution} for
details).
\end{lemma}
\begin{proof}

First, we substitute the parameters of $q$-Racah polynomials given by formulas
\eqref{QRacah_case1}-\eqref{QRacah_case4} into the statement of Lemma \ref{Lemma_qRacah_relation}.

We use \eqref{qRacah_relation_case1}, \eqref{qRacah_relation_case2}
in  cases \eqref{QRacah_case1}, \eqref{QRacah_case2},  and we use
\eqref{qRacah_relation_case3}, \eqref{qRacah_relation_case4} in
cases \eqref{QRacah_case3}, \eqref{QRacah_case4}.

In all 4 cases we rewrite the corresponding relation in terms of orthogonal functions $f_n^t(x)$
and get the following
\begin{equation}
 c_n^t f^{t+1}_n(y)=\sqrt\frac{w_{t,S}(y-1)}{w_{t+1,S}(y)}f^t_n(y-1)
 w_1(y-1)+\sqrt\frac{w_{t,S}(y)}{w_{t+1,S}(y)}f^t_n(y) w_0(y)
\end{equation}

 Multiply the last relation by $f^t_n(x)$ and sum over all meaningful $n$.

Since functions $f^t_n(y)$ form an orthonormal basis in the
corresponding $l_2$ space,
$$
 \sum_n f^{t}_n(x)f^{t}_n(y)=\delta_x^y,
$$
and the needed relation follows.
\end{proof}

Proposition \ref{Proposition_det_form_transitional_probabilities}
implies that the transition probabilities $P^{S,t}_{t+}(X,Y)$ have a
determinantal form. The last lemma yields that this form is exactly
the one required for the application of Proposition
\ref{proposition_EM}. Thus, the theorem is proved.
\end{proof}

Applying Proposition \ref{proposition_EM} for the process $X(t)$ we
obtain the following statement.
\begin{theorem}
\label{Theorem_Cor_functions}
$$
 \rho_n(t_1,x_1;\dots;t_n,x_n)= \det\left[K(k_i,x_i;k_j,x_j)\right]_{i,j=1,\dots ,n},
$$ where
\begin{gather*}
K(k,x;l,y)=\sum_{i=0}^{N-1}\frac{1}{c_i^{l,k}}f_i^k(x)f_i^l(y)
 ,\, k\ge l; \\ K(k,x;l,y)=-\sum_{i\ge N}c_i^{k,l} f_i^k(x)f_i^l(y),
 \,k<l; \\ c_i^{k,k}=1,\, c_i^{k,l}=c_i^k\cdot c_i^{k+1}\cdot\dots\cdot
 c_i^{l-1}
\end{gather*}
and functions $f_i^k(x)$ and numbers $c_i^t$ are given by the
formulas \eqref{EM_Orthonormal_system} and \eqref{EM_numbers}.
\end{theorem}

\subsection{Inverse Kasteleyn matrix}

Let us present another way to view the correlation kernel derived in the previous section.

Recall that we deal with lozenge tilings of a hexagon. Divide every
lozenge into two unit triangles and color the resulting triangles
into black and white (west triangle is black). In this way a tiling
turns into a perfect matching of the part of the dual hexagonal
lattice that fits in our hexagon. Correlation functions of the
perfect matchings can be computed using Kasteleyn's theorem (see
\cite{Ka}). Let us describe it.

Associate to every triangle the midpoint of its vertical side. Note
that in this way both black and white triangles can be parameterized
by the points of the two-dimensional lattice. Thus, we can use our
usual coordinates $(t,x)$ for the triangles.

\begin{center}

\scalebox{0.7}{\includegraphics{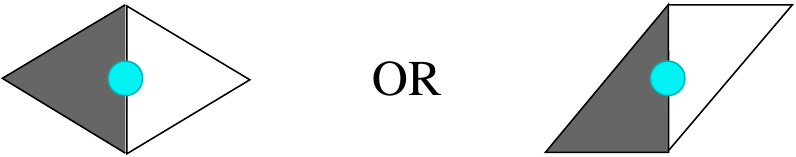}}

\end{center}

The Kasteleyn matrix ${\rm Kast} (t,x;r,y)$ is a weighted adjacency matrix. Here $(t,x)$ stand for
the coordinates of a white triangle and $(r,y)$ stand for the coordinates of a black triangle. In
our case,
$$
 {\rm Kast} (t,x;r,y)=\begin{cases}
  \kappa q^{-S/2+x-t/2+1/2}-\dfrac1{\kappa q^{-S/2+x-t/2+1/2}},&  (t,x)=(r,y),\\
  1,& (t,x)=(r-1,y-1), \\
  1,& (t,x)=(r-1,y),\\
  0,& \text{otherwise.}
  \end{cases}
$$

%
%
%
%

Set
\begin{multline}
\label{eq_conjug_factor}
g(t,x)=\frac{1}{\sqrt{w_{t,S}(x)}}\frac{(-1)^{t+x} \kappa^{-t}q^{x(T+N-t-1)+t(S/2-1/2)+t(t+1)/4}
(1- \kappa ^2q^{2x-t-S+1})}{ (q^{-1};q^{-1})_{S+N-1-x} (q;q)_{T-S+x-t} (\kappa^2
q^{x-T+1};q)_{T+N-t}}.
\end{multline}

\begin{theorem}
 Consider $n$ lozenges enumerated by pairs of triangles $((t_i,x_i),(r_i,y_i))$. The
 probability that a random tiling contains these lozenges equals
 $$
 \prod_{i=1}^n {\rm Kast} (t_i,x_i;r_i,y_i) \cdot \det \left[{K^{ext}}
 (r_i,y_i;t_j,x_j)\right]_{i,j=1,\dots,n},
 $$
 where
$$
 K^{ext}(r,y;t,x)=\frac {g(r,y)}{g(t,x)} \left(\delta_{(r,y)}^{(t,x)}- K(r,y;t,x)\right),
$$
the function $g$ is given by \eqref{eq_conjug_factor}, and $K(r,y;t,x)$ is given in Theorem
 \ref{Theorem_Cor_functions}.
\end{theorem}

\begin{proof}

Kasteleyn's theorem states that the probability to find lozenges
$$
((t_1,x_1);(r_1,y_1)),\dots,((t_n,x_n);(r_n,y_n))
$$
can be expressed via the inverse of the Kasteleyn matrix:
$$
 \prod_{i=1}^n {\rm Kast} (t_i,x_i;r_i,y_i)
  \cdot \det \left[{\rm Kast^{-1}}(r_i,y_i;t_j,x_j)\right]_{i,j=1,\dots,n}
$$

We can compare this statement with Theorem \ref{Theorem_Cor_functions}. Note that Theorem
\ref{Theorem_Cor_functions} describes the correlation functions of the particles. Consequently the
correlation kernel
$$
  \hat K(t,x;r,y)=\delta_{(t,x)}^{(r,y)}-K(t,x;r,y)
$$
(where $K(t,x;r,y)$ is the correlation kernel of Theorem
\ref{Theorem_Cor_functions}) is the correlation kernel of holes or,
equivalently of horizontal lozenges. (See, e.g., \cite[Appendix
A.3]{BOO} for some details on the particle-hole involution.)

Since the correlation kernel appears only in a determinant, it is only
determined up to conjugation: the transformation
$$
 \hat K(t,x;r,y)\mapsto \frac {g(t,x)}{g(r,y)} \hat K(t,x;r,y)
$$
does not change correlation functions. We conclude that
$$
 \frac{\hat K(t,x;r,y)}{\kappa q^{-S/2+x-t/2+1/2}-(\kappa q^{-S/2+x-t/2+1/2})^{-1}}
$$
should be (perhaps, after some conjugation) the inverse Kasteleyn matrix.

Let us verify this fact and find the appropriate conjugation factor.

We have
\begin{equation}
 \label{eq_kast_inv}
 \sum_{(h,z)} \frac{g(t,x)}{g(h,z)} \frac{\hat K(t,x;h,z)}{\kappa q^{-S/2+x-t/2+1/2}-(\kappa q^{-S/2+x-t/2+1/2})^{-1}}
 {\rm Kast} (h,z;r,y) = \delta_{(t,x)}^{(r,y)}\,.
\end{equation}

First, suppose that $t<r-1$. In this case the $((t,x);(r,y))$ matrix element of the right-hand side
is zero while the one of the left-hand side of \eqref{eq_kast_inv} is
\begin{multline*}
 \frac{g(t,x)}{\kappa q^{-S/2+t-x/2+1/2}-(\kappa q^{-S/2+t-x/2+1/2})^{-1}}
 \sum_{i\ge N}f_i^t(x)c_i^{t,r-1}\\ \times
 \Biggl[\frac{\left(\kappa q^{-S/2+y-r/2+1/2}-\frac1{\kappa
q^{-S/2+y-r/2+1/2}}\right)c_i^{r-1}f_i^r(y)}{g(r,y)}
 +
 \frac{f_i^{r-1}(y)}{g(r-1,y)}+\frac{f_i^{r-1}(y-1)}{g(r-1,y-1)}\Biggr].
\end{multline*}

Let us find such function $g$ that for every $i$
\begin{equation}
\label{eq_qRacah_relation_needed}
 \frac{\left(\kappa q^{-S/2+y-r/2+1/2}-\frac1{\kappa
q^{-S/2+y-r/2+1/2}}\right)c_i^{r-1}f_i^r(y)}{g(r,y)}
 +
 \frac{f_i^{r-1}(y)}{g(r-1,y)}+\frac{f_i^{r-1}(y-1)}{g(r-1,y-1)}=0.
\end{equation}

We know that (see Lemma \ref{Lemma_qRacah_relation_funct})
\begin{multline}
\label{eq_qRacah_relation_funct}
 c_i^t f^{t+1}_i(y)-\sqrt\frac{w_{t,S}(y-1)}{w_{t+1,S}(y)}
w_1^t(y-1)f^t_i(y-1)-\sqrt\frac{w_{t,S}(y)}{w_{t+1,S}(y)} w_0^t(y)f^t_i(y)=0,
\end{multline}
where
$$
 w_0^t(x)=-(1-q^{x+T-t-S})\frac{1- \kappa ^2q^{x+N-t}} {1- \kappa ^2
 q^{2x-t-S+1}},
$$

$$
 w_1^t(x)=q^{T+N-1-t}(1-q^{-(S+N-1-x)}) \frac{1- \kappa ^2 q^{x-T+1}} {1- \kappa ^2
 q^{2x-t-S+1}},
$$
while $w_{t,S}(x)$ stands for the weight function corresponding to the parameters $t,q,\kappa
,N,T,S$ (see Section \ref{Section_distributions} and Theorem \ref{Theorem_One_dim_distribution} for
details).

For every pair $(r,y)$ we get the following three equations defining $g$
$$
 g(r,y)\propto \frac{\kappa q^{-S/2+y-r/2+1/2}-\frac1{\kappa
q^{-S/2+z-r/2+1/2}} }{ \sqrt{w_{r,S}(y)}},
$$
$$
 g(r-1,y-1)\propto -\frac{1}{\sqrt{w_{r-1,S}(y-1)} w_1^{r-1}(y-1)},
$$
$$
 g(r-1,y)\propto -\frac{1}{\sqrt{w_{r-1,S}(y)}
 w_0^{r-1}(y)},
$$
where the proportionality coefficient is the same for all three equations (but it may depend on the
pair $(r,y)$). One checks that $g$ given by \eqref{eq_conjug_factor} satisfies these relations and
after the conjugation with $g$ the $((t,x);(r,y))$ matrix element of the left-hand side of
\eqref{eq_kast_inv} is zero.

Next, suppose that $t>r$. In this case the $((t,x);(r,y))$ matrix element of the left-hand side of
\eqref{eq_kast_inv} is zero by the similar reasoning.

If either $t=r$ or $t=r-1$ the argument becomes a little more involved, but the computation
requires no new ideas.

\end{proof}

\section{Bulk limits. Limit shapes.}

\label{Section_Bulk_limits}

\subsection{Bulk limit theorem}

In this section we compute so-called ``bulk limits'' of the correlation functions introduced in the
previous section.

We are interested in the following limit regime. Fix positive
numbers $\mathsf S$, $\mathsf T$, $\mathsf N$, $\mathsf t$, $\mathsf
x$, $\mathsf q$.
 Introduce a small parameter $\varepsilon \ll 1$, and set
 $$
   S= \mathsf S\varepsilon^{-1} + o(\varepsilon^{-1}),\quad
   T= \mathsf T\varepsilon^{-1} + o(\varepsilon^{-1}),\quad
   N= \mathsf N\varepsilon^{-1} + o(\varepsilon^{-1}),\quad
   q= {\mathsf q}^{\varepsilon + o(\varepsilon)}.
 $$

  Consider also integer valued functions $t_i=t_i(\varepsilon)$ and
  $x_i=x_i(\varepsilon)$, $i=1,\dots,n$, such that
  $$
   \lim\limits_{\varepsilon\to 0} \varepsilon
   t_i(\varepsilon)=\mathsf t,\quad \lim\limits_{\varepsilon\to 0} \varepsilon x_i(\varepsilon)=\mathsf x,\quad
   i=1,\dots,n,
  $$
  and pairwise differences $t_i-t_j$, and $x_i-x_j$ do not depend on
  $\varepsilon$.

  Then the correlation functions $\rho_n$ computed in Theorem \ref{Theorem_Cor_functions} tend to a limit
 $\hat \rho_n$ which
 depends on the parameters of the limit regime $\mathsf q,\mathsf S$, $\mathsf T$, $\mathsf N$, $\mathsf t$, $\mathsf
 x$ and the original parameter $\kappa $.

 We consider the region where the limit correlation functions are nontrivial. This region is commonly referred to as
the ``bulk'', sometimes also called the ``liquid region''. This is a simply connected domain inside
 the hexagon.

 The main result of this section is Theorem \ref{Theorem_bulk_limit_1}.

 Note that the first limit correlation function allows us to predict \emph{the limit
 shape} which appears in our model.

\begin{theorem} We have
\label{Theorem_bulk_limit_1}
 $$
 \lim_{\varepsilon\to 0}\rho_n(t_1,x_1;\dots;t_n,x_n)=
  \det\left[\hat K(t_i,x_i;t_j,x_j)\right]_{i,j=1,\dots ,n},
$$
 where
 $$\hat K( x, s; y, t)= \frac{1}{2\pi i}\oint_{e^{-i\phi}}^{e^{i\phi}}
  \left(1+cw\right)^{ t- s}w^{ x- y-1}dw. $$
  Here the integration is to be performed over the right side of the unit
  circle when
   $s\ge t$ and over the left side otherwise,

 $$
   c=\left({ \frac{{\mathsf q}^{\mathsf T-2\mathsf t}(1-{\mathsf q}^{-(\mathsf S+\mathsf N-\mathsf x)})(1-{\mathsf q}^{\mathsf x} )}
   {(1-{\mathsf q}^{\mathsf x+\mathsf T-\mathsf t-\mathsf S})(1-{\mathsf q}^{-\mathsf t-\mathsf N+\mathsf x})}
    \frac{(1-\kappa ^2 {\mathsf q}^{\mathsf x+\mathsf N-\mathsf S})(1- \kappa ^2 {\mathsf q}^{\mathsf x-\mathsf T})}
 {(1- \kappa ^2 {\mathsf q}^{\mathsf x+\mathsf N-\mathsf t})(1-\kappa ^2 {\mathsf q}^{\mathsf x-\mathsf t-\mathsf
 S})}}\right)^\frac12,
 $$
and $\phi$ is given by the formula:
 $$\phi=
   \arccos\frac{{\mathsf q}^{-\mathsf N}(1-{\mathsf q}^{\mathsf N})(1-{\mathsf q}^{-\mathsf T-\mathsf N})(1-\kappa ^2
{\mathsf q}^{-\mathsf t-\mathsf S+2\mathsf x})^2+A+B}{2\sqrt{AB}},
 $$
where
$$
A= (1-{\mathsf q}^{-\mathsf S-\mathsf N+\mathsf x})(1-\kappa ^2
{\mathsf q}^{-\mathsf T+\mathsf x})(1-{\mathsf q}^{-\mathsf
t-\mathsf N+\mathsf x})(1-\kappa ^2 {\mathsf q}^{-\mathsf t-\mathsf
S+\mathsf x}),
$$
$$
B= {\mathsf q}^{-2\mathsf N-\mathsf T}(1-{\mathsf q}^{\mathsf
x})(1-\kappa ^2 {\mathsf q}^{-\mathsf t+\mathsf N+\mathsf x})(1-
{\mathsf q}^{-\mathsf t-\mathsf S+\mathsf T+\mathsf x})(1-\kappa ^2
{\mathsf q}^{-\mathsf S+\mathsf N+\mathsf x}).
$$
 If the expression under $\arccos$ is greater than 1, then we set $\phi=0$. If the expression
 is less than $-1$, then $\phi=\pi$.

\end{theorem}

Setting $v=-cw$ (and omitting some ``conjugation factors'' again) we
get the incomplete beta-kernel form of the integral, cf. \cite{OR},
$$
 K( x, s; y, t)= \frac{1}{2\pi i}\oint_{-c\cdot e^{-i\phi}}^{-c \cdot e^{i\phi}}
  \left(1-v\right)^{ t- s}v^{ x- y-1}dv.
$$
Here the contour of integration intersects $(-\infty,0)$, if
   $ s\ge  t$, and intersects $(0,1)$ otherwise. For an explanation
   of the relation of the incomplete beta-kernel and Gibbs measures
   see \cite{KOS}, \cite{BS}.

It is not hard to compute that $z=-ce^{i\phi}$ has the form
\begin{multline*}
z=\frac12\ \frac{{\mathsf q}^{\mathsf T+\mathsf N-\mathsf t}} {(1-{\mathsf q}^{\mathsf x+\mathsf
T-\mathsf t-\mathsf S})(1-{\mathsf q}^{-\mathsf t-\mathsf N+\mathsf x}) (1-\kappa ^2 {\mathsf
q}^{-\mathsf t+\mathsf
N+\mathsf x})(1-\kappa ^2{\mathsf q}^{\mathsf x-\mathsf t-\mathsf S})}\\
\cdot\Biggl[
 {{\mathsf q}^{-\mathsf N}(1-{\mathsf q}^{\mathsf N})(1-{\mathsf q}^{-\mathsf T-\mathsf N})(1-\kappa ^2
{\mathsf q}^{-\mathsf t-\mathsf S+2\mathsf x})^2+A+B}
\\+i\sqrt{4AB-\left({{\mathsf q}^{-\mathsf N}(1-{\mathsf q}^{\mathsf N})(1-{\mathsf q}^{-\mathsf T-\mathsf N})
(1-\kappa ^2 {\mathsf q}^{-\mathsf t-\mathsf S+2\mathsf x})^2+A+B}\right)^2} \Biggr],
\end{multline*}
with
$$
A= (1-{\mathsf q}^{-\mathsf S-\mathsf N+\mathsf x})(1-\kappa ^2 {\mathsf q}^{-\mathsf T+\mathsf
x})(1-{\mathsf q}^{-\mathsf t-\mathsf N+\mathsf x})(1-\kappa ^2 {\mathsf q}^{-\mathsf t-\mathsf
S+\mathsf x})
$$
and
$$
B= {\mathsf q}^{-2\mathsf N-\mathsf T}(1-{\mathsf q}^{\mathsf
x})(1-\kappa ^2 {\mathsf q}^{-\mathsf t+\mathsf N+\mathsf x})(1-
{\mathsf q}^{-\mathsf t-\mathsf S+\mathsf T+\mathsf x})(1-\kappa ^2
{\mathsf q}^{-\mathsf S+\mathsf N+\mathsf x}).
$$

\begin{proposition}
 The parameter $z$ defined above coincides with the one defined in
 Theorem \ref{Th_Local_measures}.
\end{proposition}
\begin{proof}
The quadratic equation satisfied by $z$ is
$$
Z^2-(z+\bar{z})Z+z\bar{z}=0.
$$
Substituting the expression for $z$ given above, one obtains a
relation equivalent to the one in Theorem  \ref{Th_Local_measures}.
\end{proof}

\subsection{Proof of the bulk limit theorem}

 In this section we prove Theorem \ref{Theorem_bulk_limit_1}.

Recall that the correlation kernel before the limit is given by
\begin{gather*}
  K(x,k;y,l)=\sum_{i=0}^{N-1}\frac{1}{c_i^{l,k}}f_i^k(x)f_i^l(y)
 ,\, k\ge l; \\ K(x,k;y,l)=-\sum_{i\ge N}c_i^{k,l} f_i^k(x)f_i^l(y),
 \,k<l; \\ c_i^{k,k}=1,\, c_i^{k,l}=c_i^k\cdot c_i^{k+1}\cdot\dots\cdot
 c_i^{l-1}.
 \end{gather*}
  The functions $f_i^k(x)$ and the coefficients $c_i^k$ were defined in Section
  \ref{Section_Corr_kernel}.

 First, let us consider the case $k=l$. We want to find a limit of the projection kernel
  $$
   {\cal P}_t(x,y)=\sum_{n=0}^{N-1}f_n^t(x)f_n^t(y).
  $$

  In order to find a limit of ${\cal P}_t(x,y)$ we use the spectral projection method proposed by
  G.~Olshanski, see \cite{BO2} and \cite{O}.

  We want to consider ${\cal P}_t(x,y)$ as a matrix element of the operator ${\cal P}_t$.  It turns out that finding the limit of the operator is easier than
 computing the limit of the matrix elements. Note that functions
  $f_n^t(x)$ are eigenvectors of some difference operator (it will be explicitly given below).
  The projection operator can be regarded as the spectral projection on the segment containing
 the first $N$ eigenvalues of this difference operator.
 Now, to find the limit of the spectral projection operators we will take the limit of the difference
 operators. Note that both the difference operator and the spectral
 segment are varying simultaneously.

 To justify the limit transition we use some facts from functional analysis.

 Consider the set  $l_2^0({\mathbb Z})$ of the finite vectors from
 $l_2({\mathbb Z})$ (i.e., the algebraic span of the basis elements $\delta_x$)
 as a common essential domain of all considered difference
 operators.
 It will be clear from the following that the difference operators strongly converge
 on this domain.
 It follows that the operators converge in the strong resolvent sense  (see \cite{RS}, Theorem VIII.25).
 The last fact, continuity of the spectrum of the limit operator, and Theorem
 VIII.24 from \cite{RS} imply that the spectral projections associated with the difference
 operators strongly converge on the set of finite vectors to the limit
 spectral projection associated with the limit difference operator.

 \medskip
 Now we present some details and computations.

 Note that since $q$-Racah polynomials
 are eigenfunctions of a certain difference operator (see \cite{KS}), the same
 is true for the functions $f_n^t(x)$. The difference operator is

\begin{multline}
\label{diff_op} {q^{-n}(1-q^n)(1-\alpha\beta q^{n+1})f_n^t(x)} =
B(x)f_n^t(x +
1)\sqrt{\frac{w_{t,S}(x)}{w_{t,S}(x+1)}} \\
 - [B(x) + D(x)]
f_n^t(x) + D(x)f_n^t(x - 1)\sqrt{\frac{w_{t,S}(x)}{w_{t,S}(x-1)}},
\end{multline}
where $w_{t,S}(x)$ is the weight function corresponding to the parameters $t,q,\kappa ,N,T,S$ (see
Theorem \ref{Theorem_One_dim_distribution}), and
$$
 B(x)=\frac{(1-\alpha q^{x+1})(1-\beta\delta q^{x+1})(1-\gamma q^{x+1})(1-\gamma\delta q^{x+1})}{(1-\gamma\delta q^{2x+1})(1-\gamma\delta q^{2x+2})}\,,
$$
$$
 D(x)=\frac{q(1-q^x)(\alpha-\gamma\delta q^x)(\beta-\gamma q^x)(1-\delta q^x)}{(1-\gamma\delta q^{2x})(1-\gamma\delta q^{2x+1})}\,.
$$
(Here $\alpha$, $\beta$, $\gamma$, $\delta$ are the corresponding
parameters of $q$-Racah polynomials.)

We find
\begin{multline*}
 w_{t,S}(x+1)/w_{t,S}(x)=\frac{q^{2N+T-1}(1-\kappa ^2 q^{2x-t-S+3})}{1-\kappa ^2 q^{2x-t-S+1}}
\\ \times \frac{
 (1-q^{x-t-N+1})(1-q^{x-S-N+1})(1-\kappa ^2 q^{x-T+1})(1-\kappa ^2 q^{x-t-S+1})}
 {(1-q^{x+1})(1-q^{T-S-t+x+1})(1-\kappa ^2 q^{x+N-t+1})(1-\kappa ^2
 q^{x+N-S+1})}\,,
\end{multline*}

\begin{multline*}
 w_{t,S}(x-1)/w_{t,S}(x)=\frac{q^{-2N-T+1}(1-\kappa ^2 q^{2x-t-S-1})}{1-\kappa ^2 q^{2x-t-S+1}}
\\ \times\frac{(1-q^x)(1-q^{T-S-t+x})(1-\kappa ^2q^{x+N-t})(1-\kappa ^2 q^{x+N-S})}
{(1-q^{x-t-N})(1-q^{x-S-N})(1-\kappa ^2 q^{x-T})(1-\kappa ^2
q^{x-t-S})}\,.
\end{multline*}

Substituting the $q$-Racah parameters of our model (see Section
\ref{Section_distributions}) into $B(x)$ and $D(x)$ one computes
$$
 B(x)=\frac{(1-q^{-S-N+x+1})(1-\kappa ^2 q^{-T+x+1})(1-q^{-t-N+x+1})(1-\kappa ^2 q^{-t-S+x+1})}{(1-\kappa ^2 q^{-t-S+2x+1})(1-\kappa ^2 q^{-t-S+2x+2})},
$$
\begin{multline*}
D(x)=\frac{q(1-q^x)(q^{-S-N}-\kappa ^2 q^{-t-S+x})(q^{S-T-N}- q^{-t-N+x})(1-\kappa ^2
q^{-S+N+x})}{(1-\kappa ^2
q^{-t-S+2x})(1-\kappa ^2 q^{-t-S+2x+1})} \\
=q^{1-2N-T}\frac{(1-q^x)(1-\kappa ^2 q^{-t+N+x})(1- q^{-t-S+T+x})(1-\kappa ^2
q^{-S+N+x})}{(1-\kappa ^2 q^{-t-S+2x})(1-\kappa ^2 q^{-t-S+2x+1})} .
\end{multline*}

It follows that
\begin{multline*}
 \frac{B(x)}{\sqrt{w_{t,S}(x+1)/w_{t,S}(x)}}
\\=\left(q^{1-2N-T}
{(1-q^{x+1)}(1-q^{T-S-t+x+1})(1-q^{-S-N+x+1})(1-q^{-t-N+x+1})}\right)^\frac12
\\ \times \left(\frac{(1-\kappa ^2 q^{-T+x+1})(1-\kappa ^2
q^{-t-S+x+1})(1-\kappa ^2 q^{x+N-t+1})(1-\kappa ^2
q^{x+N-S+1})}{(1-\kappa ^2 q^{-t-S+2x+1})(1-\kappa ^2
q^{2x-t-S+3})(1-\mathsf \kappa ^2 q^{-t-S+2x+2})^2}\right)^\frac12,
\end{multline*}

\begin{multline*}
  \frac{D(x)}{\sqrt{w_{t,S}(x-1)/w_{t,S}(x)}}\\
  =
\left(q^{{1-2N-T}}{(1-q^x)(1-q^{x-t-N})}(1-q^{x-S-N})(1-
q^{x-t-S+T}) \right)^\frac12
\\
\times \left(\frac{(1-\kappa ^2 q^{-t+N+x})(1-\kappa ^2
q^{-S+N+x})(1-\kappa ^2 q^{x-T})(1-\kappa ^2 q^{x-t-S})} {(1-\kappa
^2 q^{-t-S+2x+1})(1-\kappa ^2 q^{2x-t-S-1}) (1-\kappa ^2
q^{-t-S+2x})^2}\right)^\frac12.
\end{multline*}

The eigenvalues in the left-hand side of (\ref{diff_op}) become
$$
 q^{-n}(1-q^n)(1-q^{-T-2N+n+1}), \quad n=0,\dots,N-1
$$

Note that in the ``bulk limit'' regime, $q\to 1$, $N,T,S\to\infty$
as $\varepsilon\to 0$ in such a way that $q^N, q^T, q^S$ have finite
limits.

In the limit $\frac{B(x)}{\sqrt{w(x+1)/w(x)}}$ tends to some
constant and $\frac{D(x)}{\sqrt{w(x-1)/w(x)}}$ tends to the very
same constant. After dividing by twice this constant the limit
operator becomes
$$
 f(x)\mapsto \frac{f(x+1)+f(x-1)}2-f(x)\frac{A+B}{2\sqrt{AB}},
$$
where
$$
A= (1-{\mathsf q}^{-\mathsf S-\mathsf N+\mathsf x})(1-\kappa ^2 {\mathsf q}^{-\mathsf T+\mathsf
x})(1-{\mathsf q}^{-\mathsf t-\mathsf N+\mathsf x})(1-\kappa ^2 {\mathsf q}^{-\mathsf t-\mathsf
S+\mathsf x})
$$
and
$$
B= {\mathsf q}^{-2\mathsf N-\mathsf T}(1-{\mathsf q}^{\mathsf
x})(1-\kappa ^2 {\mathsf q}^{-\mathsf t+\mathsf N+\mathsf x})(1-
{\mathsf q}^{-\mathsf t-\mathsf S+\mathsf T+\mathsf x})(1-\kappa ^2
{\mathsf q}^{-\mathsf S+\mathsf N+\mathsf x}),
$$
while the spectral interval becomes
$$
 \left[\frac{{\mathsf q}^{-\mathsf N}(1-{\mathsf q}^{\mathsf N})(1-{\mathsf q}^{-\mathsf T-\mathsf N})(1-\kappa ^2
{\mathsf q}^{-\mathsf t-\mathsf S+2\mathsf x})^2}{2\sqrt{AB}},0\right]
$$

Since the spectrum of the operator $f\to \frac{f(x+1)+f(x-1)}2$ is $[-1,1]$, while
$\frac{\sqrt{A+B}}{2\sqrt{AB}}>1$, one can equivalently write the operator in the form
$$
 f\to \frac{f(x+1)+f(x-1)}2
$$
with spectral interval
\begin{equation}
 \label{Spectral_interval_qRacah}
 \left[\frac{{\mathsf q}^{-\mathsf N}(1-{\mathsf q}^{\mathsf N})(1-{\mathsf q}^{-\mathsf T-\mathsf N})(1-\kappa ^2
{\mathsf q}^{-\mathsf t-\mathsf S+2\mathsf x})^2+A+B}{2\sqrt{AB}},1\right].
\end{equation}

If we perform the Fourier transform $l_2({\mathbb Z})\to L_2(S^1)$, where $S^1$
 is the unit circle in ${\mathbb C}$, we get an operator of the projection
 on the part of the unit
 circle with $x$-coordinate varying over precisely the spectral interval
 \eqref{Spectral_interval_qRacah}.

 If we do the inverse Fourier transform we will get the discrete sine
 kernel, cf. the end of Section 3.2 in \cite{Gor}.

 \medskip
 Next, let us consider the case $k<l$. The prelimit correlation kernel is given
 by
\begin{gather*}
  K(x,k;y,l)=-\sum_{i\ge N}c_i^{k,l} f_i^k(x)f_i^l(y),
 \qquad k<l; \\ c_i^{k,k}=1,\qquad c_i^{k,l}=c_i^k\cdot c_i^{k+1}\cdot\dots\cdot
 c_i^{l-1}.
 \end{gather*}

 Let us decompose the correlation kernel (i.e., the operator given by it)
 into the product of the static projection kernel:
  $$
   {\cal P'}_t(x,y)=-\sum_{i\ge N}^{N-1}f_i^t(x)f_i^t(y),
  $$
and a collection of transition operators (or their inverses)
  $U_h(x,y)$ with
  $$
   U_h(x,y)=\begin{cases} \sum\limits_{i \ge 0}c_i^hf_i^h(x)f_i^{h+1}(y),
        &x\in{\mathfrak X_h},y\in{\mathfrak X_{h+1}},\\
            0 &\text{for other } x,y.
  \end{cases}
  $$

 For the operators ${\cal P'}_t(x,y)$ we can use the same methods as for ${\cal P}_t(x,y)$. We get
 minus the operator of projection on the part of the unit circle complementary to the spectral interval
 \eqref{Spectral_interval_qRacah}.

 Let us turn to the transition operators $U_t$.

 By virtue of already proved facts (see Section \ref{Section_Corr_kernel}), we obtain
 $$
   U_t(x,y)=const\cdot \sqrt{\frac{w_{t,S}(x)}{w_{t+1,S}(y)}}
  \left[w_1(x)\delta_{x+1}^y+
  w_0(x)\delta_x^y\right],
 $$
 where

  \begin{multline*}
w_{t,S}(x)=\frac{q^{x(2N+T-1)}(1- \kappa ^2q^{2x-t-S+1})}{(q;q)_x (q;q)_{T-S-t+x}
(q^{-1};q^{-1})_{t+N-x-1}(q^{-1};q^{-1})_{S+N-x-1}}\\
\cdot\frac1{ ( \kappa ^2q^{x-T+1};q)_{T+N-t}( \kappa ^2q^{x-t-S+1};q)_{N+t} },
\end{multline*}

$$
 w_0(x)=-(1-q^{x+T-t-S})\frac{1-\kappa ^2q^{x+N-t}} {1-\kappa ^2
 q^{2x-t-S+1}}\,,
$$
and
$$
 w_1(x)=q^{T+N-1-t}(1-q^{-(S+N-1-x)}) \frac{1-\kappa ^2 q^{x-T+1}} {1-\kappa ^2
 q^{2x-t-S+1}}.
$$

Thus,
 $$
   U_t(x,y)=const\cdot
  \left[\tilde w_1(x)\delta_{x+1}^y+
  \tilde w_0(x)\delta_x^y\right],
 $$
where
\begin{equation*}
 \tilde w_0(x)
=\left(\frac{(1-q^{x+T-t-S})({1-q^{-t-N+x}})(1-\kappa ^2
q^{x-t-S})(1-\kappa ^2q^{x+N-t})} {(1-\kappa ^2
 q^{2x-t-S+1})(1-\kappa ^2
 q^{2x-t-S})}\right)^\frac12,
\end{equation*}
and
\begin{equation*}
 \tilde w_1(x)=\left(\frac{q^{T-t-1}{(1-q^{x-S-N+1})}{(1-q^{x+1}) }(1-\kappa ^2 q^{x+N-S+1})
 (1- \kappa ^2 q^{x-T+1}) } {(1-\kappa ^2
 q^{2x-t-S+1})(1-\kappa ^2
 q^{2x-t-S+2})}\right)^\frac12.
\end{equation*}

Passing to the limit we get the operator
$$
   \mathsf U(x,y)=const\cdot
  \left[\mathsf U_1\delta_{x+1}^y+
   \mathsf U_0\delta_x^y\right],
$$

where
$$
\mathsf U_0=
 \left(\frac{(1-{\mathsf q}^{\mathsf x+\mathsf T-\mathsf
t-\mathsf S})(1-{\mathsf q}^{-\mathsf t-\mathsf N+\mathsf
x})(1-\kappa ^2 {\mathsf q}^{\mathsf x-\mathsf t-\mathsf S})(1-
\kappa ^2 {\mathsf q}^{\mathsf x+\mathsf N-\mathsf t})}
 {(1- \kappa ^2 {\mathsf q}^{2\mathsf x-\mathsf t-\mathsf S})\
 (1- \kappa ^2 {\mathsf q}^{2\mathsf x-\mathsf t-\mathsf
 S})}\right)^\frac 12
$$
 and
$$
\mathsf U_1=\left(\frac{{\mathsf q}^{\mathsf T-\mathsf
2t}{(1-{\mathsf q}^{\mathsf x-\mathsf S-\mathsf N})(1- \kappa ^2
{\mathsf q}^{\mathsf x-\mathsf T})
  {(1-{\mathsf q}^{\mathsf x} )}(1-\kappa ^2 {\mathsf q}^{\mathsf x+\mathsf
N-\mathsf S} )}} {(1-\kappa ^2
 {\mathsf q}^{2\mathsf x-\mathsf t-\mathsf S})
 (1- \kappa ^2 {\mathsf q}^{2\mathsf x-\mathsf t-\mathsf
 S})}\right)^\frac12.
$$

Equivalently,
$$
   \mathsf U(x,y)=const\cdot
  \left[\mathsf u \delta_{x+1}^y+
  \delta_x^y\right],
$$
where
$$
\mathsf u=\frac{\mathsf U_1}{\mathsf U_0}=\left({ \frac{{\mathsf
q}^{\mathsf T-2\mathsf t}(1-{\mathsf q}^{\mathsf x-\mathsf S-\mathsf
N})(1-{\mathsf q}^{\mathsf x} )}
 {(1-{\mathsf q}^{\mathsf x+\mathsf T-\mathsf t-\mathsf S})(1-{\mathsf q}^{-\mathsf t-\mathsf N+\mathsf x})}
 \frac{(1-\kappa ^2 {\mathsf q}^{\mathsf x+\mathsf N-\mathsf S})(1- \kappa ^2 {\mathsf q}^{\mathsf x-\mathsf T})}
 {(1- \kappa ^2{\mathsf q}^{\mathsf x+\mathsf N-\mathsf t})(1-\kappa ^2
{\mathsf q}^{\mathsf x-\mathsf t-\mathsf S})}}\right)^\frac12.
$$

Fourier transform gives us the operator of multiplication by
$const(1+\mathsf u/w)$, where $w$ is the coordinate on the circle
$|w|=1$.

If we now multiply all necessary operators, perform inverse Fourier
transform and substitute $w\to 1/w$ in the resulting integral, we
get the desired limit kernel. Note that the constant prefactor in
$\mathsf U$ can be omitted since it corresponds to the conjugation
of the kernel that does not affect the correlation functions.

\medskip

The final case  $k>l$ is similar.  The interested reader can find some
details in \cite{Gor} where similar computations (with Hahn
polynomials instead of $q$-Racah polynomials) were performed.

\section{Computer simulations. Different limit regimes.}

\label{Section_comp_simulations}

Using the perfect sampling algorithm described in Section
\ref{Section_perfect_sampling_algorithm}, we performed some computer
simulations; the program that we used can be found at
\texttt{http://www.
math.caltech.edu/papers/Borodin-Gorin-Rains.exe}. We are mostly
interested in the case when the hexagon is large, since in this case
we can see some limit shapes appearing. In all pictures we color the
three types of lozenges in three different colors (as in Figure 1).
When we draw big pictures, we erase the borders between lozenges and
get some coloring of a hexagon in three colors which can also be
viewed as a stepped surface in $\mathbb R^3$.

Although we mostly show not very big pictures, our algorithm can
generate random tilings of a $1000\times 1000\times 1000$ hexagon in a
reasonable amount of time.

\smallskip

 The following picture shows a plane partition
in a $70\times 90\times 70$ box sampled from the distributions with
parameters $q=0.97$, $\kappa =1$.  The formation of a limit shape with
frozen regions is clearly visible on the picture, and the next
picture shows the border of the frozen region as predicted by
Theorem \ref{Th_Local_measures}.

\begin{center}
 {\scalebox{0.28}{\includegraphics{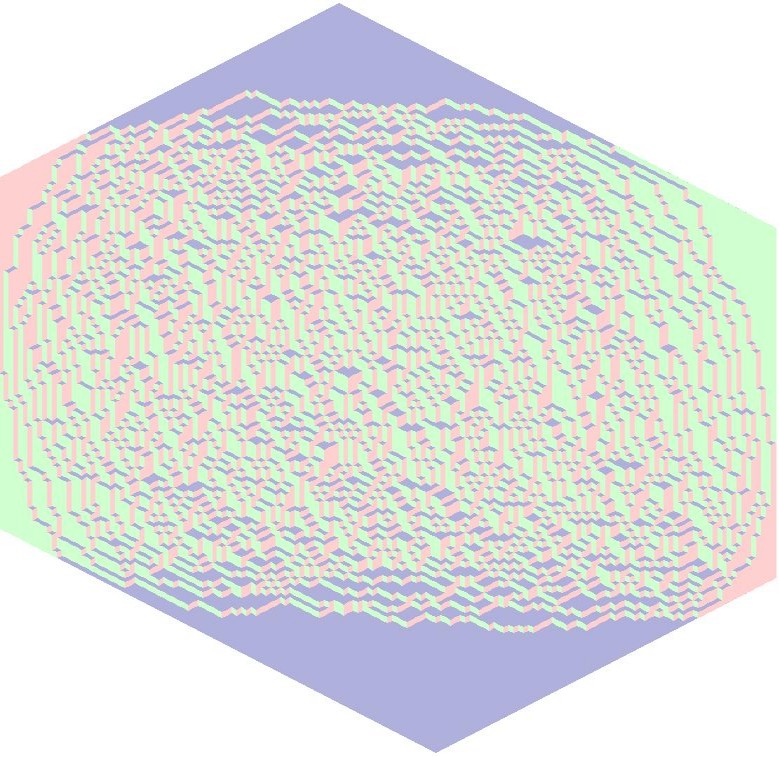}}}
 {\scalebox{0.37}{\includegraphics{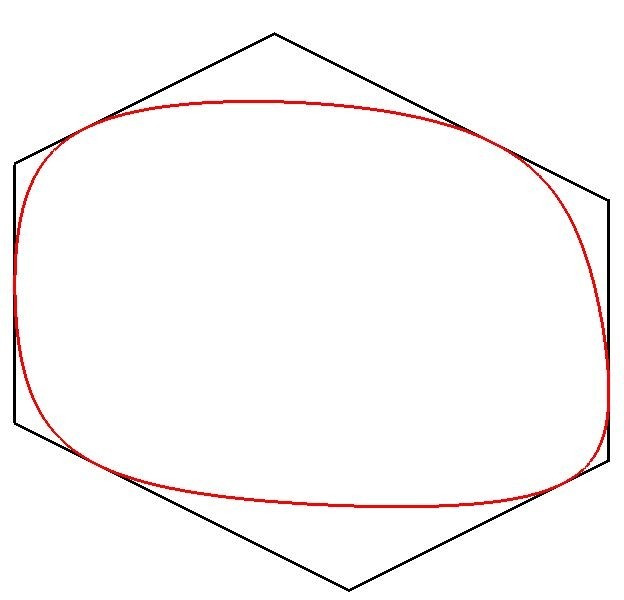}}}
\end{center}

By the appropriate limit transition, we get plane partitions
distributed as $q^{volume}$. The following picture shows a random
plane partitions in a $70\times 90\times 70$ box and the corresponding
theoretical frozen boundary with $q=1.04$.

\begin{center}
 {\scalebox{0.28}{\includegraphics{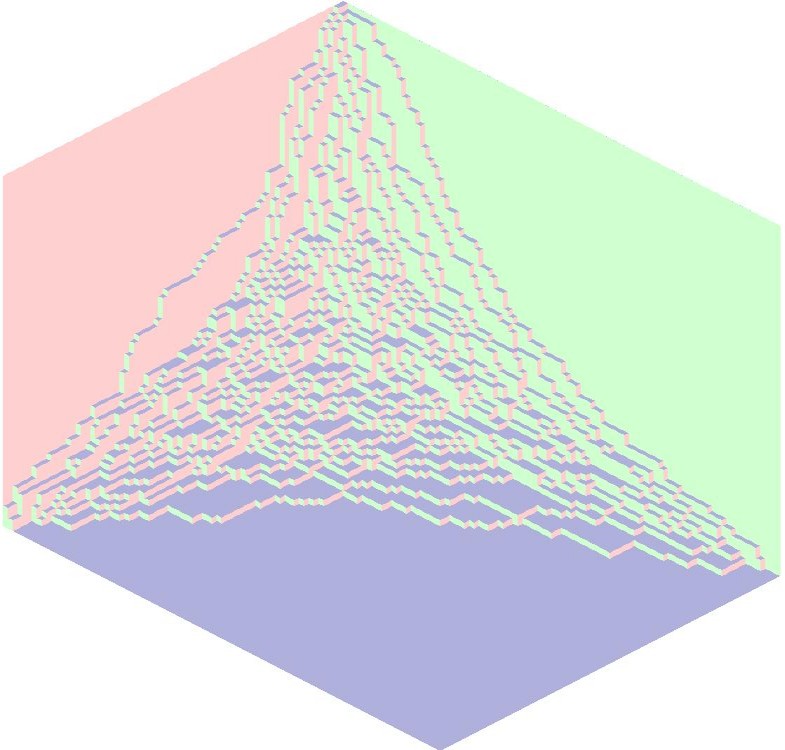}}}
 {\scalebox{0.37}{\includegraphics{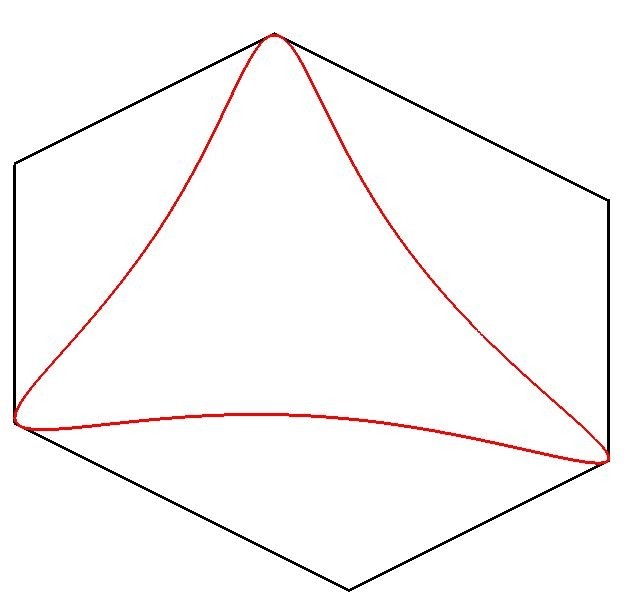}}}

\end{center}

As was explained in Section \ref{Section_distributions}, if we send
$q\to 1$ in the original model, then the weight of a horizontal
lozenge becomes a linear function in the vertical coordinate. If we
tune the parameters in such a way that our linear function has a
zero at the bottommost point of a hexagon, then we get the following
random plane partition in a $70\times 90\times 70$ box, cf. the
corresponding theoretical frozen boundary.

\begin{center}
 {\scalebox{0.28}{\includegraphics{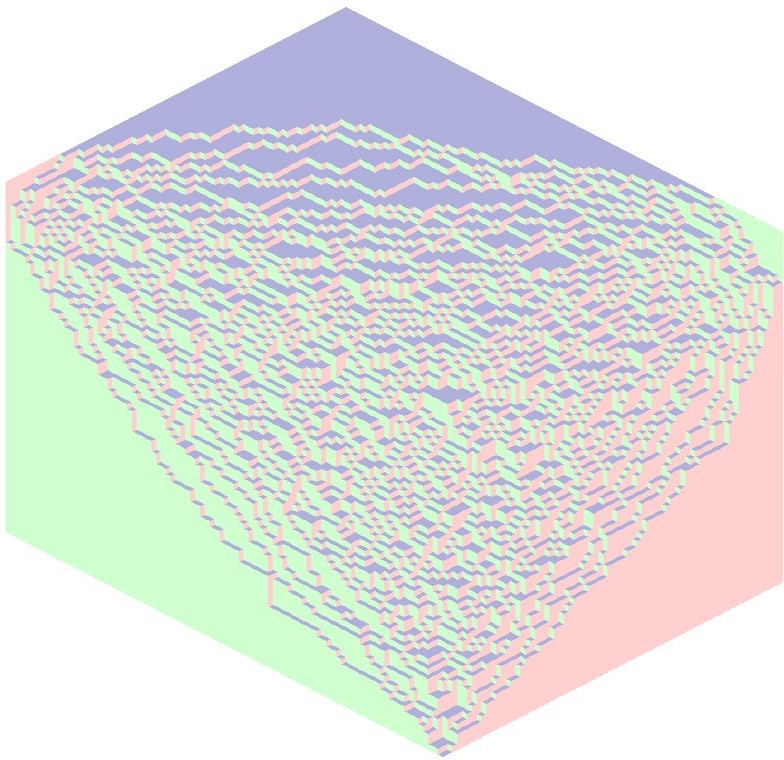}}}
 {\scalebox{0.41}{\includegraphics{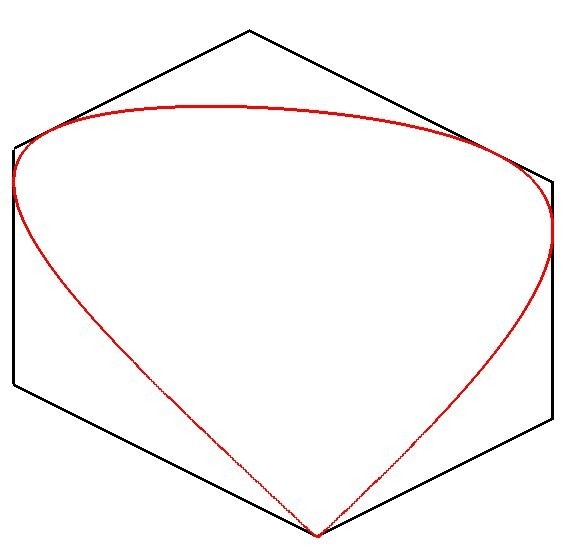}}}
\end{center}

We see that the border of frozen region has a node near the
bottommost point of a hexagon.

\medskip

Another interesting case to consider is when $q$ does not tend to
$1$ as the size of the hexagon tends to infinity. Look at the
following picture where the random plane partition with parameters
$q=0.9$, $\kappa =1$ in a $70\times 90\times 70$ box is shown. We see
that the surface becomes different from the ones shown above. We
call this new class of surfaces \emph{waterfalls}.

\begin{center}
 {\scalebox{0.28}{\includegraphics{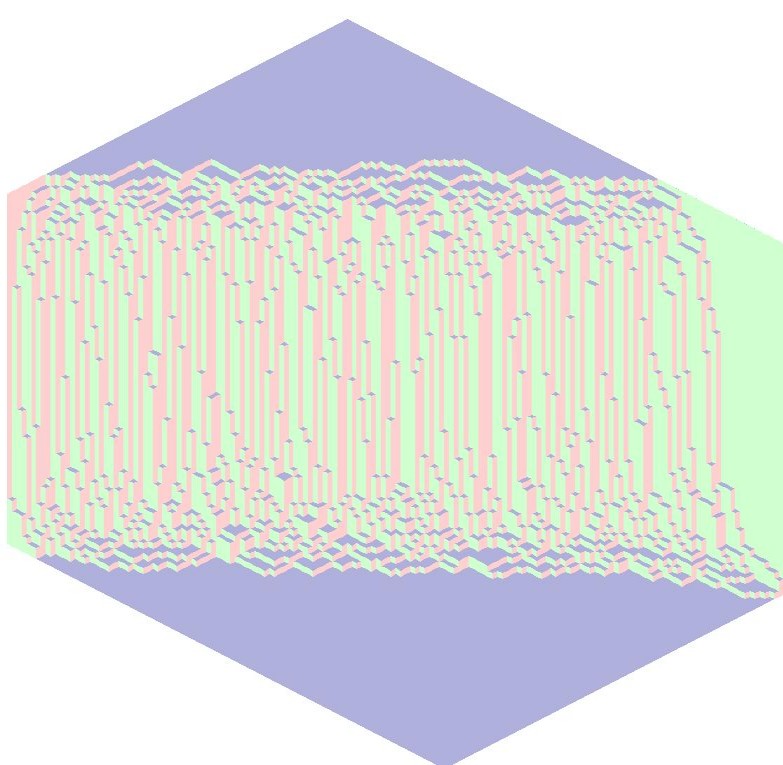}}}
\end{center}

The next picture shows an even more degenerate case $q=0.7$.
\begin{center}
 {\scalebox{0.28}{\includegraphics{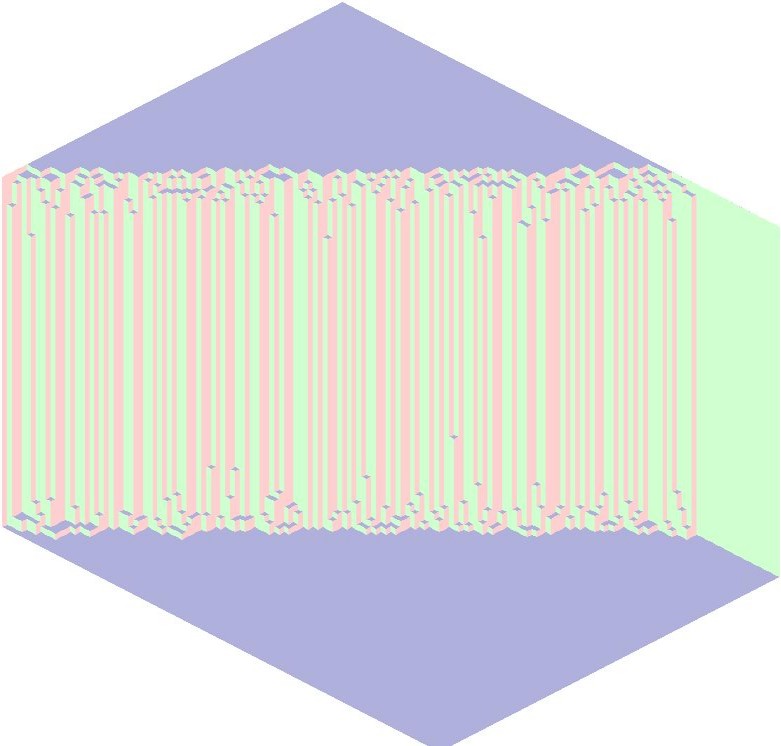}}}
\end{center}
We hope to study the asymptotic behavior of these cases in a later
publication.

\medskip

Finally, let us turn to the trigonometric $q$-Racah case. In this
case the weight of a horizontal lozenge is
\blue{$\sin(\alpha(j-(S+1)/2)+\beta)$}. One can tune the parameters
$\alpha$ and $\beta$ in such a way that this weight becomes zero at
both the topmost and the bottommost points of the hexagon. The
boundary of the frozen region has two nodes in this case.

\begin{center}
 {\scalebox{0.28}{\includegraphics{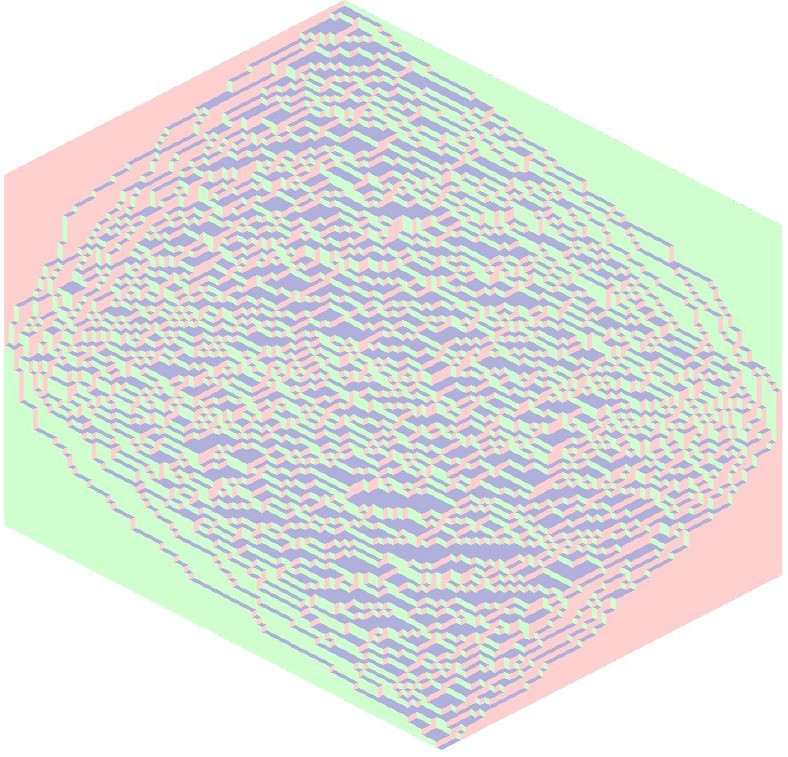}}}
 {\scalebox{0.37}{\includegraphics{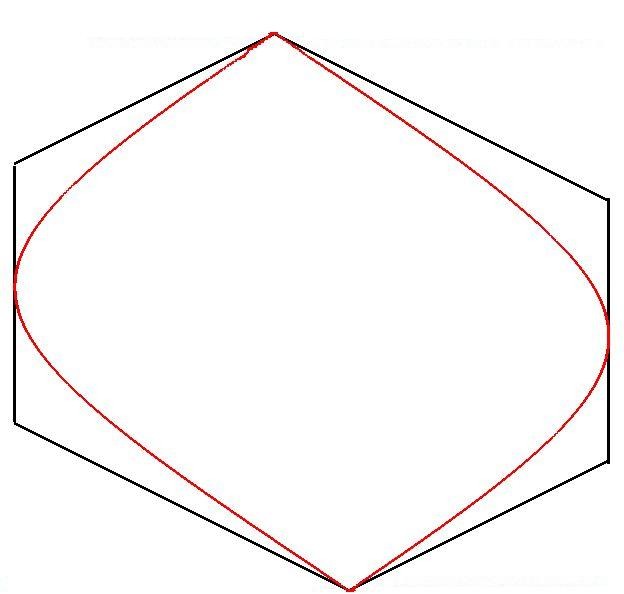}}}
\end{center}

\section{Appendix. Lozenge tilings and biorthogonal
functions}\label{sc:Appendix}

We were led to consider the $q$-Racah and Racah cases via the
realization that, in much the same way that uniform lozenge tilings
are related to Hahn polynomials, so were $q$-weighted lozenge
tilings related to $q$-Hahn polynomials.  This naturally led to the
question of whether more general (discrete) hypergeometric
orthogonal polynomials could arise in this way. In fact, as we will
consider in this section, one can generalize even further, to the
elliptic analogue, certain biorthogonal elliptic functions due to
Spridonov and Zhedanov \cite{SpiridonovVP/ZhedanovAS:2000b}.

One way to derive the required lozenge weights (essentially
equivalent weights were considered by Schlosser in
\cite{SchlosserM:2008}, although the connection to plane partitions
or lozenge tilings was not made explicit there) is via the following
desiderata:

First, we want the total weight of all tilings of a hexagon to be
``nice'', in the sense that it should be expressible as a product of
simple theta functions
\[
\theta_p(x) := \prod_{0\le k} (1-p^k x)(1-p^{k+1}/x);
\]
the same should apply to the individual weights as well.  Note that, as is
traditional in much recent work on elliptic special functions, we use the
multiplicative form for our elliptic curves and theta functions.  This can
be translated to the usual doubly-periodic form by composing with the
singly-periodic function $x\mapsto \exp(2\pi i x)$, but the multiplicative
form makes certain degenerations simpler to obtain.

Next, the sums that arise should be hypergeometric, in the sense
that any parameters should vary along geometric progressions as one
moves around in the hexagon.

Finally, the form of the weight of a given cube (i.e., the ratio of
the weights of the two ways to tile any given unit hexagon) should
be invariant under the symmetry group of the triangular lattice,
except that certain reflections should invert the weight. Note that
the weights of cubes are gauge-invariant, and any set of choices of
cube weights corresponds to a unique gauge-equivalence class.

We may rewrite these criteria in terms of plane partitions, noting
that the requirement that the cube weights come from lozenge weights
places the additional restriction that if $w(x,y,z)$ is the weight
of a cube with centroid $(x,y,z)$, then
\[
w(x,y,z) = w(x+1,y+1,z+1).
\]
If we consider plane partitions inside a $1\times 1\times n$ box, we
thus require that
\[
\sum_{0\le l\le n} \prod_{0\le m<l} w(1/2,1/2,m+1/2)
\]
should be nice, and should indeed correspond to a hypergeometric
sum.  The isotropy condition then implies more generally that
\[
\sum_{0\le l\le n} \prod_{0\le m<l} w(x,y,z+m)
\]
is a hypergeometric sum for any half-integer vector $(x,y,z)$.  In
particular, this sum is doubly-telescoping, in the sense that the
sum over any subinterval is nice.  In particular, if we assume that
the sum should be a special case of the Frenkel-Turaev sum
\cite{FrenkelIB/TuraevVG:1997}, we find that it should have the form
(up to some convenient changes of parameters)
\[
\sum_{0\le l\le n} \frac{q^l \theta_p(ab q^{2l-1})}
     {\theta_p(q^{l-1}a,q^l a,q^{l-1} b,q^l b)}
\frac{\theta_p(a/q,a,b/q,b)}
     {\theta_p(ab/q)}
= \frac{\theta_p(q^{n-1} ab,q^{n+1},a,b)}
     {\theta_p(ab/q,q,q^n a,q^n b)},
\]
where $a$ and $b$ depend on $(x,y,z)$.  In particular, we find
\[
w(x,y,z+m) = \frac{q\theta_p(q^{m-1}a,q^{m-1}b,ab q^{2m+1})}
     {\theta_p(q^{m+1}a,q^{m+1}b,ab q^{2m-1})}
= \frac{q^3\theta_p(q^{m-1}a,q^{m-1}b,q^{-2m-1}/ab)}
     {\theta_p(q^{m+1}a,q^{m+1}b,q^{1-2m}/ab)},
\]
where $a$,$b$ depend on $(x,y,z)$; consistency then implies that
$q^{-z}a$ and $q^{-z}b$ are independent of $z$.

Rotating the picture by $120$ degrees gives a similar expression for
$w(i,j,k)$ with the dependence on $i$ or $j$ factored out; comparing
the results leads us to an expression of the form
\[
w(x,y,z) =
\frac{q^3\theta_p(q^{y+z-2x-1}u_1,q^{x+z-2y-1}u_2,q^{x+y-2z-1}u_3)}
     {\theta_p(q^{y+z-2x+1}u_1,q^{x+z-2y+1}u_2,q^{x+y-2z+1}u_3)},
\]
where $u_1=a$, $u_2=b$, $u_3=1/ab$; i.e., $u_1$, $u_2$, $u_3$ are generic
such that $u_1u_2u_3=1$.  We will see that this indeed gives rise to a
factored sum over plane partitions in a cube.  Note that if we rewrite this
expression in terms of new variables $\tilde{u}_1=q^{y+z-2x}u_1$,
$\tilde{u}_2=q^{x+z-2y}u_2$, $\tilde{u}_3=q^{x+y-2z}u_3$, then
\[
w(x,y,z) =
\frac{q^3\theta_p(\tilde{u}_1/q,\tilde{u}_2/q,\tilde{u}_3/q)}
     {\theta_p(q\tilde{u}_1,q\tilde{u}_2,q\tilde{u}_3)},
\]
which can be described in a canonical way: $w$ is the value at $q$ of the
unique elliptic function with simple zeros at $\tilde{u}_i$, simple poles
at $\tilde{u}_i^{-1}$, and the value $1$ at $1$.  It follows, in particular,
that $w$ is invariant under shifting any parameter by $p$, as well as under
all modular transformations.

To write this in terms of lozenge weights, there are, of course, a
number of choices one could make.  One convenient choice is to allow
only horizontal lozenges to have nonzero weights.  We must thus have
\[
\frac{w(i,j+1)}{w(i,j)} =
\frac{q\theta_p(q^{j-3i/2-1}u_1,q^{j+3i/2-1}u_2,q^{2j+1}u_1u_2)}
     {\theta_p(q^{j-3i/2+1}u_1,q^{j+3i/2+1}u_2,q^{2j-1}u_1u_2)}
\]
where $w(i,j)$ denotes the weight of a horizontal lozenge with upper
corner at $(i,j)$ (recall that in terms of 3-D coordinates,
$i=x-y$, $j=z-(x+y)/2$).  This recurrence is straightforward to
solve, and one obtains
\[
w(i,j) = C(i) \frac{q^{j-1/2}
(u_1u_2)^{1/2}\theta_p(q^{2j-1}u_1u_2)}
     {\theta_p(q^{j-3i/2-1}u_1,q^{j-3i/2}u_1,q^{j+3i/2-1}u_2,q^{j+3i/2}u_2)},
\]
where $C(i)$ is an arbitrary (non-vanishing) function of $i$.

If we view (with a mind to applying Kasteleyn's theorem) the lozenge
weight as a matrix indexed by a right-pointing and a left-pointing
triangle, we find that, coordinatizing triangles by their upper
corners,
\begin{gather*}
w((i,j),(i,j)) = C(i) \frac{(u_1u_2)^{1/2}q^{j-1/2}
\theta_p(q^{2j-1}u_1u_2)}
     {\theta_p(q^{j-3i/2-1}u_1,q^{j-3i/2}u_1,q^{j+3i/2-1}u_2,q^{j+3i/2}u_2)}
\\
w((i,j),(i+1,j+1/2)) = 1,\qquad w((i,j),(i+1,j-1/2)) = 1,
\end{gather*}
and all other values are $0$.

Let $\Pi^{y_0,y_1}_{x_0,x_1}$ represent the parallelogram
\[
x_0\le i\le x_1, y_0\le j+i/2\le y_1,
\]
and observe that the restriction of $w$ to triangles in
$\Pi^{y_0,y_1}_{x_0,x_1}$ is a square matrix, and we can thus
attempt to invert $w$ in such a region.  In fact, not only can we
explicitly invert $w$ in such parallelograms, but the result is
independent of the choice of parallelogram.

\begin{theorem}
The inverse transpose of $w$ in $\Pi^{y_0,y_1}_{x_0,x_1}$ has the
form
\begin{multline*}
W((i_0,j_0),(i_1,j_1)) = \delta_{i_0<i_1} \prod_{i_0\le k<i_1}
C(k)\cdot(u_1u_2)^{(i_1-i_0-1)/2}
\notag\\
\times \delta_{j_1+i_1/2\le j_0+i_0/2} (-1)^{j_0-i_0/2-j_1+i_1/2-1}
q^{(i_1-i_0-1)(i_1-i_0+4j_1-2)/4}
\notag\\
\times
\frac{\theta_p(q^{j_0+i_0/2-j_1-i_1/2+1},q^{j_0+i_0/2+j_1-i_1/2}u_1u_2;q)_{i_1-i_0-1}}
     {\theta_p(q,q^{j_0-i_0/2-i_1}u_1,q^{j_1-3i_1/2+1}u_1,q^{j_1+i_0+i_1/2}u_2,q^{j_0+3i_0/2+1}u_2;q)_{i_1-i_0-1}},
\end{multline*}
where
\[
\theta_p(x;q)_k := \prod_{0\le i<k} \theta_p(q^i x).
\]
\end{theorem}

\begin{proof}
Note that by Kasteleyn's theorem, $W$ is the total weight (up to
sign) of all lozenge tilings of the parallelogram that omit the two
given triangles.  In particular, $W((i_0,j_0),(i_1,j_1))$ must
vanish if $i_0\ge i_1$ or $j_0+i_0/2<j_1+i_1/2$ simply because then
no such tiling exists.

Now, the claim that $W$ is the inverse transpose of $w$ reduces to
the statement
\begin{multline*}
W((i_0,j_0),(i_1,j_1))w(i_1,j_1) + W((i_0,j_0),(i_1+1,j_1-1/2))\\ +
W((i_0,j_0),(i_1+1,j_1+1/2)) = \delta_{(i_0,j_0),(i_1,j_1)}.
\end{multline*}
This holds trivially for $i_0>i_1$, since all three terms vanish,
and similarly for $i_0=i_1$, $j_1>j_0$.  When $i_0=i_1$, $j_0=j_1$,
the claim reduces to
\[
W((i_0,j_0),(i_0+1,j_0-1/2)) = 1,
\]
while for $i_0=i_1$, $j_0\ge j_1$, we have
\begin{multline*}
W((i_0,j_0),(i_0+1,j_1-1/2)) + W((i_0,j_0),(i_0+1,j_1+1/2)) \\=
(-1)^{j_0-j_1-1} + (-1)^{j_0-j_1} = 0,
\end{multline*}
as required.

It remains to consider the case $i_0<i_1$.  If
$j_1+i_1/2>j_0+i_0/2$, then all three terms again vanish, while if
$j_1+i_1/2=j_0+i_0/2$, the third term vanishes, and
\[
W((i_0,j_0),(i_1,j_1))w(i_1,j_1) + W((i_0,j_0),(i_1+1,j_1-1/2)) = 0
\]
as required.  Finally, when $j_1+i_1/2<j_0+i_0/2$, so that all three
terms survive, if we divide by
\[
\frac{W((i_0,j_0),(i_1,j_1))w(i_1,j_1)}
     {\theta_p(q^{i_1-i_0},q^{j_0-i_0/2-i_1-1}u_1,q^{j_0+i_0/2+i_1}u_2,
               q^{2j_1-1}u_1u_2)},
\]
we simply obtain a special case of the addition law for $\theta_p$,
in the form
\begin{multline*}
\theta_p(a_0z,a_1z,a_2z,a_0a_1a_2/z) -
\theta_p(a_0a_1,a_0a_2,a_1a_2,z^2) \\ +
\theta_p(z/a_0,a_1/z,a_2/z,a_0a_1a_2z)za_0 = 0
\end{multline*}
with
\begin{align*}
&a_0 = q^{-j_0/2-3i_0/4+j_1/2+3i_1/4},\quad a_1 =
q^{j_0/2-i_0/4+j_1/2-5i_1/4-1}u_1,\\
 &a_2 =
q^{j_0/2+3i_0/4+j_1/2+3i_1/4}u_2, \quad z =
q^{j_0/2-i_0/4-j_1/2+i_1/4}.
\end{align*}
\end{proof}

\begin{Remark} One major source of guidance regarding the
form of $W$ is that it corresponds to an enumeration of plane
partitions in a rectangular parallelepiped with dimensions $m\times
n\times 1$, say; i.e., a sum over ordinary partitions.  If we first
sum over the first part of the partition, we find that $W$ should
look like the term of a (singly) telescoping hypergeometric sum. One
can also, of course, compute ``small'' values of $W$ in the case
$p=0$, and look for patterns in the resulting factorizations.
\end{Remark}

\medskip

\begin{lemma}
\label{Lemma_Weight_left}

Consider the domain
\begin{center}
 {\scalebox{0.8}{\includegraphics{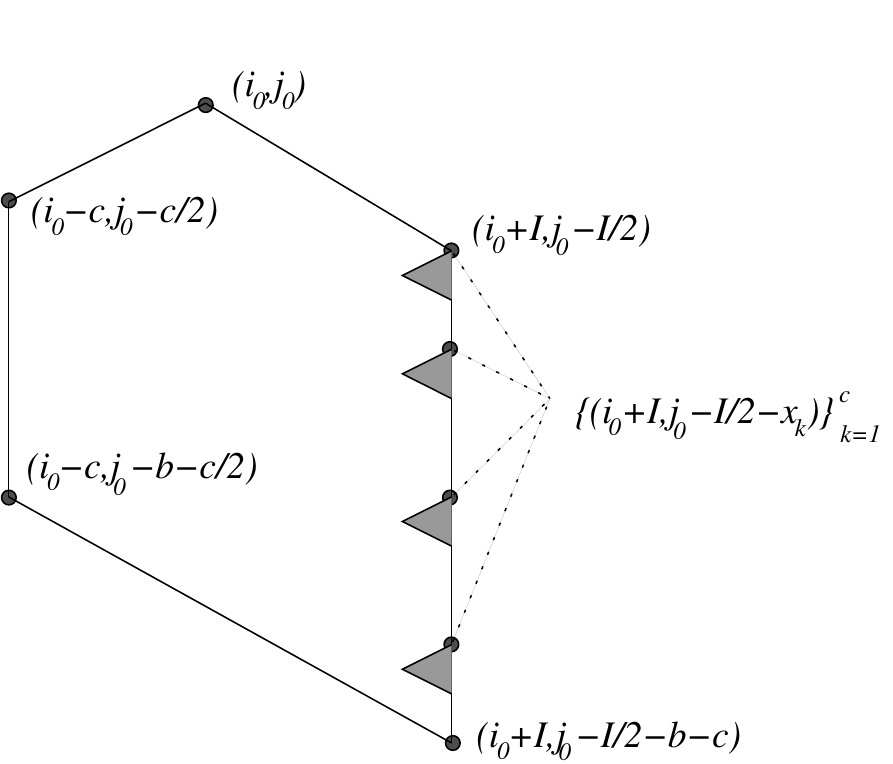}}}
\end{center}

 The total weight of lozenge tilings of this
domain is equal to a constant independent of $\{x_k\}$ times
\[
\prod_{1\le l\le c} \frac{(-1)^{x_l} \theta_p( q^{I+1},
q^{I-j_0+3i_0/2}/u_1, q^{-I+2-j_0-3i_0/2}/u_2,
q^{I+1-2j_0}/u_1u_2;q)_{x_l}} { \theta_p( q, q^{2I-j_0+3i_0/2}/u_1,
q^{2-j_0-3i_0/2}/u_2, q^{1-2j_0}/u_1u_2;q)_{x_l}}
\]
times
\[
\frac{ \prod_{1\le k<l\le c}
  q^{-x_k}\theta_p(q^{x_k-x_l},q^{x_k+x_l+I-2j_0+1}/u_1u_2)}
{ \prod_{1\le l\le c} \theta_p(q^{1-I+j_0-3i_0/2-x_l} u_1,
          q^{x_l-j_0-3i_0/2+2}/u_2;q)_{c-1}}.
\]
\end{lemma}

\begin{proof}
The problem is equivalent to computing the weight of tilings of the
domain
\begin{center}
 {\scalebox{0.8}{\includegraphics{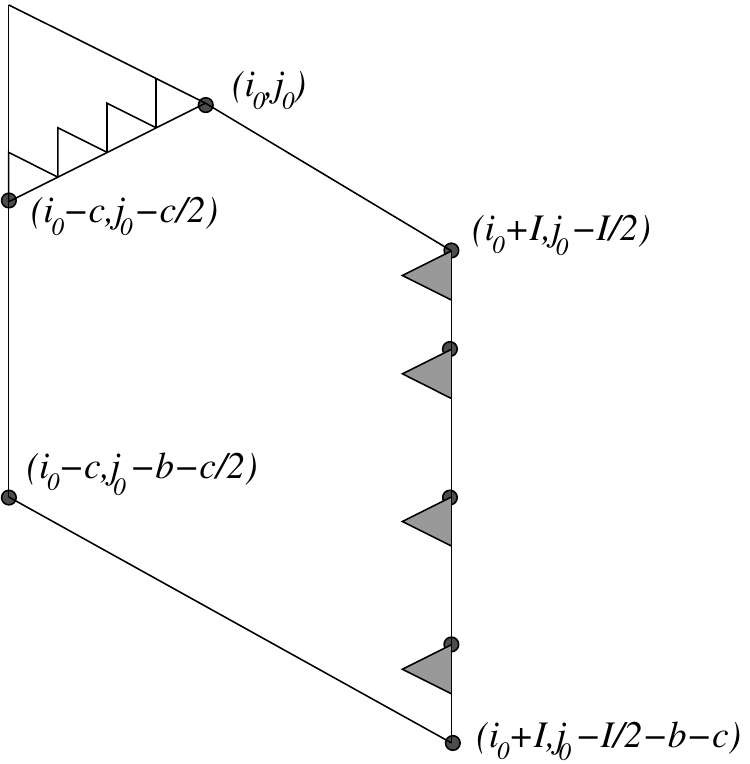}}}
\end{center}
According to Kasteleyn's theorem, cf. \cite{Ka}, we obtain
\[
\det\left[ W((i_0-k,j_0-k/2+1),(i_0+I,j_0-I/2-x_l))
\right]_{k,l=1}^c.
\]
Now, we can write
\begin{multline*}
\frac{W((i_0-k,j_0-k/2+1),(i_0+I,j_0-I/2-x_l))}
     {W((i_0-1,j_0+1/2),(i_0+I,j_0-I/2))}
\\= \frac{W((i_0-1,j_0+1/2),(i_0+I,j_0-I/2-x_l))}
     {W((i_0-1,j_0+1/2),(i_0+I,j_0-I/2))}\notag\\
\times \frac{W((i_0-k,j_0-k/2+1),(i_0+I,j_0-I/2-x_l))}
     {W((i_0-1,j_0+1/2),(i_0+I,j_0-I/2-x_l))},
\end{multline*}
noting that (for generic values of the parameters)
$W((i_0-1,j_0+1/2),(i_0+I,j_0-I/2))\ne 0$.  Since
\begin{multline*}
\frac{W((i_0-1,j_0+1/2),(i_0+I,j_0-I/2-x_l-1))}
     {W((i_0-1,j_0+1/2),(i_0+I,j_0-I/2-x_l))}\notag\\
= - \frac{\theta_p( q^{x_l} q^{I+1}, q^{x_l} q^{I-j_0+3i_0/2}/u_1,
q^{x_l} q^{-I+2-j_0-3i_0/2}/u_2, q^{x_l} q^{I+1-2j_0}/u_1u_2)} {
\theta_p( q^{x_l} q, q^{x_l} q^{2I-j_0+3i_0/2}/u_1, q^{x_l}
q^{2-j_0-3i_0/2}/u_2, q^{x_l} q^{1-2j_0}/u_1u_2)},
\end{multline*}
we find
\begin{multline*}
\frac{W((i_0-1,j_0+1/2),(i_0+I,j_0-I/2-x_l))}
     {W((i_0-1,j_0+1/2),(i_0+I,j_0-I/2))}\notag\\
= (-1)^{x_l} \frac{\theta_p( q^{I+1}, q^{I-j_0+3i_0/2}/u_1,
q^{-I+2-j_0-3i_0/2}/u_2, q^{I+1-2j_0}/u_1u_2;q)_{x_l}} { \theta_p(
q, q^{2I-j_0+3i_0/2}/u_1, q^{2-j_0-3i_0/2}/u_2,
q^{1-2j_0}/u_1u_2;q)_{x_l}}.
\end{multline*}
For the other factor, we similarly have
\begin{multline*}
\frac{W((i_0-k,j_0-k/2+1),(i_0+I,j_0-I/2-x_l))}
     {W((i_0-1,j_0+1/2),(i_0+I,j_0-I/2-x_l))}
\\ \propto \frac{ \theta_p(q^{-x_l},q^{1+x_l-2j_0+I}/u_1u_2;q)_{k-1}}
{\theta_p(q^{1-I+j_0-3i_0/2-x_l} u_1,
          q^{x_l-j_0-3i_0/2+2}/u_2;q)_{k-1}},
\end{multline*}
where we have removed factors independent of $x_l$, and observed
that the right-hand side vanishes when $0\le x_l<k$ as required.  As
a function of
\[
z_l = q^{x_l+I/2-j_0+1/2}(u_1u_2)^{-1/2},
\]
this is invariant under $z_l\mapsto 1/z_l$, and we may thus apply
Corollary 5.4 of \cite{WarnaarSO:2002} to conclude that
\begin{align*}
\det&\left[ \frac{W((i_0-k,j_0-k/2+1),(i_0+I,j_0-I/2-x_l))}
     {W((i_0-1,j_0+1/2),(i_0+I,j_0-I/2-x_l))}
\right]_{k,l=1}^c\notag\\
&\qquad\propto \frac{ \prod_{1\le k<l\le c}
  q^{-x_k}\theta_p(q^{x_k-x_l},q^{x_k+x_l+I-2j_0+1}/u_1u_2)
} {\prod_{1\le l\le c}\theta_p(q^{1-I+j_0-3i_0/2-x_l} u_1,
          q^{x_l-j_0-3i_0/2+2}/u_2;q)_{c-1}}\,.
\end{align*}
\end{proof}

\begin{lemma}
\label{Lemma_weight_right}

 Similarly, the total weight of lozenge tilings of the domain
\begin{center}
 {\scalebox{0.8}{\includegraphics{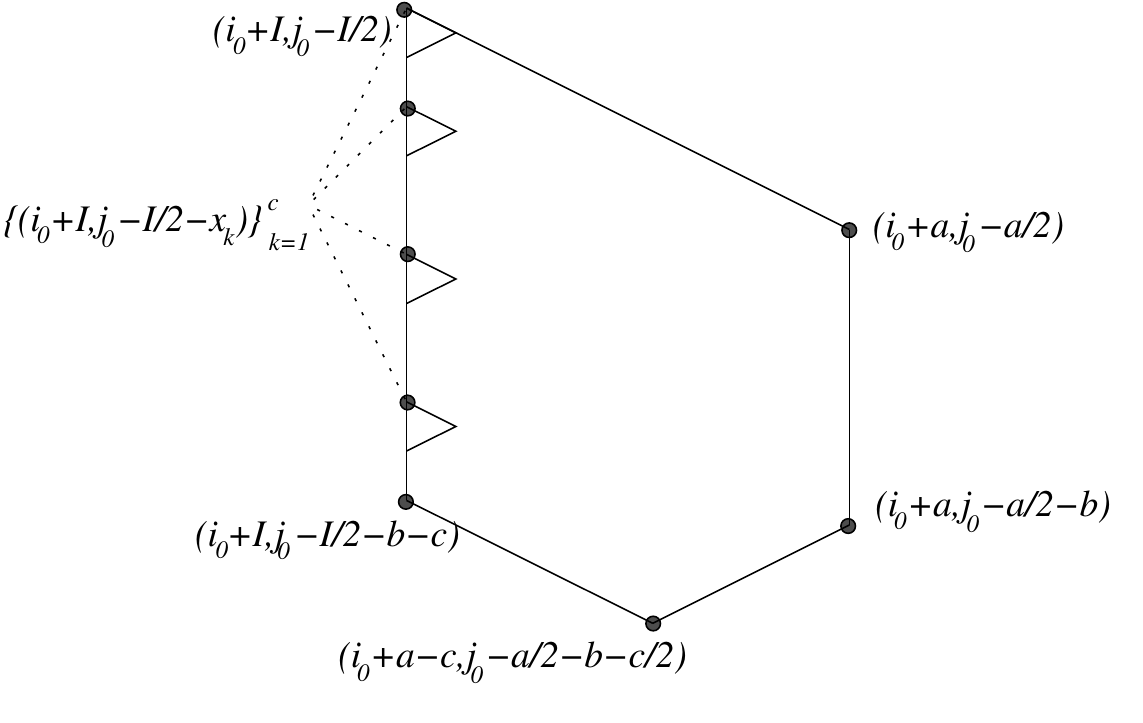}}}
\end{center}
 is
equal to a constant independent of $\{x_k\}$ times
\[
\prod_{1\le k\le c} \frac{ (-1)^{x_k} \theta_p(q^{1-b-c},
         q^{2I-j_0+2+3i_0/2}/u_1,
         q^{-a+c-j_0-3i_0/2}/u_2,
         q^{-2j_0+a+b+1}/u_1u_2;q)_{x_k}}
{ \theta_p(q^{I-a-b+1},
         q^{I+2-j_0+3i_0/2+a-c}/u_1,
         q^{-I-j_0-3i_0/2}/u_2,
         q^{I-2j_0+b+c+1}/u_1u_2;q)_{x_k}
}
\]
times
\[
\frac{ \prod_{1\le k<l\le c}
  q^{-x_k}\theta_p(q^{x_k-x_l},q^{x_k+x_l+I-2j_0+1}/u_1u_2)}
{\prod_{1\le k\le c}
\theta_p(q^{x_k+I-j_0+3i_0/2+2+a-c}/u_1,q^{j_0+3i_0/2+a+1-c-x_k}u_2;q)_{c-1}}
\]
\end{lemma}

\begin{proof}The problem is equivalent to computing the weight of
tiling of the domain
\begin{center}
 {\scalebox{0.8}{\includegraphics{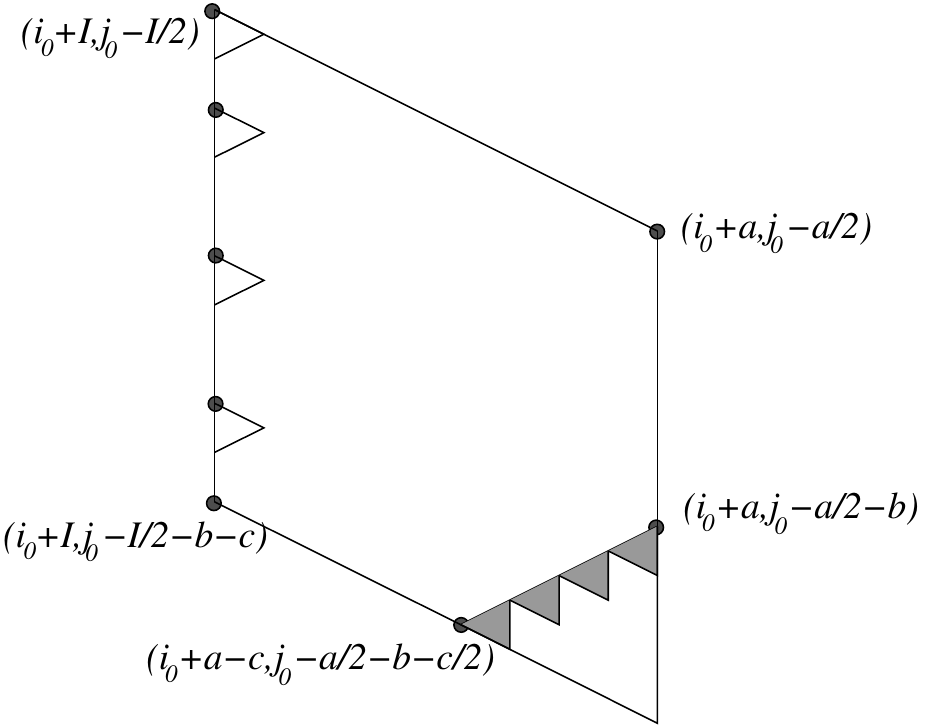}}}
\end{center}
Again, by the Kasteleyn theorem, we need to compute
\[
\det\left[ W((i_0+I,j_0-I/2-x_k),(i_0+a-c+l,j_0-a/2-b-c/2+l/2))
\right].
\]
Again,
\[
W((i_0+I,j_0-I/2-x_k),(i_0+a-c+1,j_0-a/2-b-c/2+1/2))\ne 0,
\]
allowing us to factor the weights accordingly:
\begin{gather*}
\frac{W((i_0+I,j_0-I/2-x_k),(i_0+a-c+1,j_0-a/2-b-c/2+1/2))}
     {W((i_0+I,j_0-I/2),(i_0+a-c+1,j_0-a/2-b-c/2+1/2))}\notag\\
= (-1)^{x_k} \frac{ \theta_p(q^{1-b-c},
         q^{2I-j_0+2+3i_0/2}/u_1,
         q^{-a+c-j_0-3i_0/2}/u_2,
         q^{-2j_0+a+b+1}/u_1u_2;q)_{x_k}}
{ \theta_p(q^{I-a-b+1},
         q^{I+2-j_0+3i_0/2+a-c}/u_1,
         q^{-I-j_0-3i_0/2}/u_2,
         q^{I-2j_0+b+c+1}/u_1u_2;q)_{x_k}
}
\end{gather*}
and
\begin{align*}
&\frac{W((i_0+I,j_0-I/2-x_k),(i_0+a-c+l,j_0-a/2-b-c/2+l/2))}
     {W((i_0+I,j_0-I/2-x_k),(i_0+a-c+1,j_0-a/2-b-c/2+1/2))}\notag\\
&\qquad\qquad{}\propto \frac{
\theta_p(q^{x_k-b-c+1},q^{2j_0-I-b-c-x_k}u_1u_2;q)_{l-1}} {
\theta_p(q^{I-j_0+3i_0/2+2+a-c}/u_1,q^{j_0+3i_0/2+a+1-c-x_k}u_2;q)_{l-1}}.
\end{align*}
\end{proof}

\begin{lemma}
The weight of all tilings of the hexagon
\begin{center}
 {\scalebox{0.8}{\includegraphics{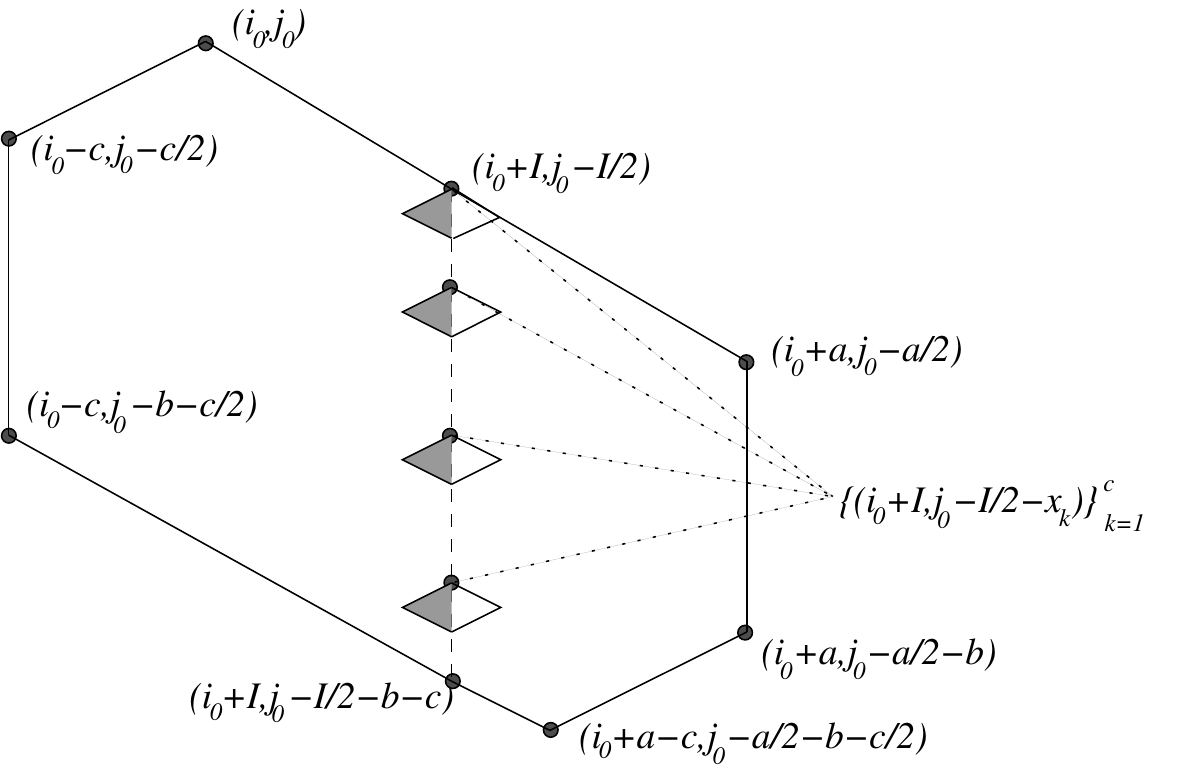}}}
\end{center}
 such that the horizonal
tiles along the line $i=i_0+I$ are as specified, is a constant
independent of $\{x_k\}$ times
\begin{multline*}
\prod_{1\le k\le c} \frac{\theta_p(q^{2x_k+I+1-2j_0}/u_1u_2)}
     {\theta_p(q^{I+1-2j_0}/u_1u_2)}\cdot \frac{q^{x_k} \theta_p( q^{I+1},q^{1-b-c})}
     {\theta_p(q,q^{I-a-b+1})} \\
\times\frac{ \theta_p( q^{I-j_0+3i_0/2}/u_1,q^{-a+c-j_0-3i_0/2}/u_2,
q^{I+1-2j_0}/u_1u_2,q^{-2j_0+a+b+1}/u_1u_2;q)_{x_k}} { \theta_p(
 q^{I+2-j_0+3i_0/2+a-c}/u_1,q^{2-j_0-3i_0/2}/u_2,
q^{1-2j_0}/u_1u_2,q^{I-2j_0+b+c+1}/u_1u_2;q)_{x_k}},
\end{multline*}
times
\begin{multline*}
 \prod_{1\le k<l\le c}
  q^{-2x_k}\theta_p(q^{x_k-x_l},q^{x_k+x_l+I-2j_0+1}/u_1u_2)^2
\\
\times
\prod_{1\le k\le c}\frac1{
\theta_p(q^{x_k+I-j_0+3i_0/2+2+a-c}/u_1,q^{I-j_0+3i_0/2+2+a-c}/u_1;q)_{c-1}}\\
\times \prod_{1\le k\le c} \frac1{\theta_p(
         q^{j_0+3i_0/2+a+1-c-x_k}u_2,q^{j_0+3i_0/2+a+1-c-x_k}u_2;q)_{c-1}}\,.
\end{multline*}
\end{lemma}

\begin{Remark}
The first ``univariate'' factor is just the product over the
variables of a special case of the weight function considered by
Spiridonov and Zhedanov in \cite{SpiridonovVP/ZhedanovAS:2000b}, and
the two sets of elliptic hypergeometric biorthogonal functions
constructed there are triangular with respect to the functions
appearing in the two determinants.  We could thus replace the second
``cross-term'' factor by a product of Vandermonde-style determinants
in which the $k$-th row consists of the $(k-1)$st biorthogonal
function evaluated at $x_1$ through $x_c$, analogously to our calculations
for the $q$-Racah case above.
\end{Remark}

\medskip
Using either lemma, we can in particular obtain a formula for the
total weight of all tilings of a hexagon.  Written in terms of plane
partitions, we obtain the following elliptic analogue of MacMahon's
identity.

\begin{theorem}\label{th:macmahon}
Let $p$, $q$, $u_1$, $u_2$, $u_3$ be generic parameters such that
$|p|<1$, $u_1u_2u_3=1$.  Then
\begin{multline*}
\sum_{\Pi\subset a\times b\times c} \prod_{(i,j,k)\in\Pi}
\frac{q^3\theta_p(q^{j+k-2i-1}u_1,q^{i+k-2j-1}u_2,q^{i+j-2k-1}u_3)}
     {\theta_p(q^{j+k-2i+1}u_1,q^{i+k-2j+1}u_2,q^{i+j-2k+1}u_3)}\notag\\
= q^{abc} \prod_{1\le i\le a,1\le j\le b,1\le k\le c}
\frac{\theta_p(q^{i+j+k-1},q^{j+k-i-1}u_1,q^{i+k-j-1}u_2,q^{i+j-k-1}u_3)}
     {\theta_p(q^{i+j+k-2},q^{j+k-i}u_1,q^{i+k-j}u_2,q^{i+j-k}u_3)}\,.
\end{multline*}
\end{theorem}

\begin{proof}
The term corresponding to a specific plane partition $\Pi$ is easily
seen by induction to be the ratio of the weight of the corresponding
tiling to the weight of the tiling associated to the empty plane
partition.  The claim follows by simplifying the corresponding
determinant of $W$ using Corollary 5.4 of \cite{WarnaarSO:2002}.
\end{proof}

\medskip

Although it is possible to arrange for the elliptic weights to be positive,
there are difficulties in the analysis.  For one thing, the algebra
required to replace orthogonal polynomials by biorthogonal functions in
constructing the kernel has not been fully developed.  A further
complication in computing the limit kernel is that, although the
biorthogonal functions do satisfy reasonably simple difference equations,
they are not eigenfunctions of any difference operators.  In addition, the
corresponding variational problem is more difficult; while we can indeed
solve the associated PDE, we have so far been unable to {\em derive} the
solution. This is why we focus on a limiting case in the present paper.

Fix $u_1u_2 = pq\zeta^2$, let $u_1,u_2=\Omega(\sqrt{p})$ as $p\to 0$, and similarly let $C(i)\sim
p^{1/2}$. In this limit, one has
\[
\lim_{p\to 0} w(i,j) = C'(i) \left(\zeta q^j -\dfrac1{\zeta q^j}\right),
\]
or in other words the $q$-Racah weight(s) discussed above. In the
notations of Sections
\ref{section_Model_and_results}-\ref{Section_comp_simulations},
\blue{ $\zeta =\kappa q^{-(S+1)/2}$.} The corresponding limit of the
elliptic MacMahon identity is
\begin{multline*}
\sum_{\Pi\subset a\times b\times c} \prod_{(i,j,k)\in\Pi}
\frac{q^{2k+1}\zeta^2-q^{i+j-1}}{q^{2k}\zeta^2-q^{i+j}} \\
= \prod_{1\le i\le a,1\le j\le b,1\le k\le c} \frac{(1-q^{i+j+k-1})(\zeta^2-q^{i+j-k-2})}
     {(1-q^{i+j+k-2})(\zeta^2-q^{i+j-k-1})},
\end{multline*}
or equivalently (performing the products over $k$ and simplifying)
\[
\sum_{\Pi\subset a\times b\times c} q^{|\Pi|} \prod_{1\le i\le a,1\le j\le b}
\frac{\zeta^2-q^{i+j-2\Pi_{ij}-2}}
     {\zeta^2-q^{i+j-c-2}}
= \prod_{1\le i\le a,1\le j\le b}
  \frac{1-q^{i+j+c-1}}
       {1-q^{i+j-1}}.
\]
This, of course, becomes the usual MacMahon identity upon taking the limit $\zeta\to\infty$, cf.
Section 7.21 in \cite{St}.

\end{document}